\documentclass[12pt]{article}

\usepackage{amssymb,amsmath,amsfonts,eurosym,geometry,ulem,graphicx,caption,color,setspace,sectsty,comment,footmisc,caption,natbib,pdflscape,subfigure,array}
\usepackage{tikz}
\usepackage{pgfplots}

\pgfplotsset{width=10cm,compat=1.9}

\usepgfplotslibrary{fillbetween}
\usepackage{hyperref}

\usepackage{color}  
\hypersetup{hidelinks,colorlinks=true,linkcolor=blue,citecolor=blue}

\usepackage{natbib}
\makeatletter
\newcommand*{\textlabel}[2]{%
  \edef\@currentlabel{#1}
  \phantomsection
  #1\label{#2}
}
\makeatother

\newcolumntype{L}[1]{>{\raggedright\let\newline\\arraybackslash\hspace{0pt}}m{#1}}
\newcolumntype{C}[1]{>{\centering\let\newline\\arraybackslash\hspace{0pt}}m{#1}}
\newcolumntype{R}[1]{>{\raggedleft\let\newline\\arraybackslash\hspace{0pt}}m{#1}}

\usepackage[T1]{fontenc}
\usepackage[utf8]{inputenc}
\usepackage{setspace}
\usepackage{babel}
\usepackage{blindtext}
\usepackage[utf8]{inputenc}
\usepackage{amssymb}
\usepackage{amsthm}
\usepackage{amsmath}
\newtheorem{theorem}{Theorem}

\newtheorem{example}{Example}
\newtheorem{proposition}{Proposition}
\newtheorem{observation}{Observation}
\newtheorem{remark}{Remark}
\newtheorem{lemma}{Lemma}
\newtheorem{claim}{Claim}
\newtheorem{manualpropositioninner}{Proposition}
\newenvironment{manualproposition}[1]{%
  \renewcommand\themanualpropositioninner{#1}%
  \manualpropositioninner
}{\endmanualpropositioninner}

\newtheorem{manuallemmainner}{Lemma}
\newenvironment{manuallemma}[1]{%
  \renewcommand\themanuallemmainner{#1}%
  \manuallemmainner
}{\endmanuallemmainner}

\usepackage{array}
\usepackage{footmisc}
\DefineFNsymbols{mySymbols}{{\ensuremath\dagger}{*}}
\setfnsymbol{mySymbols}

\newtheorem{definition}{Definition}

\title{Robust Private Supply of a Public Good}
\author{Wanchang Zhang\thanks{Department of Economics, University of California, San Diego. Email: waz024@ucsd.edu}\thanks{I have  benefited  from discussions with Songzi Du and Joel Sobel.}
}
\date{This Version: Dec 2021\\
First Draft: May 2021}

\begin{document}

\maketitle
\begin{abstract}
   We study the mechanism design problem of  selling a public good to a group of agents by a principal in the correlated private value environment. We assume the  principal  only knows the expectations of the agents' values, but does not know the joint distribution of the values.  The principal evaluates 
 a mechanism by the worst-case expected revenue  over
 joint distributions that are consistent with the known expectations.  We characterize maxmin public good mechanisms among  dominant-strategy incentive compatible and ex-post individually rational mechanisms for the two-agent case and for a special $N$-agent ($N>2$) case.  \\
\vspace{0in}
\vspace{0in}

\bigskip

\end{abstract}
\newpage
\section{Introduction}
How should a profit-maximizing private sector sell a public good to a group of agents? This question is explored by \cite{guth1986private} in the independent private value environment. However, in the real world, the private sector may not know the joint distribution of the privates values. Instead, the private sector may form an overall  expectation  about the market from, for example, the result of a market survey.   \\
\indent In this paper, we consider the problem of a monopolistic provider (referred to as the \textit{principal}) selling a public good in the correlated private value environment when the principal knows little of the information about the  values of agents. Precisely,  we assume the principal only knows the expectations of the values of the agents, but does not know the joint distribution of the values. The space of mechanisms is restricted to the set of  dominant strategy incentive compatible (DSIC) and ex-post individually rational (EPIR) mechanisms. The principal evaluates a mechanism by its lowest expected revenue across all possible joint distributions consistent with the known expectations, which will be referred to as its \textit{revenue guarantee}. The principal aims to find a mechanism, referred to as a \textit{maxmin public good mechanism},  that maximizes the revenue guarantee.\\
\indent We observe that the principal's problem  can be interpreted as a zero-sum game between the principal and adversarial nature. We will take the saddle point approach for our results. Specifically, we will construct a saddle point, which is a pair of a mechanism and a joint distribution, such that the mechanism maximizes the expected revenue over all DSIC and EPIR public good mechanisms under the joint distribution, and the joint distribution minimizes the expected revenue over all joint distributions consistent with the known expectations under the mechanism. By the properties of a saddle point, the mechanism is a maxmin public good mechanism. The joint distribution will be referred to as a \textit{worst-case joint distribution}. 
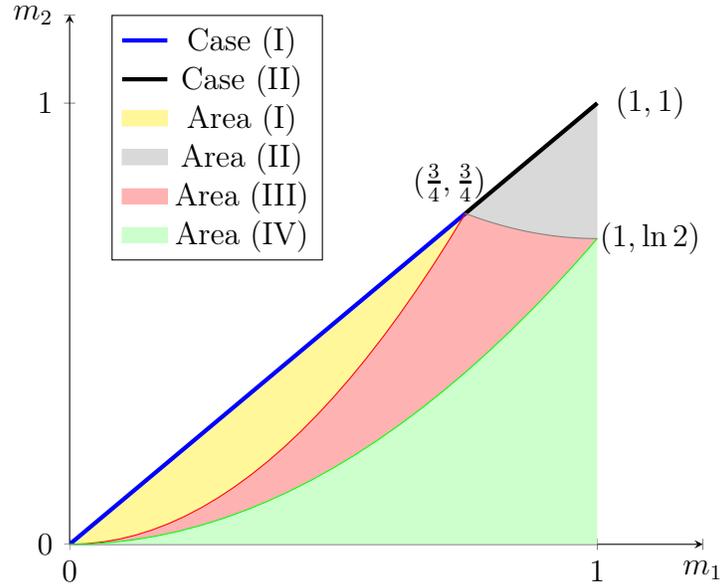
\begin{figure}
\centering
\begin{tikzpicture}
\begin{axis}[
    axis lines = left,
    xmin=0,
        xmax=1.2,
        ymin=0,
        ymax=1.2,
        xtick={0,1,1.2},
        ytick={0,1,1.2},
        xticklabels = {$0$, $1$, $m_1$},
        yticklabels = {$0$, $1$, $m_2$},
        legend style={at={(0.4,1)}}
]
\addplot[blue, ultra thick, name path=D] coordinates {(0,0) (0.75, 0.75)};
\addplot[black, ultra thick,name path=E] coordinates {(1,1) (0.75, 0.75)};
\addplot[domain=0:0.75,color=red, name path=A] {4*x^2/3};
\addplot[domain=0.75:1,color=gray, name path=C]{0.912*(x-1)^2+0.693};
\addplot[domain=0:1,color=green, name path=B] {0.693*x^2};

\path[name path=axis] (axis cs:0,0) -- (axis cs:1,0);
\addplot[area legend, yellow!50] fill between[of=A and D];
\addplot[gray!30] fill between[of=C and E];
\addplot[red!30] fill between[of=A and B, soft clip={domain=0:0.75}];
\addplot[green!20] fill between[of=B and axis];
\addplot[red!30] fill between[of=C and B, soft clip={domain=0.74:1}];
\legend{Case (I), Case (II), , , , Area (I), Area (II), Area (III), Area (IV)};
\node  at (axis cs:  1.1,  .693) {$(1,\ln{2})$};
\node  at (axis cs:  .72,  .82) {$(\frac{3}{4},\frac{3}{4})$};
\node  at (axis cs:  1.1,  1) {$(1,1)$};
\end{axis}
\end{tikzpicture}
\caption{For the symmetric case ($m_1=m_2$), we divide the expectations into two regions:  the thick blue line is \textbf{Case (I)}, and the thick black line is \textbf{Case (II)}. For the asymmetric case ($m_1>m_2$), we use (the red curve) Boundary (I), (the gray curve) Boundary (II) and (the green curve) Boundary (III) (the functional forms of these boundaries will be defined in the main result) to divide the expectations into four regions: the light yellow area is \textbf{Area (I)}, the light gray area (including Boundary (II) ) is \textbf{Area (II)}, the light red area (including Boundary (I)) is \textbf{Area (III)} and the light green area (including Boundary (III)) is \textbf{Area (IV)}. } \label{fig:M1}
\end{figure}\\
\indent Our main result offers a complete characterization of the maxmin public good mechanisms and the worst-case joint distributions when there are two agents. We normalize the support of each agent's value to the unit interval [0,1].  Without loss, we assume that  the expectation of the value of agent 1 ($m_1$) is weakly higher than that of the value of agent 2 ($m_2$). We divide the expectations into two cases and four areas (see Figure 1), and offer a characterization for each case and for each area. Let us first discuss the symmetric cases where the expectations are the same ($m_1=m_2=m$). For the maxmin public good mechanisms, the public good is provided with  a positive probability if and only if  the sum of the reported values exceeds a threshold.  In addition, the provision rule\footnote{This is the probability that the public is provided given a reported value profile. } is separable and  strictly increasing in each agent's  reported value.  Moreover, the public good is provided with a probability of 1 when each agent's reported value is 1. When the symmetric mean is low (Case (I)), the value profiles are divided into four regions and the provision rules are different across these regions. When the symmetric mean is high (Case (II)), there is only one provision rule and it  is linear in each agent's reported value. For the worst-case joint distribution,  the support coincides with the regions where the public good is likely to be provided.  When the symmetric mean is low (Case (I)), its  marginal distribution  is a combination of a uniform distribution and an equal revenue distribution; its conditional distribution is some truncated Pareto distribution with an atom on 1. When the symmetric mean is high (Case (II)), its marginal distribution is a combination of a uniform distribution on [0,1) and an atom on 1; its conditional distribution is some truncated Pareto distribution with an atom on 1.\\
\indent Then let us briefly discuss the asymmetric cases\footnote{We only discuss  maxmin public good mechanisms here, and defer the details of worst-case joint distributions to Section \ref{s6}.} where the expectation of the value of agent 1 is higher than that of the value of agent 2 ($m_1>m_2$). For the maxmin public good mechanisms, the public good is provided with  a positive probability if and only if  the sum of some \textit{weighted} reported values exceeds a threshold. Other properties are similar to those in the symmetric cases (albeit with a different functional form of the provision rule). More specifically, when the expectations are close and low (Area (I)), the maxmin public good mechanism share similar properties with that in Case (I); when the expectations are close and high (Area (II)), the maxmin public good mechanism share similar properties with that in Case (II).  When the  difference between the  expectations are moderate, the value profiles are divided into two regions and the provision rules are different across these regions.    When the difference between the expectations is high  (Area (IV)),  the maxmin public good mechanism  is a \textit{dictatorship} mechanism where the provision rule is determined by agent 1's reported value only\footnote{Equivalently, the weight on agent 2's value is 0.}. Precisely, the public good is provided with a positive probability if and only if the agent 1's reported value exceeds a threshold $w_1\in (0,1]$. In addition, the provision probability is strictly increasing in agent 1's reported value. Moreover, the public good is provided with a probability of 1 when  agent 1's reported value is 1.\\
\indent Now let us discuss the intuitions of  maxmin public good mechanisms. First, the public good will never be provided if the sum of (weighted) values falls short of some threshold for the (a)symmetric cases. This is because the principal exercises monopoly power to raise revenue. Second, the maxmin public good mechanisms require randomization to hedge against uncertainty about the joint distribution. The exact functional form of the randomization is solved by a complementary slackness condition stating that the revenue generated from a value profile is a linear function in the reported values for any value profile in the support of the worst-case joint distribution. Indeed, this condition makes adversarial nature indifferent to any feasible joint distribution whose support is the same  as that of the worst-case joint distribution. To construct the worst-case joint distribution, we first derive  a \textit{weighted virtual value} for our environment so that the expected revenue can be expressed as the inner product of the weighted virtual value and the provision rule. The worst-case joint distribution is solved by a condition requiring the weighted virtual value is 0 for any value profile in the support except for the value profile where each agent's value is 1. The intuition behind is that the principal is indifferent between providing and not providing the public good for these value profiles. Indeed, this condition guarantees that principal is indifferent to  any feasible and monotone mechanism that provides the good with a positive probability only if the value profile is inside the support and provides the good with probability of 1 when each agent's value is 1.\\
\indent These intuitions are useful to study the case when there are $N>2$ agents. For a special $N$-agent case where the symmetric expectation is high,  we characterize the maxmin public good mechanism and the worst-case joint distribution using the above conditions. We find that the provision rule in the maxmin public good mechanism is no longer separable, but admits a simple form: it is a polynomial function of the sum of the reported values. We discuss the main difficulty for general cases  when there are $N>2$ agents in Section \ref{s71}. \\
\indent In addition, we characterize maxmin \textit{deterministic}\footnote{The provision probability of the public good is either 0 or 1.} public good mechanisms for the two-agent case. There are practical concerns for studying deterministic mechanisms. To wit, deterministic mechanisms are easier to understand and more practical than randomized mechanisms in many situations, e.g.,  when the agents do not trust the randomization device. We find that when agent 2's expectation is low, the maxmin deterministic public good mechanism is a dictatorship mechanism: agent 2's reported value is irrelevant and the   public good is provided with probability of 1 if and only if agent 1's reported value is above a threshold. When agent 2's expectation is high, we characterize the class of maxmin deterministic public good mechanisms, including  a \textit{linear mechanism} where  the public good is provided with probability of 1 if and only if the sum of weighted reported value is above some threshold and a \textit{posted price mechanism} where the public good is provided with probability of 1 if and only if each agent's reported value is above some threshold. \\
\indent Lastly, we consider the problem of providing an \textit{excludable} good to $N\ge 2$ agents. In many situations, the principal can restrict access to the good for each agent, i.e., the provision rule can be agent-specific. Examples of excludable goods include movies provided by a movie theater and online courses provided by private educational institutions. We find that   in the maxmin excludable  good mechanism, the provision rule to each agent depends on the agent's reported value only. For the worst-case joint distribution, the marginal distribution for each agent is some equal revenue distribution; the values across agents are independent.\\ 
\indent The remaining of the paper proceeds as follows. Section \ref{s2} provides a literature review.  Section \ref{s3} presents the model. Section \ref{s4} presents preliminary analysis.   Section \ref{s5} characterizes the result for the two-agent symmetric cases. Section \ref{s6} characterizes the  result for the two-agent asymmetric cases. Section \ref{s7} characterizes  maxmin mechanisms for a special $N$-agent case, deterministic maxmin mechanisms for two-agent case and maxmin excludable mechanisms for general $N$-agent cases. Section \ref{s8} is a conclusion. The Appendix A contains details of the characterization of maxmin public good mechanisms and worst-case joint distributions. All proofs are in the Appendix  B and C. 

\section{Literature Review}\label{s2}
This paper is closely related to a group of papers that study robust mechanism design in different models, but all assume that the mechanism designer only knows the expectations of the values and evaluates mechanisms under the worst-case criterion. \cite{che2020distributionally}, \cite{suzdaltsev2020optimal} and \cite{koccyiugit2020distributionally} consider a model of auction design.  \cite{che2020distributionally} shows that a second price auction with an optimal random reserve distribution is a maxmin auction among a class of competitive mechanisms. \cite{suzdaltsev2020optimal} characterizes a linear version of Myerson's optimal auction as a maxmin auction among deterministic DSIC and EPIR mechanisms. \cite{koccyiugit2020distributionally} study, among others, a setting with symmetric  expectations of bidders' values and   characterize a second price auction with a random reserve as a maxmin auction among DSIC and EPIR mechanisms in which only the highest bidder(s) can be allocated. \cite{zhang2021robust} studies a model of bilateral trade and characterizes maxmin trade mechanisms among DSIC and EPIR mechanisms. The maxmin trade mechanism features a fixed commission rate and a uniformly random spread in the symmetric case. \cite{carrasco2018optimal} study the problem of selling a single good to a buyer and characterize maxmin selling mechanisms assuming the seller knows the first $N$ moments of the distribution, which includes known expectations assumption as a special case.\\
\indent In addition, this paper is related to papers that assume the mechanism designer has limited information about the joint distribution of values and evaluates the mechanism under the worst-case criterion. Precisely, a set of papers assume the mechanism designer knows the marginal distribution of values. \cite{carroll2017robustness} studies   a model of multi-dimensional screening and finds separate screening is a maxmin mechanism.  \cite{he2020correlation} study a model of auction and characterize, among other results, optimal random reserve for the second price auction under certain conditions. \cite{zhang2021correlation} also studies the model of auction and shows that, under different conditions,  a second price auction with uniformly distributed random reserve is a maxmin auction among DSIC and EPIR mechanisms for the two-bidder case and a second price auction with $Beta (\frac{1}{N-1},1)$ distributed random reserve is a maxmin auction among DSIC and EPIR mechanisms in which only the highest bidder(s) can be allocated for the $N$-bidder case.\\
\indent More broadly, this  paper is related to a strand of papers that focus on the case where the agents may have arbitrary high-order beliefs about each other unknown to the mechanism designer, e.g., \cite{bergemann2005robust}, \cite{chung2007foundations}, \cite{chen2018revisiting}, \cite{yamashita2018foundations}, \cite{bergemann2016informationally,bergemann2017first,bergemann2019revenue},  \cite{du2018robust},\cite{brooks2021optimal}, \cite{libgober2021informational}. 
\section{Model}\label{s3}
We consider an environment where a  single non-excludable public good is privately provided to $N\ge 2$ risk-neutral agents by a monopolistic provider\footnote{We will refer to the monopolistic provider as the principal.}. The cost of providing the public good is normalized to 0.  We denote by $I = \{1, 2, . . . , N\}$ the set of agents.  Each agent $i$ has private information about her
valuation for the public good, which is modeled as a random variable $v_i$ with cumulative distribution
function $F_i$\footnote{We do not make any assumption on the distributions of these random variables. It could be continuous, discrete, or any mixtures.}.  We use $f_i(v_i)$ to denote
the density of $v_i$ in the distribution $F_i$ when $F_i$ is differentiable at $v_i$; We use $Pr_i(v_i)$ to denote
the probability of $v_i$ in the distribution $F_i$ when $F_i$ has a probability mass at $v_i$.  We denote $V_i$ as the support of  $F_i$. We assume each $V_i$ is bounded. We assume the agents have a common support. Then without loss of generality,  we assume $V_i = [0,1]$ as a normalization.  The joint support of all $F_i$
is denoted as $V:= \times_{i=1}^N V_i=[0,1]^N$ with a typical value profile $v$. The joint distribution is denoted as $F$. We denote agent $i$'s opponent value profiles as $v_{-i}$, i.e., $v_{-i}\in V_{-i}:=\times_{j\neq i}V_j$.\\
\indent The principal only knows the expectation $m_i$ of the private value of each agent $i$  as well as the support,  but does not know the joint distribution of the values \footnote{That is, except for the expectations, the designer  know neither the marginal distributions nor the correlation structure.}.  Formally, let $\mathbf{m}:=(m_1,m_2,\cdots, m_N)$, and we denote by
$$\Pi_N(\mathbf{m}) = \{
\pi \in \Delta V : \forall i \in I, \int v_i \pi(v)dv = m_i\}$$
the collection of such joint distributions. We shall drop the subscript $N$ when there is no ambiguity. \\
\indent The
principal  seeks a  dominant strategy
incentive compatible (DSIC) and ex-post individually
rational (EPIR) mechanism.  A direct mechanism \footnote{It is without loss of generality to restrict attention to direct mechanisms since we focus on dominant strategy mechanisms and therefore the Revelation Principle holds.} $(q,t)$ is defined as a provision rule $q : V \to [0, 1]$ and a payment function $t : V \to \mathbb{R}^N$ where $t(v)=(t_1(v), t_2(v),\cdots, t_N(v))$. With a little abuse of notations, each agent submits a report $v_i\in V_i$ to the auctioneer. Upon receiving the report profile $v=(v_1, v_2, \cdots, v_N)$,  the public good is provided with a probability of  $q(v)\in [0,1]$ and each agent $i$ pays $t_i(v)\in \mathbb{R}$ . The set of all DSIC and EPIR mechanisms is denoted as $\mathbb{D}_N$ when there are $N$ agents. We
shall drop the dependency of $\mathbb{D}_N$  on the number of agents $N$ when there
is no confusion. Formally, the mechanism $(q,t)$ satisfies the following constraints:
\[v_iq(v)-t_i(v)\ge v_iq(v_i',v_{-i})-t_i(v_i',v_{-i})\quad \forall i, v,v_i' \tag{DSIC}\label{dsic}\]
\[v_iq(v)-t_i(v)\ge 0 \quad \forall i,v \tag{EPIR}\label{epir}\]
\indent We are interested in the principal’s expected revenue in the dominant strategy equilibrium in which each agent truthfully reports her valuation of the
public good. Then the expected revenue of a DSIC and EPIR mechanism $(q,t)$ when the joint distribution is $\pi$ is $U((q,t),\pi)\equiv \int_{v\in V}\pi(v)\sum_{i=1}^Nt_i(v) dv$. The principal evaluates each such mechanism $(q,t)$ by its worst-case expected revenue over the uncertainty of joint distributions that are consistent with the known expectations. Formally, the principal evaluates a mechanism $(q,t)$ by $GR((q,t))=\inf_{\pi\in\Pi(\mathbf{m})}\int_{v\in V}\pi(v)\sum_{i=1}^Nt_i(v) dv$, referred to as $(q,t)$'s \textit{revenue guarantee}.   The principal's goal is to find a mechanism   with the maximal revenue guarantee among DSIC and EPIR mechanisms. Formally, the principal aims to find a mechanism $(q^*,t^*)$, referred to as a \textit{maxmin public good mechanism},  that solves the following problem:
\[\sup_{(q,t)\in \mathbb{D}}GR((q,t))\tag{MPGM}\label{mpgm}\]
\indent We observe that the  problem (\ref{mpgm}) can be interpreted as a two-player sequential zero-sum
game. The two players are the principal and adversarial nature. The principal first chooses
a mechanism $(q,t)\in \mathbb{D}$. After observing the principal’s choice of
the mechanism, adversarial nature chooses a joint distribution $\pi\in \Pi(m)$. The principal’s
payoff is $U((q,t),\pi)$, and adversarial nature’s payoff is $-U((q,t),\pi)$. Now instead
of solving directly for such a subgame perfect equilibrium we can solve for a Nash equilibrium
$((q^*,t^*),\pi^*)$ of the simultaneous move version of this zero-sum game, which corresponds to a saddle point of the payoff functional $U$, i.e.,$$U((q^*,t^*), \pi) \ge U((q^*,t^*),\pi^*) \ge U((q,t),\pi^*) $$
for any $(q,t)$ and any $\pi$. The properties of a saddle point imply that the principal’s
equilibrium strategy in the simultaneous move game, $(q^*,t^*)$, is also his maxmin strategy (i.e. his
equilibrium strategy in the subgame perfect equilibrium of the sequential game). For our main result, we will explicitly construct a saddle point $((q^*,t^*),\pi^*)$ for each given expectation vector $\mathbf{m}$ for the two-agent case. 
\section{Preliminary Analysis}\label{s4}
\indent  We  use the following proposition to simplify the problem: its proof is standard
but included in the Appendix B for completeness.
\begin{proposition}[Revenue Equivalence]\label{p1}
\textit{Maxmin public good mechanisms have the following properties:\\1. $q(\cdot,v_{-i})$ is nondecreasing in $v_i$ for all $v_{-i}$}.\\2. $t_i(v_i,v_{-i})=v_iq(v_i,v_{-i})-\int_0^{v_i}q(s,v_{-i})ds$.
\end{proposition} 
Proposition \ref{p1} is a version of  \cite{myerson1981optimal}. The provision rule is monotonic and the payment rule of the maxmin public good mechanisms can be characterized by the provision rule only.\\
\indent Now  let us  consider the problem that  fixing any joint distribution $\pi$, the principal  designs an optimal mechanism $(q,t)$. For exposition, we assume $\pi$ is differentiable. We denote the density of value profile $v=(v_1,v_2,\cdots, v_N)$ as $\pi(v)$. We define $\pi_i(v_i):=\int_{[0,1]^{N-1}}\pi(v_i,v_{-i})dv_{-i}$. We define $\pi_i(v_{-i}):=\int_{0}^1\pi(v_i,v_{-i})dv_i$.  We denote the density of $v_i$ conditional on $v_{-i}$ as $\pi_i(v_i| v_{-i})$, the cumulative  distribution function of $v_i$ conditional on $v_{-i}$ as $\Pi_i(v_i|v_{-i}): = \int_{s_i\le v_i} \pi_i(s_i|v_{-i})ds_i$. We define $\Pi_i(v_i,v_{-i})\equiv \pi(v_{-i})\Pi_i(v_i|v_{-i})$. An direct implication of Proposition \ref{p1} is that the expected revenue of $(q,t)$ under the joint distribution $\pi$ is $$E[\sum_{i=1}^Nt_i(v)]=\int_vq(v)\Phi(v)dv$$
where $$\Phi(v)= 
\pi(v)\sum_{i=1}^Nv_i-\sum_{i=1}^N[\pi_i(v_{-i})-\Pi_i(v_i,v_{-i})] $$
We refer to $\Phi(v)$  as \textit{the weighted  virtual value}\footnote{Note $\Phi(v)=\pi(v)\sum_{i=1}^N(v_i-\frac{1-\Pi_i(v_i|v_{-i})}{\pi_i(v_i|v_{-i})})$ when $\pi(v)$ is not 0. Here $\phi(v):= \sum_{i=1}^N(v_i-\frac{1-\Pi_i(v_i|v_{-i})}{\pi_i(v_i|v_{-i})})$ is the virtual value in our environment, which is the sum of   the conditional virtual value of each agent. However, it turns out the weighted virtual values is more convenient for our analysis because it is well defined even for $\pi(v)=0$. Henceforth we directly work with the weighted virtual values. }  when the value profile is $v$. Thus the problem of designing an optimal mechanism given a joint distribution can be viewed as maximizing the product of the provision rule and the weighted virtual values given that the provision rule is feasible and satisfies the monotonicity condition stated in Proposition \ref{p1}.  \\
\indent Next let us consider the problem that fixing any mechanism $(q,t)$, adversarial nature chooses a joint distribution $\pi$ that minimizes the expected revenue.   We observe this is a semi-infinite dimensional linear program.  By Theorem 3.12 in \cite{anderson1987linear}, we establish the strong duality, which implies the following lemma.
\begin{lemma}\label{l1}
If $\pi$ is a best response for adversarial nature to a given mechanism $(q,t)$, then there exists some real numbers $\lambda_1,\cdots, \lambda_N,\mu$ such that 
\begin{equation} \label{eq15}
\sum_{i=1}^N\lambda_iv_i+\mu \le \sum_{i=1}^N t_i(v)\quad \forall v \in V
\end{equation}
\begin{equation} \label{eq16}
\sum_{i=1}^N\lambda_iv_i+\mu =\sum_{i=1}^N t_i(v)\quad \forall v \in supp(\pi)
\end{equation}
\end{lemma}
Here \eqref{eq15} is the feasibility constraint in the dual program; \eqref{eq16} is the complementary slackness condition, which states that the revenue from the value profile $v$ is some linear function in each agent's value if $v$ belongs to the support of the worst-case joint distribution.

\section{Symmetric Case}\label{s5}
In this section, we focus on the two-agent symmetric case, i.e., $N=2$ and $m_1=m_2:=m$. We divide $m$ into two cases and then characterize the maxmin public good mechanism and the worst-case joint distribution for each case.
\subsection{Case (I): Low Expectations ($m<\frac{3}{4}$)}
\textbf{Symmetric Maxmin Public Good Mechanism (I)}\\  Let $v=(v_1,v_2)$ be the reported value profile of the two agents. Let $r_1$ be the solution to the following equation:
\begin{equation}
    -\frac{r_1\ln{r_1}}{2}+\frac{3}{4}r_1=m
\end{equation}
Let $a:=\frac{1}{1-2\ln{r_1}}$.
Divide the value profiles into four regions as follows: $SR^I(1):=\{(v_1,v_2)|v_1+v_2\ge r_1, v_1\le r_1, v_2\le r_1\}; SR^I(2): =\{(v_1,v_2)|v_1\ge r_1, v_2\le r_1\}; SR^I(3): = \{(v_1,v_2)|v_1\le r_1, v_2\ge r_1\}; SR^I(4): =\{(v_1,v_2)|v_1\ge r_1, v_2\ge r_1\}$. The provision rule is as follows: $$ q^*(v_1,v_2)=  \left\{
\begin{array}{lll}
 \frac{a}{r_1}(v_1+v_2-r_1)     &      & { v\in SR^I(1)}\\
a\ln{v_1}+\frac{a}{r_1}v_2-a\ln{r_1}    &      & {v\in SR^I(2)}\\
a\ln{v_2}+\frac{a}{r_1}v_1-a\ln{r_1}    &      & {v\in SR^I(3)}\\
a\ln{v_1}+a\ln{v_2}+1    &      & {v\in SR^I(4)}
\\
0    &      & {otherwise}
\end{array} \right. $$
The payment rule is characterized by Proposition \ref{p1}.
\begin{figure}
\centering
\begin{tikzpicture}
\begin{axis}[
    axis lines = left,
    xmin=0,
        xmax=1.3,
        ymin=0,
        ymax=1.3,
        xtick={0,0.4,1,1.3},
        ytick={0,0.4,1,1.3},
        xticklabels = {$0$, $r_1$, $1$, $v_1$},
        yticklabels = {$0$, $r_1$, $1$, $v_2$},
        legend style={at={(1.1,1)}}
]

\path[name path=axis] (axis cs:0,0) -- (axis cs:1,0);
\path[name path=A] (axis cs:0,0.4) -- (axis cs:0.4,0);
\path[name path=B] (axis cs:0,0.4) -- (axis cs:1,0.4);
\path[name path=C] (axis cs:0,1) -- (axis cs:1,1);
\addplot[area legend, red!30] fill between[of=A and B,  soft clip={domain=0:0.4}];
\addplot[black!50] fill between[of=axis and B, soft clip={domain=0.4:1}];
\addplot[gray!50] fill between[of=B and C,  soft clip={domain=0:0.4}];
\addplot[blue!30] fill between[of=C and B, soft clip={domain=0.4:1}];
\legend{$SR^I(1)$,$SR^I(2)$, $SR^I(3)$, $SR^I(4)$};
\end{axis}
\end{tikzpicture}
\caption{Provision regions of Symmetric Maxmin Public Good Mechanism (I) } \label{fig:S1}
\end{figure}
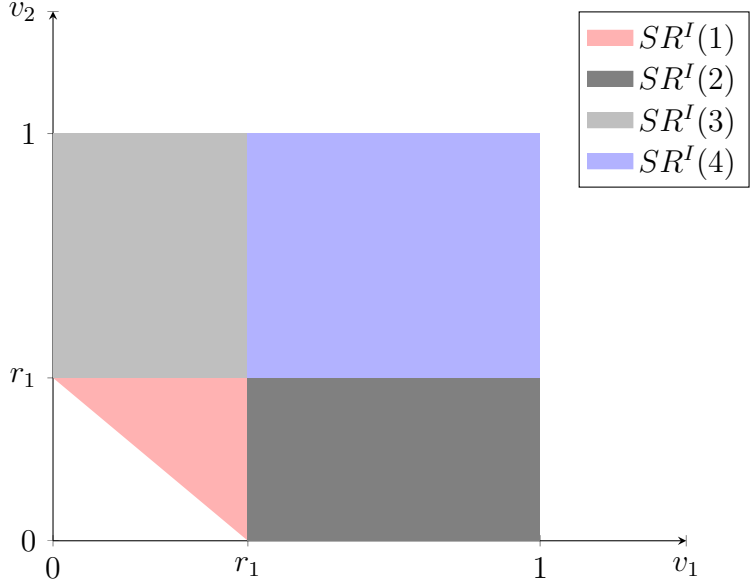\\
\textbf{Symmetric Worst-Case Joint Distribution (I)}\\
Let  $\pi^*(v_1,v_2)$ denote the density of the value profile $(v_1,v_2)$ whenever the density exists. Let $Pr^*(v_1,v_2)$ denote the probability mass of the value profile $(v_1,v_2)$ whenever there is some probability mass on $(v_1,v_2)$. Let $SV(I):=\{v|v_1+v_2\ge r_1\}$.    Symmetric Worst-Case Joint Distribution (I) has the support $SV(I)$ and  is defined as follows:
$$\pi^*(v_1,v_2)=  \left\{
\begin{array}{lll}
\frac{r_1}{(v_1+v_2)^3}     &      & { v_1+v_2\ge r_1, v_1 \neq 1 ,v_2 \neq 1}\\
\frac{r_1}{2(1+v_2)^2}    &      & {v_1=1,0\le v_2 < 1}\\
\frac{r_1}{2(v_1+1)^2}    &      & {0\le v_1 < 1,v_2=1}
\end{array} \right. $$
$$Pr^*(1,1)=\frac{r_1}{4}$$
Equivalently, \textbf{Symmetric Worst-Case Joint Distribution (I)} can be described by its marginal distributions and conditional distributions.  The marginal distributions are as follows: $\pi_1^*(v_1)=\pi_2^*(v_2)=\frac{1}{2r_1}$ for $0\le v_1,v_2 \le r_1$, $\pi_1^*(v_1)=\frac{r_1}{2v_1^2}$ and $\pi_2^*(v_2)=\frac{r_1}{2v_2^2}$ for $r_1<v_1,v_2<1$, $Pr^*_1(1)=Pr^*_2(1)=\frac{r_1}{2}$. That is, the marginal distribution of each agent is a combination of a uniform distribution and some equal revenue distribution. The conditional distributions are as follows: if $0\le v_j\le r_1$, then $\pi_i^*(v_i|v_j)=\frac{2r_1^2}{(v_i+v_j)^3}$ for $r_1-v_j\le v_i<1$ and $Pr_i^*(v_i=1|v_j)=\frac{r_1^2}{(1+v_j)^2}$; if $r_1< v_j<1$, then $\pi_i^*(v_i|v_j)=\frac{2v_j^2}{(v_i+v_j)^3}$ for $0\le v_i<1$ and $Pr_i^*(v_i=1|v_j)=\frac{v_j^2}{(1+v_j)^2}$; if $v_j=1$,  then $\pi_i^*(v_i|v_j=1)=\frac{1}{(v_i+1)^2}$ for $0\le v_i<1$ and $Pr_i^*(v_i=1|v_j=1)=\frac{1}{2}$. That is, the conditional distribution is some truncated generalized Pareto distribution with some mass on 1 (the exact distribution depends on the other agent's value).
\begin{remark}
Symmetric Worst Case Joint Distribution (I) exhibits negative correlation when $v_i\in [0,r_1)$ and positive correlation when $v_i\in [r_1,1]$. \footnote{Precisely, when $0\le v_i<v_i'<r_1$, the  distribution of $v_j$ conditional on $v_i'$ is first order stochastic dominated by the  distribution of $v_j$ conditional on $v_i$; when $r_1\le v_i<v_i'\le 1$, the  distribution of $v_j$ conditional on $v_i'$ first order stochastic dominates the  distribution of $v_j$ conditional on $v_i$.} 
\end{remark}
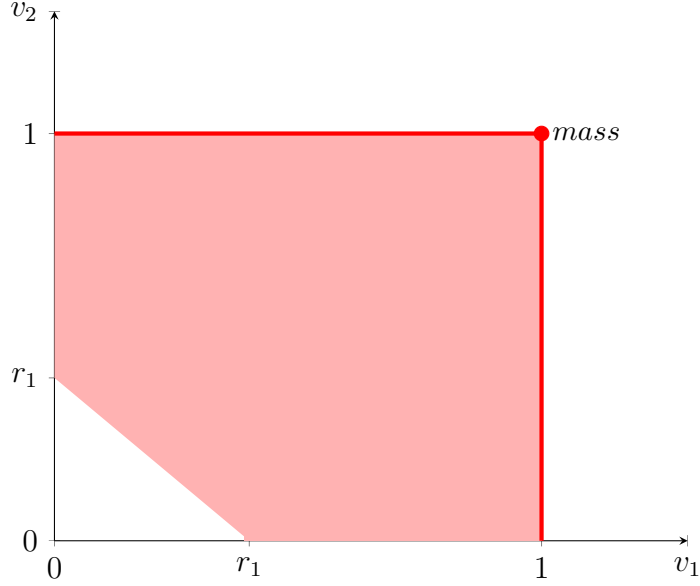
\begin{figure}
\centering
\begin{tikzpicture}
\begin{axis}[
    axis lines = left,
    xmin=0,
        xmax=1.3,
        ymin=0,
        ymax=1.3,
        xtick={0,0.4,1,1.3},
        ytick={0,0.4,1,1.3},
        xticklabels = {$0$, $r_1$, $1$, $v_1$},
        yticklabels = {$0$, $r_1$, $1$, $v_2$},
        legend style={at={(1.1,1)}}
]

\path[name path=axis] (axis cs:0,0) -- (axis cs:1,0);
\path[name path=A] (axis cs:0,0.4) -- (axis cs:0.4,0);
\path[name path=B] (axis cs:0,0.4) -- (axis cs:1,0.4);
\path[name path=C] (axis cs:0,1) -- (axis cs:1,1);
\addplot[area legend, red!30] fill between[of=A and C,  soft clip={domain=0:0.4}];
\addplot[red!30] fill between[of=axis and C, soft clip={domain=0.39:1}];
\addplot[red, ultra thick] coordinates {(0, 1) (1, 1)};
\addplot[red, ultra thick] coordinates {(1, 0) (1, 1)};
\node[black,right] at (axis cs:1,1){\small{$mass$}};
\node at (axis cs:1,1) [circle, scale=0.5, draw=red,fill=red] {};
\end{axis}
\end{tikzpicture}
\caption{Symmetric Worst-Case Joint Distribution (I)} \label{fig:S2}
\end{figure}
\begin{theorem}\label{t1}
When $m< \frac{3}{4}$, \textbf{Symmetric Maxmin Public Good Mechanism (I)} and  \textbf{Symmetric Worst-Case Joint Distribution (I)} form a Nash equilibrium. In addition, the revenue guarantee is $\frac{1}{2}\exp(W_{-1}(-2m\exp(-\frac{3}{2}))+\frac{3}{2})$.
\end{theorem}
Let us  illustrate \textbf{Symmetric Maxmin Public Good Mechanism (I)}. First,  we guess  \textlabel{(A1)}{a} that in the maxmin solution, the principle provides the  public good with positive probability if and only if the sum of the reported values  exceeds certain threshold $r_1\in (0,1)$, i.e., $v_1+v_2>r_1$ where $r_1\in (0,1)$. Second,  note that in the maxmin solution, it is without loss to assume  \textlabel{(B)}{b} that the provision probability is 1 when the value profile is $(1,1)$.\footnote{This does not affect the monotonicity constraints because the value profile $(1,1)$ is the highest type in our environment.} Third,  we conjecture that the support of the worst-case joint distribution $\pi^*$ coincides with the provision region. We assume that $q^*(r_1,r_1)=a$ for some $a\in [0,1]$. Then we divide the support into four regions: $SR^I(1), SR^I(2), SR^I(3)$ and $SR^I(4)$.  The idea is to solve for the provision probability $q^*$ for each region sequentially using the complementary slackness condition \eqref{eq16} and finally solve for $a$ by using \ref{b}. The details are deferred to the Appendix A. \\
\indent Next, we illustrate \textbf{Symmetric Worst-Case Joint Distribution (I)}. The worst-case joint distribution exhibits the property that the weighted virtual value is positive only for the highest type (1,1), zero for the other value profiles in the support and weakly negative for value profiles outside the support\footnote{By the definition of the weighted virtual values, the weighted virtual values are weakly negative for value profiles outside the support.}. Formally,  in the worst-case joint distribution, we have
\begin{equation}\label{eq53}
    \Phi(1,1)>0
\end{equation}
\begin{equation}\label{eq54}
  \Phi(v)=0  \quad \forall v_1+v_2\ge r_1\quad and \quad v \neq (1,1)  
\end{equation}
\begin{equation}\label{eq55}
  \Phi(v)\le 0  \quad \forall v_1+v_2< r_1
\end{equation} 
Now if the joint distribution satisfies \eqref{eq53}, \eqref{eq54} and \eqref{eq55}, then any feasible and monotone mechanism in which the public good is provided with some positive probability if and only if $v_1+v_2>r_1$ and the public good is provided with probability of  1 when $(v_1,v_2)=(1,1)$  is optimal for the principal. Then, the only remaining issue is whether we can construct a  joint distribution satisfying \eqref{eq53}, \eqref{eq54} and \eqref{eq55}. We give an affirmative answer by taking a constructive approach. The details of the construction are deferred to the Appendix A.  
\subsection{Case (II): High Expectations ($m\ge \frac{3}{4}$)}
\textbf{Symmetric Maxmin Public Good Mechanism (II)}\\  Let $v=(v_1,v_2)$ be the reported value profile of the two agents.
Let $r_2:=1-2\sqrt{1-m}$. Let $SR^{II}:=\{(v_1,v_2)|v_1+v_2\ge 1+r_2\}$.  The provision rule is as follows: $$ q^*(v_1,v_2)=  \left\{
\begin{array}{lll}
\frac{1}{1-r_2}(v_1+v_2-1-r_2)    &      & {v\in SR^{II}}
\\
0    &      & {otherwise}
\end{array} \right. $$
The payment rule is characterized by Proposition \ref{p1}.
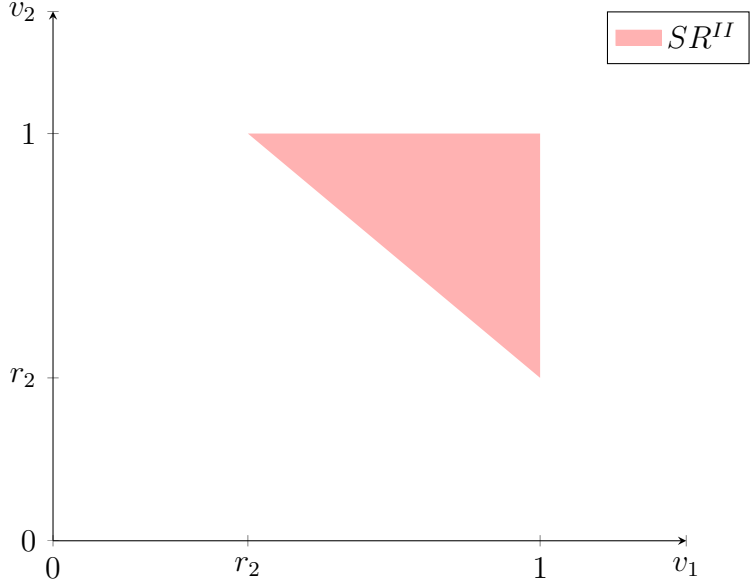
\begin{figure}
\centering
\begin{tikzpicture}
\begin{axis}[
    axis lines = left,
    xmin=0,
        xmax=1.3,
        ymin=0,
        ymax=1.3,
        xtick={0,0.4,1,1.3},
        ytick={0,0.4,1,1.3},
        xticklabels = {$0$, $r_2$, $1$, $v_1$},
        yticklabels = {$0$, $r_2$, $1$, $v_2$},
        legend style={at={(1.1,1)}}
]

\path[name path=axis] (axis cs:0,0) -- (axis cs:1,0);
\path[name path=A] (axis cs:0.4,1) -- (axis cs:1,0.4);
\path[name path=B] (axis cs:0,0.4) -- (axis cs:1,0.4);
\path[name path=C] (axis cs:0,1) -- (axis cs:1,1);
\addplot[area legend, red!30] fill between[of=A and C,  soft clip={domain=0.4:1}];

\legend{$SR^{II}$};
 
\end{axis}
\end{tikzpicture}
\caption{Provision regions of Symmetric Maxmin Public Good Mechanism (II) } \label{fig:S1}
\end{figure}\\
\textbf{Symmetric Worst-Case Joint Distribution (II)}\\
Let $\pi^*(v_1,v_2)$ denote the density of the value profile $(v_1,v_2)$ whenever the density exists. Let $Pr^*(v_1,v_2)$ denote the probability mass of the value profile $(v_1,v_2)$ whenever there is some probability mass on $(v_1,v_2)$.   Symmetric Worst-Case Joint Distribution (II) has the support $SR^{II}$ and is defined as follows:
$$\pi^*(v_1,v_2)=  \left\{
\begin{array}{lll}
\frac{(r_2+1)^2}{(v_1+v_2)^3}   &      & { v_1+v_2\ge 1+r_2, v_1 \neq 1 ,v_2 \neq 1}\\
\frac{(r_2+1)^2}{2(1+v_2)^2}    &      & {v_1=1,r_2\le v_2 < 1}\\
\frac{(r_2+1)^2}{2(v_1+1)^2}    &      & {r_2\le v_1 < 1,v_2=1}
\end{array} \right. $$
$$Pr^*(1,1)=\frac{(r_2+1)^2}{4}$$
Equivalently, \textbf{Symmetric Worst-Case Joint Distribution (II)} can be described by its marginal distributions and conditional distributions.  The marginal distributions are as follows: $\pi_1^*(v_1)=\pi_2^*(v_2)=\frac{1}{2}$ for $r_2\le v_1,v_2 < 1$, $Pr^*_1(1)=Pr^*_2(1)=\frac{1+r_2}{2}$. That is, the marginal distribution of each agent is a combination of a uniform distribution on $[r_2,1)$ and an atom on 1. The conditional distributions are as follows: if $v_j=r_2$, then $Pr_i^*(v_i=1|v_j=r_2)=1$; if $r_2< v_j< 1$, then $\pi_i^*(v_i|v_j)=\frac{2(1+r_2)^2}{(v_i+v_j)^3}$ for $1+r_2-v_j\le v_i<1$ and $Pr_i^*(v_i=1|v_j)=\frac{(1+r_2)^2}{(1+v_j)^2}$;  if $v_j=1$,  then $\pi_i^*(v_i|v_j=1)=\frac{1+r_2}{(v_i+1)^2}$ for $r_2\le v_i<1$ and $Pr_i^*(v_i=1|v_j=1)=\frac{1+r_2}{2}$. That is, the conditional distribution is (generically) some truncated generalized Pareto distribution with some mass on 1 (the exact distribution depends on the other agent's value). 
\begin{remark}
Symmetric Worst Case Joint Distribution (II) exhibits negative correlation when $v_i\in [1+r_2,1)$;  the negative  correlation breaks when $v_i=1$.
\end{remark}

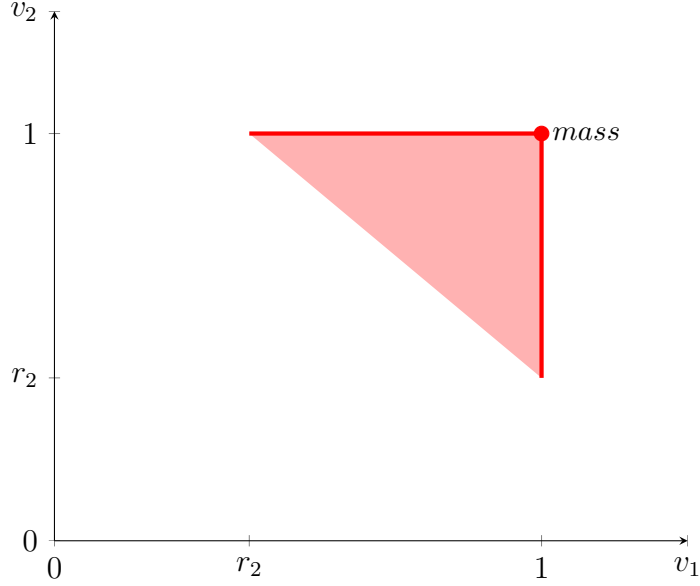
\begin{figure}
\centering
\begin{tikzpicture}
\begin{axis}[
    axis lines = left,
    xmin=0,
        xmax=1.3,
        ymin=0,
        ymax=1.3,
        xtick={0,0.4,1,1.3},
        ytick={0,0.4,1,1.3},
        xticklabels = {$0$, $r_2$, $1$, $v_1$},
        yticklabels = {$0$, $r_2$, $1$, $v_2$},
        legend style={at={(1.1,1)}}
]

\path[name path=axis] (axis cs:0,0) -- (axis cs:1,0);
\path[name path=A] (axis cs: 0.4,1) -- (axis cs:1,0.4);
\path[name path=B] (axis cs:0,0.4) -- (axis cs:1,0.4);
\path[name path=C] (axis cs:0,1) -- (axis cs:1,1);
\addplot[area legend, red!30] fill between[of=A and C,  soft clip={domain=0.4:1}];
\addplot[red, ultra thick] coordinates {(0.4, 1) (1, 1)};
\addplot[red, ultra thick] coordinates {(1, 0.4) (1, 1)};
\node[black,right] at (axis cs:1,1){\small{$mass$}};
\node at (axis cs:1,1) [circle, scale=0.5, draw=red,fill=red] {};
\end{axis}
\end{tikzpicture}
\caption{Symmetric Worst-Case Joint Distribution (II)} \label{fig:S2}
\end{figure}
\begin{theorem}\label{t2}
When $m\ge \frac{3}{4}$, \textbf{Symmetric Maxmin Public Good Mechanism (II)} and  \textbf{ Symmetric Worst-Case Joint Distribution (II)} form a Nash equilibrium. In addition, the revenue guarantee is $2(2-m-2\sqrt{1-m})$.
\end{theorem}
Let us illustrate \textbf{Symmetric Maxmin Public Good Mechanism (II)}. First, we  guess   \textlabel{(A2)}{a2} that in the maxmin solution, the principle provides the  public good with some positive probability if and only if the sum of the reported values  exceeds certain threshold $1+r_2$ in which $r_2\in [0,1]$, i.e., $v_1+v_2>1+r_2$ where $r_2\in [0,1]$. Note the threshold here is higher than that in Case (I). Intuitively, as the expectations of the values become higher, the principal exercises more monopoly power to raise revenue.  Second, similarly,  in the maxmin solution, it is without loss  to assume  \ref{b}. Third, similarly, we conjecture that the support of the worst case joint distribution $\pi^*$ is the area in which $v\in SR^{II}$. Together with (iv) in Proposition \ref{p1}, \ref{a2} and \eqref{eq16}, we obtain that for any $v\in SR^{II}$,
\begin{equation} \label{eq65}
\lambda_1v_1+\lambda_2v_2+\mu = (v_1+v_2)q^*(v)-\int_{1+r_2-v_2}^{v_1}q^*(x,v_2)dx-\int_{1+r_2-v_1}^{v_2}q^*(v_1,x)dx
\end{equation}
Then following similar procedures for solving for provision probability $q^*(v)$ when $v\in SR^I(1)$ in Case (I), we obtain  \textbf{Symmetric Maxmin Public Good Mechanism (II)}. For \textbf{Symmetric Worst-Case Joint Distribution (II)}, we have 
\begin{equation}
    \Phi(1,1)>0
\end{equation}
\begin{equation}
  \Phi(v)=0  \quad \forall v_1+v_2\ge 1+r_2\quad and \quad v \neq (1,1)  
\end{equation}
\begin{equation}
  \Phi(v)\le 0  \quad \forall v_1+v_2< 1+r_2
\end{equation}  
The construction procedure for the joint distribution is similar. Therefore we omit it.  To make sure that \textbf{Symmetric Worst-Case Joint Distribution (II)} satisfies the mean constraints, we must have \begin{equation}\label{eq66}
    \int_{r_2}^1\frac{1}{2}xdx +\frac{1+r_2}{2}=m 
\end{equation} 
Therefore we obtain that $r_2=1-2\sqrt{1-m}$ for any $m\ge \frac{3}{4}$.
\section{Asymmetric Case}\label{s6}
In this section, we focus on the two-agent asymmetric case, i.e., $N=2$ and $m_1\neq m_2$. Without loss of generality, we restrict attention to the case in which $m_1> m_2$. Let $A:=\{(m_1,m_2)|1\ge m_1>m_2\ge 0\}$ be the set of all  asymmetric expectations. We will divide the set $A$  into four areas and then characterize the maxmin public good mechanism and the worst-case joint distribution for each area.
\subsection{Area (I): Close and Low Expectations}
Let $Boundary (I): = \{(m_1,m_2)|\frac{r}{1+r}(\frac{2r-1}{(1-r)^2}\ln{r}+\frac{1}{1-r})=m_1, \frac{r}{1+r}(-\frac{1}{(1-r)^2}\ln{r}-\frac{r}{1-r})=m_2, 0<r<1\}$. It can be shown that Boundary (I) is indeed equivalent to some increasing function $B_I(m_1))$ where $0<m_1<1$. To see this, note $Z^I_1(r):= \frac{r}{1+r}(\frac{2r-1}{(1-r)^2}\ln{r}+\frac{1}{1-r})$ and $Z^I_2(r):= \frac{r}{1+r}(-\frac{1}{(1-r)^2}\ln{r}-\frac{r}{1-r})$ are both increasing w.r.t. $r$. In addition, $m_1>B_I(m_1)$. To see this, it can be shown that $Z^I(r): = Z^I_1(r)-Z^I_2(r)>0$ for $0<r<1$. Now let $Area (I): =\{(m_1,m_2)|m_2<m_1, m_2>B_{I}(m_1), \frac{3}{4}> m_1> 0\}$. We propose a pair of strategy profile as follows. \\
\textbf{Asymmetric Maxmin Public Good Mechanism (I)}\\   
Let $v=(v_1,v_2)$ be the reported value profile of the two agents. Let $s_1\in [0,1], s_2\in [0,1]$ be the  solution to the following system of  equations:
\begin{equation}\label{eq2}
   m_1=\frac{s_1s_2}{s_1+s_2}(\frac{s_1^2}{(s_1-s_2)^2}\ln{\frac{s_1}{s_2}}-\ln{s_1}+\frac{s_2}{s_2-s_1}):= H^I_1(s_1,s_2)
\end{equation}
\begin{equation}\label{eq3}
    m_2=\frac{s_1s_2}{s_1+s_2}(\frac{s_2^2}{(s_1-s_2)^2}\ln{\frac{s_2}{s_1}}-\ln{s_2}+\frac{s_1}{s_1-s_2}):=H^I_2(s_1,s_2)
\end{equation}
Let $c:=\frac{1}{1-\frac{(1-\frac{s_2}{s_1})\ln{s_1}-(1-\frac{s_1}{s_2})\ln{s_2}}{\ln{\frac{s_1}{s_2}}}}$.
Divide the value profiles into four regions as follows: $AR^I(1):=\{(v_1,v_2)|s_2v_1+s_1v_2\ge s_1s_2, v_1\le s_1, v_2\le s_2\}; AR^I(2): =\{(v_1,v_2)|v_1\ge s_1, v_2\le s_2\}; AR^I(3): = \{(v_1,v_2)|v_1\le s_1, v_2\ge s_2\}; AR^I(4): =\{(v_1,v_2)|v_1\ge s_1, v_2\ge s_2\}$. The provision rule is as follows: $$ q^*(v_1,v_2)=  \left\{
\begin{array}{lll}
 \frac{c}{\ln{\frac{s_1}{s_2}}}(\ln{(v_1+s_2-\frac{s_2}{s_1}v_1)}-\ln{(v_2+s_1-\frac{s_1}{s_2}v_2)})     &      & { v\in AR^I(1)}\\
\frac{c}{\ln{\frac{s_1}{s_2}}}((1-\frac{s_2}{s_1})\ln{v_1}-\ln{(v_2+s_1-\frac{s_1}{s_2}v_2)}+\frac{s_2}{s_1}\ln{s_1})   &      & {v\in AR^I(2)}\\
\frac{c}{\ln{\frac{s_1}{s_2}}}(\ln{(v_1+s_2-\frac{s_2}{s_1}v_1)}-(1-\frac{s_1}{s_2})\ln{v_2}-\frac{s_1}{s_2}\ln{s_2})     &      & {v\in AR^I(3)}\\
\frac{c}{\ln{\frac{s_1}{s_2}}}((1-\frac{s_2}{s_1})\ln{v_1}-(1-\frac{s_1}{s_2})\ln{v_2})+1   &      & {v\in AR^I(4)}
\\
0    &      & {otherwise}
\end{array} \right. $$
The payment rule is characterized by Proposition \ref{p1}.
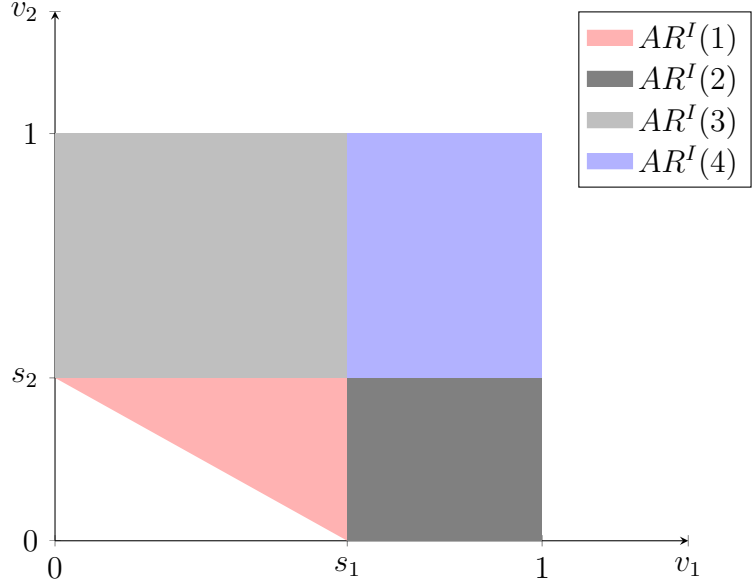
\begin{figure}
\centering
\begin{tikzpicture}
\begin{axis}[
    axis lines = left,
    xmin=0,
        xmax=1.3,
        ymin=0,
        ymax=1.3,
        xtick={0,0.6,1,1.3},
        ytick={0,0.4,1,1.3},
        xticklabels = {$0$, $s_1$, $1$, $v_1$},
        yticklabels = {$0$, $s_2$, $1$, $v_2$},
        legend style={at={(1.1,1)}}
]

\path[name path=axis] (axis cs:0,0) -- (axis cs:1,0);
\path[name path=A] (axis cs:0,0.4) -- (axis cs:0.6,0);
\path[name path=B] (axis cs:0,0.4) -- (axis cs:1,0.4);
\path[name path=C] (axis cs:0,1) -- (axis cs:1,1);
\addplot[area legend, red!30] fill between[of=A and B,  soft clip={domain=0:0.6}];
\addplot[black!50] fill between[of=axis and B, soft clip={domain=0.6:1}];
\addplot[gray!50] fill between[of=B and C,  soft clip={domain=0:0.6}];
\addplot[blue!30] fill between[of=C and B, soft clip={domain=0.6:1}];
\legend{$AR^I(1)$,$AR^I(2)$, $AR^I(3)$, $AR^I(4)$};
 
\end{axis}
\end{tikzpicture}
\caption{Provision regions of Asymmetric Maxmin Public Good Mechanism (I) } \label{fig:S1}
\end{figure}\\
\textbf{Asymmetric Worst-Case Joint Distribution (I)}\\
Let  $\pi^*(v_1,v_2)$ denote the density of the value profile $(v_1,v_2)$ whenever the density exists. Let $Pr^*(v_1,v_2)$ denote the probability mass of the value profile $(v_1,v_2)$ whenever there is some probability mass on $(v_1,v_2)$. Let $AV(I):=\{v|s_2v_1+s_1v_2\ge s_1s_2\}$.     Asymmetric Worst-Case Joint Distribution (I) has the support $AV(I)$ and  is defined as follows:
$$\pi^*(v_1,v_2)=  \left\{
\begin{array}{lll}
\frac{2s_1s_2}{(s_1+s_2)(v_1+v_2)^3}     &      & { s_2v_1+s_1v_2\ge s_1s_2, v_1 \neq 1 ,v_2 \neq 1}\\
\frac{s_1s_2}{(s_1+s_2)(1+v_2)^2}  &      & {v_1=1,0\le v_2 < 1}\\
\frac{s_1s_2}{(s_1+s_2)(1+v_1)^2}   &      & {0\le v_1 < 1,v_2=1}
\end{array} \right. $$
$$Pr^*(1,1)=\frac{s_1s_2}{2(s_1+s_2)}$$
Equivalently, \textbf{Asymmetric Worst-Case Joint Distribution (I)} can be described by its marginal distributions and conditional distributions.  The marginal distributions are as follows: $\pi_i^*(v_i)=\frac{s_1s_2}{(s_1+s_2)(\frac{s_j}{s_i}(s_i-v_i)+v_i)^2}$ for $0\le v_i \le s_i$, $\pi_i^*(v_i)=\frac{s_1s_2}{(s_1+s_2)(v_i)^2}$ for $s_i<v_i<1$, $Pr^*_i(1)=\frac{s_1s_2}{(s_1+s_2)}$. That is, the marginal distribution of each agent is a combination of some generalized Pareto distribution and some equal revenue distribution. The conditional distributions are as follows: if $0\le v_j\le s_j$, then $\pi_i^*(v_i|v_j)=\frac{2(\frac{s_i}{s_j}(s_j-v_j)+v_j)^2}{(v_i+v_j)^3}$ for $s_i-\frac{s_i}{s_j}v_j\le v_i<1$ and $Pr_i^*(v_i=1|v_j)=\frac{(\frac{s_i}{s_j}(s_j-v_j)+v_j)^2}{(1+v_j)^2}$; if $s_j< v_j<1$, then $\pi_i^*(v_i|v_j)=\frac{2v_j^2}{(v_i+v_j)^3}$ for $0\le v_i<1$ and $Pr_i^*(v_i=1|v_j)=\frac{v_j^2}{(1+v_j)^2}$; if $v_j=1$,  then $\pi_i^*(v_i|v_j=1)=\frac{1}{(v_i+1)^2}$ for $0\le v_i<1$ and $Pr_i^*(v_i=1|v_j=1)=\frac{1}{2}$. That is, the conditional distribution is some truncated generalized Pareto distribution with some mass on 1 (the exact distribution depends on the other agent's value). 
\begin{remark}
When $(m_1,m_2)\in Boundary (I)$, $s_2=1$ and $s_1\in (0,1)$.
\end{remark}
\begin{figure}
\centering
\begin{tikzpicture}
\begin{axis}[
    axis lines = left,
    xmin=0,
        xmax=1.3,
        ymin=0,
        ymax=1.3,
        xtick={0,0.6,1,1.3},
        ytick={0,0.4,1,1.3},
        xticklabels = {$0$, $s_1$, $1$, $v_1$},
        yticklabels = {$0$, $s_2$, $1$, $v_2$},
        legend style={at={(1.1,1)}}
]

\path[name path=axis] (axis cs:0,0) -- (axis cs:1,0);
\path[name path=A] (axis cs:0,0.4) -- (axis cs:0.6,0);
\path[name path=B] (axis cs:0,0.4) -- (axis cs:1,0.4);
\path[name path=C] (axis cs:0,1) -- (axis cs:1,1);
\addplot[area legend, red!30] fill between[of=A and C,  soft clip={domain=0:0.6}];
\addplot[red!30] fill between[of=axis and C, soft clip={domain=0.59:1}];
\addplot[red, ultra thick] coordinates {(0, 1) (1, 1)};
\addplot[red, ultra thick] coordinates {(1, 0) (1, 1)};
\node[black,right] at (axis cs:1,1){\small{$mass$}};
\node at (axis cs:1,1) [circle, scale=0.5, draw=red,fill=red] {};
\end{axis}
\end{tikzpicture}
\caption{Asymmetric Worst-Case Joint Distribution (I)} \label{fig:S2}
\end{figure}
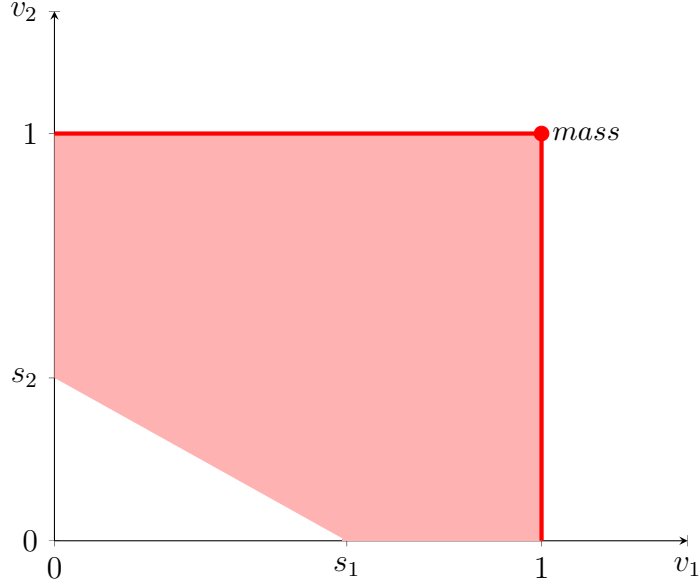
\begin{lemma}\label{l2}
For any given $(m_1,m_2)\in Area (I)$, there exists a  solution $s_1,s_2$ to the system of equations \eqref{eq2} and \eqref{eq3}.
\end{lemma}
\begin{theorem}\label{t3}
When $(m_1,m_2)\in Area (I)$, \textbf{Asymmetric Maxmin Public Good  Mechanism (I)} and  \textbf{Asymmetric Worst-Case Joint Distribution (I)} form a Nash equilibrium. The revenue guarantee is $\frac{s_1s_2}{s_1+s_2}$.
\end{theorem}
Let us illustrate \textbf{Asymmetric Maxmin Public Good Mechanism (I)}. First, we  guess that  \textlabel{(A3)}{a3} that in the maxmin solution, the principal provides the  public good with positive probability if and only if the sum of  weighted reported values  exceeds certain threshold, i.e., $s_2v_1+ s_1v_2>s_1s_2$ where $s_1, s_2 \in (0,1]$ and $s_1\neq s_2$. Second, similarly,  in the maxmin solution, it is without loss to assume  \ref{b}. Third, similarly,  we conjecture that the support of the worst case joint distribution $\pi^*$ is the area in which $v\in AV(I)$. We assume that $q^*(r_1,r_1)=c$ for some $c\in [0,1]$. Then we divide the support into four regions: $AR^I(1), AR^I(2), AR^I(3)$ and $AR^I(4)$. Similar to Case (I), the idea is to solve for the provision probability $q^*$ for each region sequentially using the complentary slackness condition \eqref{eq16} and finally solve for $c$ by using \ref{b}. The details are deferred to the Appendix A. For \textbf{Asymmetric Worst-Case Joint Distribution (I)}, we have 
\begin{equation}
    \Phi(1,1)>0
\end{equation}
\begin{equation}
  \Phi(v)=0  \quad \forall s_2v_1+s_1v_2\ge s_1s_2\quad and \quad v \neq (1,1)  
\end{equation}
\begin{equation}
  \Phi(v)\le 0  \quad \forall s_2v_1+s_1v_2< s_1s_2
\end{equation}   
The construction procedure for the joint distribution is similar. Therefore we omit it.  To make sure that \textbf{Asymmetric Worst-Case Joint Distribution (I)} satisfies the mean constraints, we have a system of two equations \eqref{eq2} and \eqref{eq3}. Lemma \ref{l2} states that a solution exists.
\subsection{Area (II): Close and High Expectations}
Let $Boundary (II): \{(m_1,m_2)|\frac{(1+r)(1-r^2)}{2r^2}\ln{\frac{1}{1+r}}+\frac{1+r^2}{2r}=m_1, \frac{1+r}{2r^2}\ln{(1+r)}-\frac{1}{2r}+\frac{1}{2}=m_2, 0<r\le 1 \}$. It can be shown that $Boundary (II)$ is indeed equivalent to some decreasing function $B_{II}(m_1)$ where $0<m_1\le 1$. To see this, note $Z^{II}_1(r):= \frac{(1+r)(1-r^2)}{2r^2}\ln{\frac{1}{1+r}}+\frac{1-r}{2r}+\frac{1+r}{2}$ is increasing w.r.t $r$ and $Z^{II}_2(r):= \frac{1+r}{2r^2}\ln{(1+r)}-\frac{1}{2r}+\frac{1}{2}$ is decreasing w.r.t. $r$. In addition, $m_1>B_{II}(m_1)$. To see this, it can be shown that $Z^{II}(r): = Z^{II}_1(r)-Z^{II}_2(r)>0$ for $0<r<1$. Now let $Area (II): =\{(m_1,m_2)|m_2<m_1, m_2\ge B_{II}(m_1), 1\ge m_1> \frac{3}{4}\}$. We propose a pair of strategy profile as follows. \\
\textbf{Asymmetric Maxmin Public Good Mechanism (II)}\\   
Let $v=(v_1,v_2)$ be the reported value profile of the two agents. Let $t_1, t_2$ be the unique solution to the following equation:
\begin{equation}\label{eq4}
   m_1=\frac{(1+t_1)(1+t_2)(1-t_1)^2}{2(t_1-t_2)^2}\ln{\frac{1+t_2}{1+t_1}}+\frac{(1-t_1t_2)(1-t_1)}{2(t_1-t_2)}+\frac{1+t_1}{2}:= H^{II}_1(t_1,t_2)
\end{equation}
\begin{equation}\label{eq5}
   m_2=\frac{(1+t_1)(1+t_2)(1-t_2)^2}{2(t_1-t_2)^2}\ln{\frac{1+t_1}{1+t_2}}+\frac{(1-t_1t_2)(1-t_2)}{2(t_2-t_1)}+\frac{1+t_2}{2}:=H^{II}_2(t_1,t_2)
\end{equation}
Let $AR^{II}: = \{(v_1,v_2)|(1-t_2)v_1+(1-t_1)v_2\ge 1-t_1t_2\}$. The provision rule is as follows: $$ q^*(v_1,v_2)=  \left\{
\begin{array}{lll}
\frac{1}{\ln{\frac{1+t_2}{1+t_1}}}(\ln{(v_1+\frac{t_2-1}{1-t_1}(v_1-t_1)+1)}-\ln{(v_2+\frac{1-t_1}{t_2-1}(v_2-1)+t_1)})    &      & {v\in AR^{II}}
\\
0    &      & {otherwise}
\end{array} \right. $$
The payment rule is characterized by Proposition \ref{p1}.
\begin{figure}
\centering
\begin{tikzpicture}
\begin{axis}[
    axis lines = left,
    xmin=0,
        xmax=1.3,
        ymin=0,
        ymax=1.3,
        xtick={0,0.6,1,1.3},
        ytick={0,0.4,1,1.3},
        xticklabels = {$0$, $t_1$, $1$, $v_1$},
        yticklabels = {$0$, $t_2$, $1$, $v_2$},
        legend style={at={(1.1,1)}}
]

\path[name path=axis] (axis cs:0,0) -- (axis cs:1,0);
\path[name path=A] (axis cs:0.6,1) -- (axis cs:1,0.4);
\path[name path=B] (axis cs:0,0.4) -- (axis cs:1,0.4);
\path[name path=C] (axis cs:0,1) -- (axis cs:1,1);
\addplot[area legend, red!30] fill between[of=A and C,  soft clip={domain=0.6:1}];
\legend{$AR^{II}$};
 
\end{axis}
\end{tikzpicture}
\caption{Provision regions of Asymmetric Maxmin Public Good Mechanism (II) } \label{fig:S1}
\end{figure}
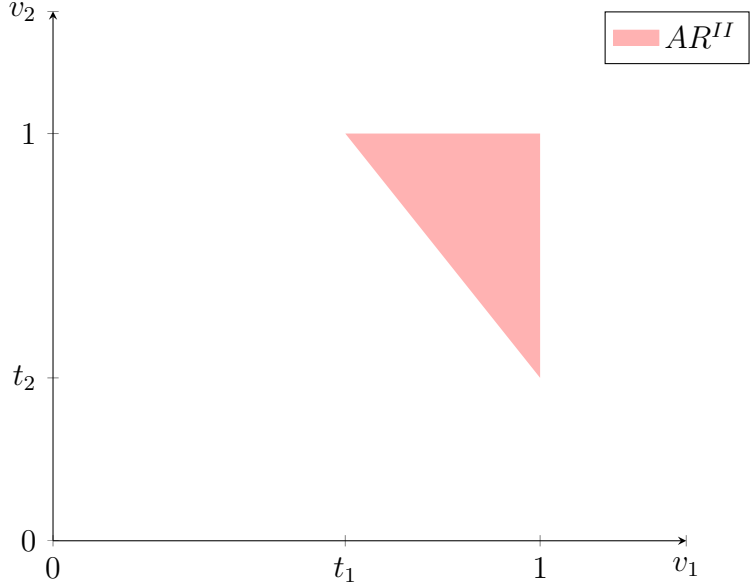\\
\textbf{Asymmetric Worst-Case Joint Distribution (II)}\\
Let  $\pi^*(v_1,v_2)$ denote the density of the value profile $(v_1,v_2)$ whenever the density exists. Let $Pr^*(v_1,v_2)$ denote the probability mass of the value profile $(v_1,v_2)$ whenever there is some probability mass on $(v_1,v_2)$.  Asymmetric Worst-Case Joint Distribution (II) has the support $SR^{II}$ and  is defined as follows:
$$\pi^*(v_1,v_2)=  \left\{
\begin{array}{lll}
\frac{(1+t_1)(1+t_2)}{(v_1+v_2)^3}     &      & { (1-t_2)v_1+(1-t_1)v_2\ge 1-t_1t_2, v_1 \neq 1 ,v_2 \neq 1}\\
\frac{(1+t_1)(1+t_2)}{2(1+v_2)^2}  &      & {v_1=1,t_2\le v_2 < 1}\\
\frac{(1+t_1)(1+t_2)}{2(1+v_1)^2}   &      & {t_1\le v_1 < 1,v_2=1}
\end{array} \right. $$
$$Pr^*(1,1)= \frac{(1+t_1)(1+t_2)}{4}$$
Equivalently, \textbf{Asymmetric Worst-Case Joint Distribution (II)} can be described by its marginal distributions and conditional distributions.  The marginal distributions are as follows: $\pi_i^*(v_i)=\frac{(1+t_1)(1+t_2)}{2(\frac{t_j-t_i}{1-t_i}v_i+\frac{1-t_1t_2}{1-t_i})^2}$ for $t_i\le v_i < 1$, $Pr^*_i(1)=\frac{1+t_i}{2}$. That is, the marginal distribution of each agent is a combination of a uniform distribution on $[t_i,1)$ and an atom on 1. The conditional distributions are as follows: if $v_j=t_j$, then $Pr_i^*(v_i=1|v_j=t_j)=1$; if $t_j< v_j< 1$, then $\pi_i^*(v_i|v_j)=\frac{2(\frac{t_i-t_j}{1-t_j}v_j+\frac{1-t_1t_2}{1-t_j})^2}{(v_i+v_j)^3}$ for $\frac{1-t_1t_2}{1-t_j}-\frac{1-t_i}{1-t_j}v_j\le v_i<1$ and $Pr_i^*(v_i=1|v_j)=\frac{(\frac{t_i-t_j}{1-t_j}v_j+\frac{1-t_1t_2}{1-t_j})^2}{(1+v_j)^2}$;  if $v_j=1$,  then $\pi_i^*(v_i|v_j=1)=\frac{1+t_i}{(v_i+1)^2}$ for $t_i\le v_i<1$ and $Pr_i^*(v_i=1|v_j=1)=\frac{1+t_i}{2}$. That is, the conditional distribution is (generically) some truncated generalized Pareto distribution with some mass on 1 (the exact distribution depends on the other agent's value). 
\begin{remark}
When $(m_1,m_2)\in Boundary (II)$, $t_2=0$ and $t_1\in (0,1]$.
\end{remark}
\begin{figure}
\centering
\begin{tikzpicture}
\begin{axis}[
    axis lines = left,
    xmin=0,
        xmax=1.3,
        ymin=0,
        ymax=1.3,
        xtick={0,0.6,1,1.3},
        ytick={0,0.4,1,1.3},
        xticklabels = {$0$, $t_1$, $1$, $v_1$},
        yticklabels = {$0$, $t_2$, $1$, $v_2$},
        legend style={at={(1.1,1)}}
]

\path[name path=axis] (axis cs:0,0) -- (axis cs:1,0);
\path[name path=A] (axis cs: 0.6,1) -- (axis cs:1,0.4);
\path[name path=B] (axis cs:0,0.4) -- (axis cs:1,0.4);
\path[name path=C] (axis cs:0,1) -- (axis cs:1,1);
\addplot[area legend, red!30] fill between[of=A and C,  soft clip={domain=0.6:1}];
\addplot[red, ultra thick] coordinates {(0.6, 1) (1, 1)};
\addplot[red, ultra thick] coordinates {(1, 0.4) (1, 1)};
\node[black,right] at (axis cs:1,1){\small{$mass$}};
\node at (axis cs:1,1) [circle, scale=0.5, draw=red,fill=red] {};
\end{axis}
\end{tikzpicture}
\caption{Asymmetric Worst-Case Joint Distribution (II)} \label{fig:S2}
\end{figure}
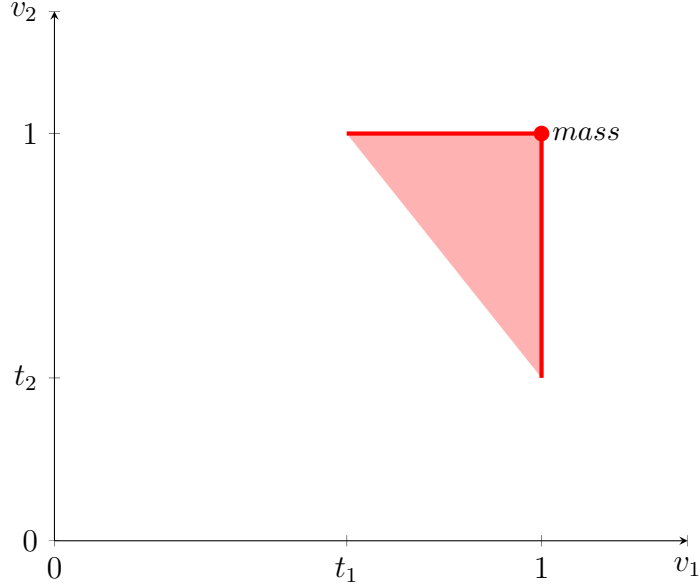
\begin{lemma}\label{l3}
For any given $(m_1,m_2)\in Area (II)$, there exists a solution $t_1,t_2$ to  to the system of equations \eqref{eq4} and \eqref{eq5}.
\end{lemma}
\begin{theorem}\label{t4}
When $(m_1,m_2)\in Area (II)$, \textbf{Asymmetric Maxmin Public Good  Mechanism (II)} and  \textbf{Asymmetric Worst-Case Joint Distribution (II)} form a Nash equilibrium. The revenue guarantee is $\frac{(1+t_1)(1+t_2)}{2}$.
\end{theorem}
Let us  illustrate \textbf{Asymmetric Maxmin Public Good Mechanism (II)}. We  guess   \textlabel{(A4)}{a4} that in the maxmin solution, the principal provides the  public good with positive probability if and only if the sum of weighted reported values exceeds certain threshold, i.e., $(1-t_2)v_1+(1-t_1)v_2>1-t_1t_2$ where $t_1,t_2\in [0,1)$. Second, similarly,  in the maxmin solution, it is without loss  to assume  \ref{b}. Third, similarly, we conjecture that the support of the worst case joint distribution $\pi^*$ is the area in which $v\in AR^{II}$. Together with (iv) in Proposition \ref{p1}, \ref{a4} and \eqref{eq16}, we obtain that for any $v\in AR^{II}$,
\begin{equation} \label{eq103}
\lambda_1v_1+\lambda_2v_2+\mu = (v_1+v_2)q^*(v)-\int_{\frac{1-t_1t_2}{1-t_2}-\frac{1-t_1}{1-t_2}v_2}^{v_1}q^*(x,v_2)dx-\int_{\frac{1-t_1t_2}{1-t_1}-\frac{1-t_2}{1-t_1}v_1}^{v_2}q^*(v_1,x)dx
\end{equation}
Then following similar procedures for solving for provision probability $q^*(v)$ when $v\in AR^I(1)$ in Area (I), we obtain  \textbf{Asymmetric Maxmin Public Good Mechanism (II)}.  For \textbf{Asymmetric Worst-Case Joint Distribution (II)}, we have 
\begin{equation}
    \Phi(1,1)>0
\end{equation}
\begin{equation}
  \Phi(v)=0  \quad \forall (1-t_2)v_1+(1-t_1)v_2\ge 1-t_1t_2\quad and \quad v \neq (1,1)  
\end{equation}
\begin{equation}
  \Phi(v)\le 0  \quad \forall (1-t_2)v_1+(1-t_1)v_2< 1-t_1t_2
\end{equation}   The construction procedure for the joint distribution is similar. Therefore we omit it.  To make sure that \textbf{Asymmetric Worst-Case Joint Distribution (II)} satisfies the mean constraints, we have a system of two equations \eqref{eq4} and \eqref{eq5}. Lemma \ref{l3} states the solution exists and is unique.
\subsection{Area (III): Moderate-Distance Expectations}
Let $Boundary (III):= \{(m_1,m_2)|r(1-\ln{r})=m_1,  r\ln{\frac{r+1}{r}}=m_2, 0<r\le 1\}$. It can be shown that $Boundary (III)$ is indeed equivalent to some increasing function $B_{III}(m_1)$ where $0<m_1\le 1$. To see this, note $Z^{III}_1(r):= r(1-\ln{r})$ and $Z^{III}_2(r):= r\ln{\frac{r+1}{r}}$ are both increasing w.r.t. $r$. In addition, $m_1>B_{III}(m_1)$. To see this, it can be shown that $Z^{III}(r): = Z^{III}_1(r)-Z^{III}_2(r)>0$ for $0<r\le 1$. Now let $Area (III): =\{(m_1,m_2)|m_2\le B_{I}(m_1), m_2< B_{II}(m_1), m_2> B_{III}(m_1),  1\ge m_1> 0\}$. We propose a pair of strategy profile as follows. \\
\textbf{Asymmetric Maxmin Public Good Mechanism (III)}\\   
Let $v=(v_1,v_2)$ be the reported value profile of the two agents. Let $u_1, u_2$ be the unique solution to the following equation:
\begin{equation}\label{eq6}
   m_1=\frac{u_1(u_2+1)}{u_1+1}(\frac{(u_2-u_1)^2}{(u_2-u_1+1)^2}\ln{\frac{u_1}{1+u_2}}-\ln{u_1}+1-\frac{u_2-u_1}{(1+u_2)(1+u_2-u_1)}):= H^{III}_1(u_1,u_2)
\end{equation}
\begin{equation}\label{eq7}
   m_2=\frac{u_1(u_2+1)}{u_1+1}(\frac{1}{(u_2-u_1+1)^2}\ln{\frac{1+u_2}{u_1}}+\frac{1}{1+u_2}-\frac{1}{(1+u_2)(1+u_2-u_1)}):=H^{III}_2(u_1,u_2)
\end{equation}
Let $d:=\frac{\ln{\frac{u_1}{1+u2}}}{\frac{1}{u_1-u_2}\ln{u_1}-\ln{(1+u_2)}}$.
Divide the value profiles into two regions as follows: $AR^{III}(1):=\{(v_1,v_2)|v_1+(u_1-u_2)v_2\ge u_1, v_1\le u_1, v_2\le 1\}; AR^{III}(2): =\{(v_1,v_2)|v_1\ge u_1,  v_2\le 1\}$. The provision rule is as follows: $$ q^*(v_1,v_2)=  \left\{
\begin{array}{lll}
 \frac{d}{\ln{\frac{u_1}{1+u_2}}}(\ln{(v_1+\frac{1}{u_2-u_1}(v_1-u_1))}-\ln{(v_2+(u_2-u_1)v_2+u_1)})    &      & { v\in AR^{III}(1)}\\
\frac{d}{\ln{\frac{u_1}{1+u_2}}}((1+\frac{1}{u_2-u_1})\ln{v_1}-\ln{(v_2+(u_2-u_1)v_2+u_1)}-\frac{1}{u_2-u_1}\ln{u_1})   &      & {v\in AR^{III}(2)}\\
0    &      & {otherwise}
\end{array} \right. $$
The payment rule is characterized by Proposition \ref{p1}.
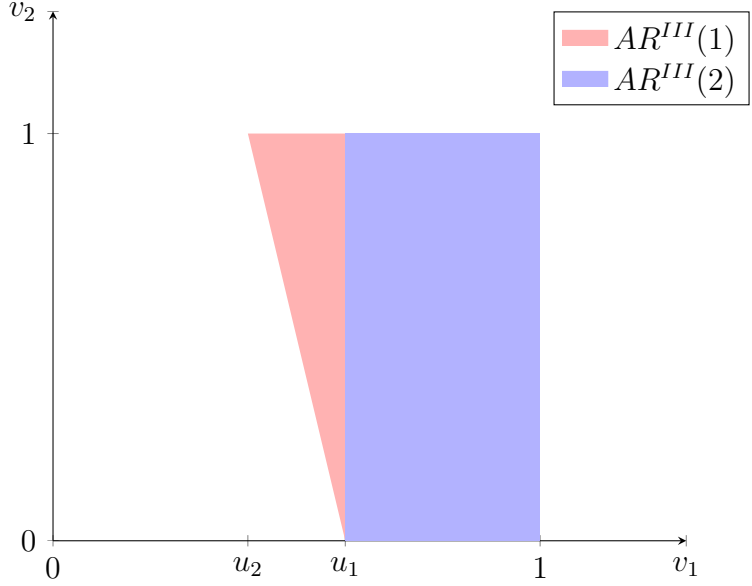
\begin{figure}
\centering
\begin{tikzpicture}
\begin{axis}[
    axis lines = left,
    xmin=0,
        xmax=1.3,
        ymin=0,
        ymax=1.3,
        xtick={0,0.4,0.6,1,1.3},
        ytick={0,1,1.3},
        xticklabels = {$0$, $u_2$, $u_1$, $1$, $v_1$},
        yticklabels = {$0$, $1$, $v_2$},
        legend style={at={(1.1,1)}}
]

\path[name path=axis] (axis cs:0,0) -- (axis cs:1,0);
\path[name path=A] (axis cs:0.4,1) -- (axis cs:0.6,0);
\path[name path=B] (axis cs:0,0.4) -- (axis cs:1,0.4);
\path[name path=C] (axis cs:0,1) -- (axis cs:1,1);
\addplot[area legend, red!30] fill between[of=A and C,  soft clip={domain=0.4:0.6}];
\addplot[blue!30] fill between[of=axis and C, soft clip={domain=0.6:1}];
\legend{$AR^{III}(1)$,$AR^{III}(2)$};
 
\end{axis}
\end{tikzpicture}
\caption{Provision regions of Asymmetric Maxmin Public Good Mechanism (III) } \label{fig:S1}
\end{figure}\\
\textbf{Asymmetric Worst-Case Joint Distribution (III)}\\
Let  $\pi^*(v_1,v_2)$ denote the density of the value profile $(v_1,v_2)$ whenever the density exists. Let $Pr^*(v_1,v_2)$ denote the probability mass of the value profile $(v_1,v_2)$ whenever there is some probability mass on $(v_1,v_2)$. Let $AV(III):=\{v|v_1+(u_1-u_2)v_2\ge u_1\}$.     Asymmetric Worst-Case Joint Distribution (III) has the support $AV(III)$ and  is defined as follows:
$$\pi^*(v_1,v_2)=  \left\{
\begin{array}{lll}
\frac{2u_1(u_2+1)}{(u_1+1)(v_1+v_2)^3}     &      & {v_1+(u_1-u_2)v_2\ge u_1, v_1 \neq 1 ,v_2 \neq 1}\\
\frac{u_1(u_2+1)}{(u_1+1)(1+v_2)^2}  &      & {v_1=1,0\le v_2 < 1}\\
\frac{u_1(u_2+1)}{(u_1+1)(1+v_1)^2}   &      & {u_2\le v_1 < 1,v_2=1}
\end{array} \right. $$
$$Pr^*(1,1)=\frac{u_1(u_2+1)}{2(u_1+1)}$$
Equivalently, \textbf{Asymmetric Worst-Case Joint Distribution (III)} can be described by its marginal distributions and conditional distributions.  The marginal distributions are as follows:  $\pi_1^*(v_1)=\frac{u_1(u_2+1)}{(u_1+1)(\frac{1}{u_2-u_1}(v_1-u_1)+v_1)^2}$ for $ u_2 \le v_1 \le u_1$, $\pi_1^*(v_1)=\frac{u_1(u_2+1)}{(u_1+1)(v_1)^2}$ for $u_1<v_1<1$, $Pr^*_1(1)=\frac{u_1(u_2+1)}{(u_1+1)}$. $\pi_2^*(v_2)=\frac{u_1(u_2+1)}{(u_1+1)((u_2-u_1)v_2+u_1+v_2)^2}$ for $ 0 \le v_2 < 1$,  $Pr^*_2(1)=\frac{u_1}{u_1+1}$. That is, the marginal distribution of  agent 1 is a combination of some generalized Pareto distribution and some equal revenue distribution, while the marginal distribution of  agent 2 is some generalized Pareto distribution with an atom on 1. The conditional distributions are as follows: if $0\le v_2< 1$, then $\pi_1^*(v_1|v_2)=\frac{2((u_2-u_1)v_2+u_1+v_2)^2}{(v_1+v_2)^3}$ for $u_1-(u_1-u_2)v_2\le v_1<1$ and $Pr_1^*(v_1=1|v_2)=\frac{((u_2-u_1)v_2+u_1+v_2)^2}{(1+v_2)^2}$;  if $v_2=1$,  then $\pi_1^*(v_1|v_2=1)=\frac{1+u_2}{(v_1+1)^2}$ for $u_2\le v_1<1$ and $Pr_1^*(v_1=1|v_2=1)=\frac{1+u_2}{2}$. If $u_2\le v_1< u_1$, then $\pi_2^*(v_2|v_1)=\frac{2(\frac{1}{u_2-u_1}(v_1-u_1)+v_1)^2}{(v_1+v_2)^3}$ for $\frac{u_1-v_1}{u_1-u_2}\le v_2<1$ and $Pr_2^*(v_2=1|v_1)=\frac{(\frac{1}{u_2-u_1}(v_1-u_1)+v_1)^2}{(1+v_2)^2}$; if $u_1\le  v_1<1$, then $\pi_2^*(v_2|v_1)=\frac{2v_1^2}{(v_1+v_2)^3}$ for $0\le v_2<1$ and $Pr_2^*(v_2=1|v_1)=\frac{v_1^2}{(1+v_1)^2}$; if $v_1=1$,  then $\pi_2^*(v_2|v_1=1)=\frac{1}{(v_1+1)^2}$ for $0\le v_2<1$ and $Pr_2^*(v_2=1|v_1=1)=\frac{1}{2}$.   That is, the conditional distribution is some truncated generalized Pareto distribution with some mass on 1 (the exact distribution depends on the other agent's value). 
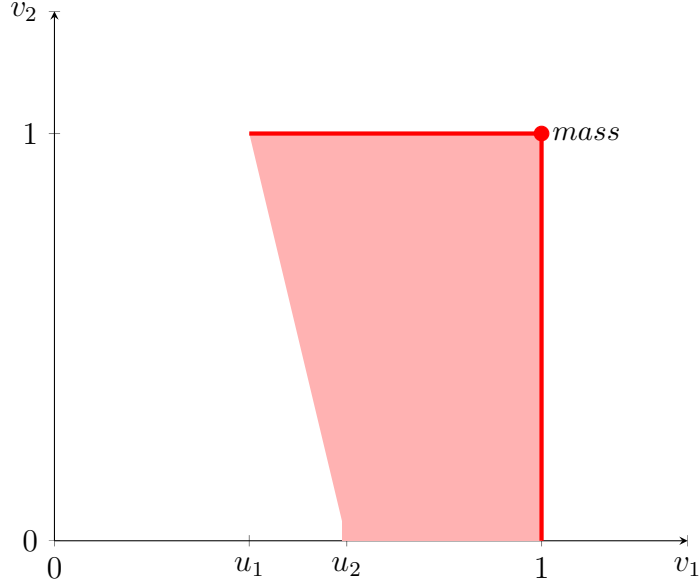
\begin{figure}
\centering
\begin{tikzpicture}
\begin{axis}[
    axis lines = left,
    xmin=0,
        xmax=1.3,
        ymin=0,
        ymax=1.3,
        xtick={0,0.4,0.6,1,1.3},
        ytick={0,1,1.3},
        xticklabels = {$0$, $u_1$, $u_2$, $1$, $v_1$},
        yticklabels = {$0$, $1$, $v_2$},
        legend style={at={(1.1,1)}}
]

\path[name path=axis] (axis cs:0,0) -- (axis cs:1,0);
\path[name path=A] (axis cs: 0.4,1) -- (axis cs:0.6,0);
\path[name path=B] (axis cs:0,0.4) -- (axis cs:1,0.4);
\path[name path=C] (axis cs:0,1) -- (axis cs:1,1);
\addplot[area legend, red!30] fill between[of=A and C,  soft clip={domain=0.4:0.6}];
\addplot[area legend, red!30] fill between[of=axis and C,  soft clip={domain=0.59:1}];
\addplot[red, ultra thick] coordinates {(0.4, 1) (1, 1)};
\addplot[red, ultra thick] coordinates {(1, 0) (1, 1)};
\node[black,right] at (axis cs:1,1){\small{$mass$}};
\node at (axis cs:1,1) [circle, scale=0.5, draw=red,fill=red] {};
\end{axis}
\end{tikzpicture}
\caption{Asymmetric Worst-Case Joint Distribution (III)} \label{fig:S2}
\end{figure}
\begin{lemma}\label{l4}
For any given $(m_1,m_2)\in Area (III)$, there exists a  solution $u_1,u_2$ to to the system of equations \eqref{eq6} and \eqref{eq7}. In addition, $u_1>u_2$.
\end{lemma}
\begin{theorem}\label{t5}
When $(m_1,m_2)\in Area (III)$, \textbf{Asymmetric Maxmin Public Good  Mechanism (III)} and  \textbf{Asymmetric Worst-Case Joint Distribution (III)} form a Nash equilibrium. The revenue guarantee is $\frac{u_1(u_2+1)}{u_1+1}$.
\end{theorem}
Let us illustrate \textbf{Asymmetric Maxmin Public Good Mechanism (III)}. We guess   \textlabel{(A5)}{a5} that in the maxmin solution, the principal provides the  public good with positive probability if and only if the sum of  weighted reported values exceeds certain threshold, i.e., $v_1+(u_1-u_2)v_2>u_1$ where $0\le u_2<u_1<1$. Second, similarly,  in the maxmin solution, it is without loss to assume  \ref{b}. Third, similarly, we conjecture that the support of the worst case joint distribution $\pi^*$ is the area in which $v\in AV^{III}$. Together with (iv) in Proposition \ref{p1}, \ref{a5} and \eqref{eq16}, we obtain that for any $v\in AR^{III}(1)$,
\begin{equation} \label{eq104}
\lambda_1v_1+\lambda_2v_2+\mu = (v_1+v_2)q^*(v)-\int_{u_1-(u_1-u_2)v_2}^{v_1}q^*(x,v_2)dx-\int_{\frac{u_1-v_1}{u_1-u_2}}^{v_2}q^*(v_1,x)dx
\end{equation}
and that for any $v\in AR^{III}(2)$,
\begin{equation} \label{eq104}
\lambda_1v_1+\lambda_2v_2+\mu = (v_1+v_2)q^*(v)-\int_{u_1-(u_1-u_2)v_2}^{v_1}q^*(x,v_2)dx-\int_{0}^{v_2}q^*(v_1,x)dx
\end{equation}
Then following similar procedures for solving for provision probability $q^*(v)$ when $v\in AR^I(1)$ and $v\in AR^I(2)$ in Area (I), we obtain  \textbf{Asymmetric Maxmin Public Good Mechanism (III)}. For \textbf{Asymmetric Worst-Case Joint Distribution (III)}, we have 
\begin{equation}
    \Phi(1,1)>0
\end{equation}
\begin{equation}
  \Phi(v)=0  \quad \forall v_1+(u_1-u_2)v_2\ge u_1\quad and \quad v \neq (1,1)  
\end{equation}
\begin{equation}
  \Phi(v)\le 0  \quad \forall v_1+(u_1-u_2)v_2< u_1
\end{equation}   
The construction procedure for the joint distribution is similar. Therefore we omit it.  To make sure that \textbf{Asymmetric Worst-Case Joint Distribution (III)} satisfies the mean constraints, we have a system of two equations \eqref{eq6} and \eqref{eq7}. Lemma \ref{l4} states the solution exists and is unique. 
\subsection{Area (IV): Distant Expectations}
Let $ Area (IV):=\{ (m_1,m_2)|m_2\le B_{III}(m_1), 1 \ge m_1\ge 0\}$.  We propose a pair of strategy profile as follows. \\
\textbf{Asymmetric Maxmin Public Good Mechanism (IV)}\\   
Let $v=(v_1,v_2)$ be the reported value profile of the two agents. Let $w_1, w_2$ be the unique solution to the following equation:
\begin{equation}\label{eq8}
    m_1=w_1(1-\ln{w_1}):= H^{IV}(w_1,w_2)
\end{equation}
\begin{equation}\label{eq9}
    m_2=w_1\ln{\frac{w_1+w_2}{w_1}}:= H^{IV}(w_1,w_2)
\end{equation}
Let $AR^{IV}:=\{(v_1,v_2)|v_1\ge w_1\}$.  The provision rule is as follows: $$ q^*(v_1,v_2)=  \left\{
\begin{array}{lll}
-\frac{\ln{v_1}}{\ln{w_1}}+1    &      & {v\in AR^{IV}}
\\
0    &      & {otherwise}
\end{array} \right. $$
The payment rule is characterized by Proposition \ref{p1}.
\begin{figure}
\centering
\begin{tikzpicture}
\begin{axis}[
    axis lines = left,
    xmin=0,
        xmax=1.3,
        ymin=0,
        ymax=1.3,
        xtick={0,0.4,1,1.3},
        ytick={0,1,1.3},
        xticklabels = {$0$, $w_1$, $1$, $v_1$},
        yticklabels = {$0$,  $1$, $v_2$},
        legend style={at={(1.1,1)}}
]

\path[name path=axis] (axis cs:0,0) -- (axis cs:1,0);
\path[name path=A] (axis cs:0,0.4) -- (axis cs:0.4,0);
\path[name path=B] (axis cs:0,0.4) -- (axis cs:1,0.4);
\path[name path=C] (axis cs:0,1) -- (axis cs:1,1);
\addplot[area legend, red!30] fill between[of=axis and C,  soft clip={domain=0.4:1}];
\legend{$AR^{IV}$};
 
\end{axis}
\end{tikzpicture}
\caption{Provision regions of Symmetric Maxmin Public Good Mechanism (I) } \label{fig:S1}
\end{figure}
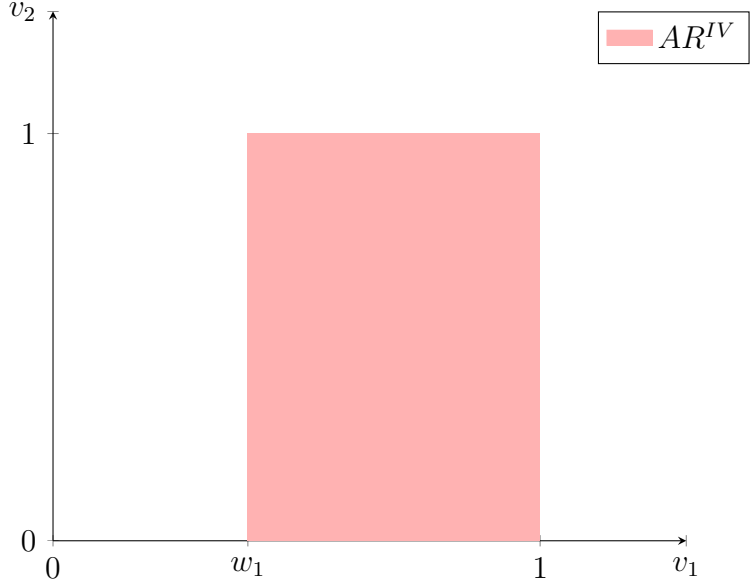\\
\textbf{Asymmetric Worst-Case Joint Distribution (IV)}\\
Let  $\pi^*(v_1,v_2)$ denote the density of the value profile $(v_1,v_2)$ whenever the density exists. Let $Pr^*(v_1,v_2)$ denote the probability mass of the value profile $(v_1,v_2)$ whenever there is some probability mass on $(v_1,v_2)$. Let $AV(IV):=[w_1,1]\times [0,w_2]$.     Asymmetric Worst-Case Joint Distribution (IV) has the support $AV(IV)$ and  is defined as follows:
$$\pi^*(v_1,v_2)=  \left\{
\begin{array}{lll}
\frac{2w_1}{(v_1+v_2)^3}     &      & {v_1\ge w_1 , v_1 \neq 1 ,v_2 \neq w_2}\\
\frac{w_1}{(1+v_2)^2}  &      & {v_1=1,0\le v_2 < w_2}\\
\frac{w_1}{(w_2+v_1)^2}  &      & {r_1\le v_1 < 1,v_2=w_2}
\end{array} \right. $$
$$Pr^*(1,w_2)=\frac{w_1}{w_2+1}$$
Equivalently, \textbf{Asymmetric Worst-Case Joint Distribution (IV)} can be described by its marginal distributions and conditional distributions.  The marginal distributions are as follows: $\pi_1^*(v_1)=\frac{w_1}{v_1^2}$ for $w_1\le v_1<1$, $Pr^*_1(1)=w_1$; $\pi_2^*(v_2)=\frac{w_1}{(w_1+v_2)^2}$ for $0\le v_2 < w_2$, $Pr^*_2(w_2)=\frac{w_1}{w_1+w_2}$.  That is, the marginal distribution of  agent 1 is some equal revenue distribution, while the marginal distribution of agent 2 is   some generalized Pareto distribution with some probability mass on $w_2$. The conditional distributions are as follows: if $0\le v_2< w_2$, then $\pi_1^*(v_1|v_2)=\frac{2(w_1+v_2)^2}{(v_1+v_2)^3}$ for $w_1\le v_1<1$ and $Pr_1^*(v_1=1|v_2)=\frac{(w_1+v_2)^2}{(1+v_2)^2}$; if $v_2=w_2$, then $\pi_1^*(v_1|v_2=w_2)=\frac{w_1+w_2}{(w_2+v_1)^2}$ for $w_1 \le v_1<1$ and $Pr_1^*(v_1=1|v_2=w_2)=\frac{w_1+w_2}{w_2+1}$. If $w_1\le v_1< 1$, then $\pi_2^*(v_2|v_1)=\frac{2(v_1)^2}{(v_1+v_2)^3}$ for $0\le v_2<w_2$ and $Pr_2^*(v_2=w_2|v_1)=\frac{(v_1)^2}{(w_2+v_1)^2}$; if $v_1=1$, then $\pi_2^*(v_2|v_1=1)=\frac{1}{(1+v_2)^2}$ for $0 \le v_2<w_2$ and $Pr_2^*(v_2=w_2|v_1=1)=\frac{1}{w_2+1}$.  That is, the conditional distribution is some truncated generalized Pareto distribution with some probability mass on 1 for agent 1 or on $w_2$ for agent 2 (the exact distribution depends on the other agent's value). 
\begin{remark}
When $(m_1,m_2)\in Boundary (III)$, $w_2=1$ and $w_1\in (0,1]$.
\end{remark}
\begin{remark}
Asymmetric Maxmin Public Good Mechanism (IV) exhibits positive correlation. 
\end{remark}
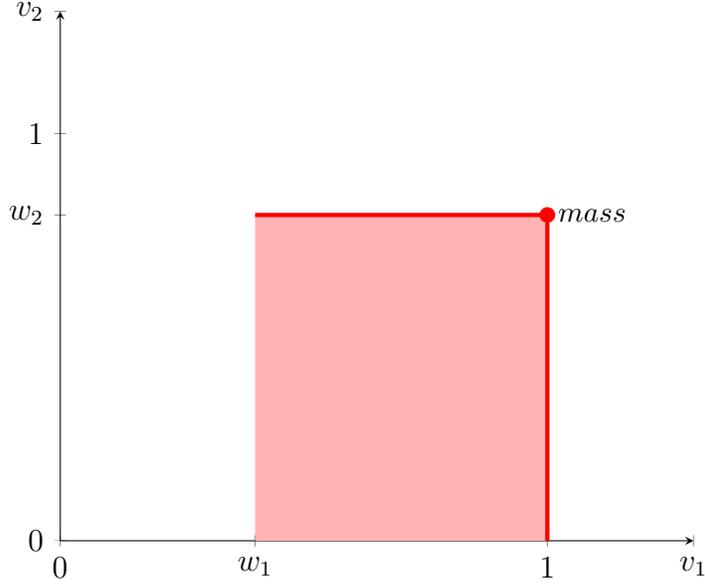
\begin{figure}
\centering
\begin{tikzpicture}
\begin{axis}[
    axis lines = left,
    xmin=0,
        xmax=1.3,
        ymin=0,
        ymax=1.3,
        xtick={0,0.4, 1,1.3},
        ytick={0,0.8, 1,1.3},
        xticklabels = {$0$, $w_1$, $1$, $v_1$},
        yticklabels = {$0$, $w_2$, $1$, $v_2$},
        legend style={at={(1.1,1)}}
]

\path[name path=axis] (axis cs:0,0) -- (axis cs:1,0);
\path[name path=A] (axis cs: 0.4,0.8) -- (axis cs:1,0.8);
\path[name path=B] (axis cs:0,0.4) -- (axis cs:1,0.4);
\path[name path=C] (axis cs:0,1) -- (axis cs:1,1);
\addplot[area legend, red!30] fill between[of=A and axis,  soft clip={domain=0.4:1}];
\addplot[red, ultra thick] coordinates {(0.4, 0.8) (1, 0.8)};
\addplot[red, ultra thick] coordinates {(1, 0) (1, 0.8)};
\node[black,right] at (axis cs:1,0.8){\small{$mass$}};
\node at (axis cs:1,0.8) [circle, scale=0.5, draw=red,fill=red] {};
\end{axis}
\end{tikzpicture}
\caption{Asymmetric Worst-Case Joint Distribution (IV)} \label{fig:S2}
\end{figure}
\begin{lemma}\label{l5}
For any given $(m_1,m_2)\in Area (IV)$, there exists a  solution $w_1,w_2$ to  the system of equations \eqref{eq8} and \eqref{eq9}.
\end{lemma}
\begin{theorem}\label{t6}
When $(m_1,m_2)\in Area (IV)$, \textbf{Asymmetric Maxmin Public Good  Mechanism (IV)} and  \textbf{Asymmetric Worst-Case Joint Distribution (IV)} form a Nash equilibrium. The revenue guarantee is $exp(W_{-1}(-\frac{w_1}{exp(1)})+1)$.
\end{theorem}
Let us illustrate \textbf{Asymmetric Maxmin Public Good Mechanism (IV)}. We guess   \textlabel{(A6)}{a6} that in the maxmin solution, the principal provides the  public good with positive probability if and only if the  reported  value of agent 1 exceeds certain threshold, i.e., $v_1>w_1$ where $w_1\in (0,1]$. Second,  in the maxmin solution, it is without loss  to assume  $q^*(1,v_2)=1$ for any $v_2$. Third, we conjecture that the support of the worst case joint distribution $\pi^*$ is the area in which $v\in AV^{IV}$. Together with (iv) in Proposition \ref{p1}, \ref{a6} and \eqref{eq16}, we obtain that for any $v\in AV^{IV}$,
\begin{equation} \label{eq105}
\lambda_1v_1+\lambda_2v_2+\mu = (v_1+v_2)q^*(v)-\int_{w_1}^{v_1}q^*(x,v_2)dx-\int_{0}^{v_2}q^*(v_1,x)dx
\end{equation}
In the maxmin solution, $q^*(v_1,v_2)$ is independent of $v_2$. It is essentially reduced to one-agent case, to which the solution is known, e.g., \cite{carrasco2018optimal}. Thus we obtain  \textbf{Asymmetric Maxmin Public Good Mechanism (IV)}. For \textbf{Asymmetric Worst-Case Joint Distribution (IV)}, we have 
\begin{equation}
    \Phi(1,1)>0
\end{equation}
\begin{equation}
  \Phi(v)=0  \quad \forall v_1\ge w_1\quad and \quad v_1 \neq 1
\end{equation}
\begin{equation}
  \Phi(v)\le 0  \quad \forall v_1< w_1
\end{equation}    The construction procedure for the joint distribution is similar. Therefore we omit it.  To make sure that \textbf{Asymmetric Worst-Case Joint Distribution (IV)} satisfies the mean constraints, we have a system of two equations \eqref{eq8} and \eqref{eq9}. Lemma \ref{l5} states the solution exists and is unique.

\section{Discussions}\label{s7}
\subsection{More than two agents}\label{s71}
The idea and the methodology are useful to study  the case when there is $N>2$ (any general number) agents. Specifically, we need to use the complementary slackness condition \eqref{eq16} to solve for the maxmin public good mechanism, and use $\Phi(v)=0$ for any value profile in the support except for $v=(\underbrace{1,\cdots, 1}_{N})$ to solve for the worst-case joint distribution. However,  we are not able to provide a complete characterization for the maxmin public good mechanism. The difficulty arises from two sources: we have to divide the expectations into many cases as $N$ is getting large; for each case, we may have to divide the value profiles into many regions and characterize the provision rule for each region. Nonetheless, we are able to provide a characterization of the maxmin public good mechanism and the worst-case joint distribution for a special case in which the symmetric expectation $m\ge 1-\frac{(N-1)^{N-1}}{N^N}$. \\
\textbf{$N$-agent Symmetric Maxmin Public Good Mechanism }\\  Let $v=(v_1,v_2,\cdots, v_N)$ be the reported value profile of the $N$ agents.
Let $r$ be the solution to $m=\frac{(r+N-1)(N^N-(r+N-1)^{N-1})}{(N-1)N^N}$. Let $SR_N:=\{(v_1,v_2,\cdots,v_N)|\sum_{i=1}^Nv_i\ge N-1+r\}$.  The provision rule is as follows: $$ q^*(v_1,v_2,\cdots, v_N)=  \left\{
\begin{array}{lll}
\frac{(\sum_{i=1}^N{v_i})^{N-1}-(r+N-1)^{N-1}}{N^{N-1}-(r+N-1)^{N-1}}    &      & {v\in SR_N}
\\
0    &      & {otherwise}
\end{array} \right. $$
The payment rule is characterized by Proposition \ref{p1}.\\
\textbf{$N$-agent Symmetric Worst-Case Joint Distribution }\\
Let  $\pi^*(v_1,v_2,\cdots, v_N)$ denote the density of the value profile $(v_1,v_2,\cdots, v_N)$ whenever the density exists. Let $Pr^*(v_1,v_2,\cdots, v_N)$ denote the probability mass of the value profile $(v_1,v_2,\cdots, v_N)$ whenever there is some probability mass on $(v_1,v_2,\cdots, v_N)$. Let $A(v)$ denote the set of agents whose value is not 1 given a value profile $v$. Formally, $A(v):=\{i|v_i\neq 1\quad \text{given a value profile $v$}\}$.  $N$-agent Symmetric Worst-Case Joint Distribution  has the support $SR_N$ and is defined as follows:
$$\pi^*(v)=  \frac{|A(v)|!(r+N-1)^N}{N^{N-1}(\sum_{i=1}^Nv_i)^{|A(v)|+1}},\quad v\in SR_N, v\neq (\underbrace{1,\cdots, 1}_{N})$$
$$Pr^*(\underbrace{1,\cdots, 1}_{N})=\frac{(r+N-1)^N}{N^N}$$
\begin{theorem}\label{t8}
For $N$-agent symmetric case, when $m\ge 1-\frac{(N-1)^{N-1}}{N^N}$, \textbf{$N$-agent Symmetric Maxmin Public Good Mechanism} and  \textbf{$N$-agent Symmetric Worst-Case Joint Distribution} form a Nash equilibrium. In addition, the revenue guarantee is $\frac{(r+N-1)^N}{N^{N-1}}$.

\end{theorem}
\subsection{Deterministic mechanisms}
In this section, we restrict attention to deterministic DSIC and EPIR public good mechanisms and  characterize the maxmin public good mechanisms in this class of mechanisms for the two-agent case. Note that  Proposition \ref{p1} still holds, with an additional property that $q(v)$ is either 0 or 1 for any $v$. 
\begin{definition}
\textbf{Provision boundary} of a given deterministic DSIC and EPIR public good mechanism with a provision rule $q$ is a set of value profiles  $\mathcal{B}:=\{\bar{v}=(\bar{v_1},\bar{v_2})| q(\bar{v})=0;\text{for any small}\quad \epsilon>0, \text{}q(\bar{v_1}+\epsilon,\bar{v_2})=1\quad \text{or}\quad q(\bar{v_1},\bar{v_2}+\epsilon)=1\}$.\footnote{For technical reasons, we assume the public good provision probability on the provision boundary is 0. This is to have a minimization problem for Nature. Otherwise we have to replace $\min$ with $\inf$. See also in \cite{carrasco2018optimal}. }
\end{definition}
We observe the provision boundary exhibits a monotone property, which is summarized below. \footnote{To see this, since $\bar{v}' \in \mathcal{B}$ and $\bar{v_1}>\bar{v_1}'$, $q(v_1,\bar{v_2}')=1$. Then by definition, $\bar{v_2}\le \bar{v_2}'$.}
\begin{observation}\label{ob1}
If $\bar{v},\bar{v}' \in \mathcal{B}$ and $\bar{v_1}>\bar{v_1}'$, then $\bar{v_2}\le \bar{v_2}'$.
\end{observation}
\indent The main idea is as follows. We divide all  deterministic DSIC and EPIR public mechanisms into four classes according to the provision boundary. By strong duality, we will focus on the dual program. We propose a relaxation of the dual program by omitting many constraints. Then we are faced with a finite dimensional linear programming problem.  Then we derive an upper bound of the value of the relaxation for each class. Finally we show that the upper bound is tight by constructing a deterministic public good mechanism and a worst-case joint distribution.
\begin{theorem}\label{t7}
(i) When $m_2\ge 2(\sqrt{2}-1)$,  any deterministic DSIC and EPIR public good mechanism satisfying the following properties is a maxmin deterministic public good mechanism:\\
(a). $(1-\sqrt{2(1-m_1)},1) \in \mathcal{B}, (1,1-\sqrt{2(1-m_2)})\in \mathcal{B}$.\\
(b). $\mathcal{B}$ is below (including) the line boundary $\sqrt{2(1-m_2)}v_1+\sqrt{2(1-m_1)}v_2=1-(1-\sqrt{2(1-m_1)})(1-\sqrt{2(1-m_2)})$.\\
(c). Payments are characterized by Proposition \ref{p1}.\\
The worst-case joint distribution put point mass $\sqrt{\frac{1-m_1}{2}},\sqrt{\frac{1-m_2}{2}}$ and $1-\sqrt{\frac{1-m_1}{2}}-\sqrt{\frac{1-m_2}{2}}$ on value profile $(1-\sqrt{2(1-m_1)},1), (1,1-\sqrt{2(1-m_2)})$ and $(1,1)$ respectively. The revenue guarantee is $2(1-\sqrt{\frac{1-m_1}{2}}-\sqrt{\frac{1-m_2}{2}})^2$.\\
(ii) When $m_2< 2(\sqrt{2}-1)$,   the deterministic maxmin public good mechanism is  a dictatorship mechanism: provide the public good if and only if agent 1's reported value exceeds $1-\sqrt{1-m_1}$; payments are characterized by Proposition \ref{p1}. Any joint distribution whose marginal distribution for agent 1 puts point mass  $\sqrt{1-m_1}$ and $1-\sqrt{1-m_1}$ on the value $1-\sqrt{1-m_1}$ and $1$ respectively is a worst-case joint distribution. The revenue guarantee is $(1-\sqrt{1-m_1})^2$.
\end{theorem}
Here are two examples of maxmin deterministic public good mechanisms when $m_2\ge 2(\sqrt{2}-1)$.
\begin{example}\label{e1}
\textit{Linear Mechanism}: the public good is provided with probability of 1 if and only if  $\sqrt{2(1-m_2)}v_1+\sqrt{2(1-m_1)}v_2>1-(1-\sqrt{2(1-m_1)})(1-\sqrt{2(1-m_2)})$.
\end{example}
\begin{example}\label{e2}
\textit{Posted Price Mechanism}: the public good is provided with probability of 1 if and only if 
$v_1>1-\sqrt{2(1-m_1)}$ and $v_2>1-\sqrt{2(1-m_2)}$.
\end{example}
\subsection{Excludable good}
In this section, we consider the design of profit-maximizing \textit{excludable}  good mechanism when the principal only knows the expectations of the private values of the agents. That is, we consider the same problem with the only difference that we allow for agent-specific provision rules. We use $q^E_i(v)\in [0,1]$ to denote the provision rule for agent $i$ when the reported value profile is $v$. 
\begin{manualproposition}{1'}[Revenue Equivalence]\label{p1'}
\textit{Maxmin excludable public good mechanisms have the following properties:\\1. $q^E_i(\cdot,v_{-i})$ is nondecreasing in $v_i$ for all $v_{-i}$}.\\2. $t^E_i(v_i,v_{-i})=v_iq^E_i(v_i,v_{-i})-\int_0^{v_i}q^E_i(s,v_{-i})ds$.
\end{manualproposition}
Proposition \ref{p1'} is a simple adaption of Proposition \ref{p1}. Therefore we omit the proof. Then consider the problem that fixing any joint distribution $\pi$, the principal designs an optimal mechanism $(q^E,t^E)$. An direct implication of Proposition \ref{p1'} is that the expected revenue of $(q^E,t^E)$ under $\pi$ is $$
E[\sum_{i=1}^Nt_i^E(v)]=\int \sum_{i=1}^N q_i^E(v)\Phi^E_i(v)dv$$
where $\Phi_i^E(v)=\pi(v)v_i-[\pi_i(v_{-i})-\Pi_i(v_i,v_{-i})]$. Here $\Phi_i^E(v)$ is the weighted virtual value of type $v_i$ of agent $i$ when the value profile is $v$. Next consider the problem that fixing any mechanism $(q^E,t^E)$, adversarial nature chooses a joint distribution $\pi$ that minimizes the expected revenue.\begin{manuallemma}{1'}\label{l1'}
If $\pi$ is a best response for adversarial nature to a given mechanism $(q^E,t^E)$, then there exists some real numbers $\lambda_1,\cdots, \lambda_N,\mu$ such that 
\begin{equation} \label{eq111}
\sum_{i=1}^N\lambda_iv_i+\mu \le \sum_{i=1}^N t^E_i(v)\quad \forall v \in V
\end{equation}
\begin{equation} \label{eq112}
\sum_{i=1}^N\lambda_iv_i+\mu =\sum_{i=1}^N t^E_i(v)\quad \forall v \in supp(\pi)
\end{equation}
\end{manuallemma} 
Lemma \ref{l1'} is a simple adaption of lemma \ref{l1}. Therefore we omit the proof. Now we give the construction of the maxmin excludable public good mechanism and the worst-case joint distribution for the general $N-$agent case.\\
\textbf{$N$-agent Maxmin Excludable Public Good Mechanism }\\  Let $v=(v_1,v_2,\cdots, v_N)$ be the reported value profile of the $N$ agents.
Let $\gamma_i$ be the solution to $m_i=\gamma_i-\gamma_i\ln{\gamma_i}$.   The provision rule is as follows: $$ q_i^{E*}(v_i,v_{-i})=  \left\{
\begin{array}{lll}
1-\frac{\ln{v_i}}{\ln{\gamma_i}}    &      & {v_i\ge \gamma_i}
\\
0    &      & {otherwise}
\end{array} \right. $$
The payment rule is characterized by an adaption of Proposition \ref{p1}.\\
\textbf{$N$-agent Independent Equal Revenue  Distribution }\\
The marginal distribution of each agent $i$'s value is the equal revenue distribution supported on $[\gamma_i,1]$, i.e., $F^*_i(v_i)=1-\frac{\gamma_i}{v_i}$ for $v_i\in [\gamma_i,1)$ and $F^*_i(1)=1$. $N$-agent Independent Equal Revenue  Distribution is the one in which all values are independently distributed.
\begin{theorem}\label{t9}
For $N$-agent case, \textbf{$N$-agent Excludable Maxmin Public Good Mechanism} and  \textbf{$N$-agent Independent Equal Revenue Distribution} form a Nash equilibrium. In addition, the revenue guarantee is $\sum_{i=1}^N\gamma_i$.

\end{theorem}
\section{Concluding Remarks}\label{s8}
In this paper we characterize    maxmin public good mechanisms  among DSIC and EPIR mechanisms  for the two-agent case given general expectations and for a special $N$-agent case given high symmetric expectations. An important direction for future research is to extend the analysis beyond  the dominant strategy mechanisms.  Another direction is to study situations in which the principal knows other aspects of the joint distribution of values, e.g., the marginals distributions and other moments. 
\section{Appendix A}
\subsection{Characterization of Symmetric Maxmin Public Good Mechanism (I) }
We first consider $v\in SR^I(1)$.  Together with (iv) in Proposition \ref{p1}, \ref{a} and \eqref{eq16}, we obtain that for any $v\in SR^I(1)$,
\begin{equation} \label{eq17}
\lambda_1v_1+\lambda_2v_2+\mu = (v_1+v_2)q^*(v)-\int_{r_1-v_2}^{v_1}q^*(x,v_2)dx-\int_{r_1-v_1}^{v_2}q^*(v_1,x)dx
\end{equation}
To solve for the provision probability, first we  take first order derivatives with respect to $v_1$ and $v_2$ respectively, and we obtain
\begin{equation} \label{eq18}
(v_1+v_2)\frac{\partial q^*(v_1,v_2)}{\partial v_1}-\frac{\partial \int_{r_1-v_1}^{v_2}q^*(v_1,x)dx }{\partial v_1}=\lambda_1
\end{equation}
\begin{equation} \label{eq19}
(v_1+v_2)\frac{\partial q^*(v_1,v_2)}{\partial v_2}-\frac{\partial \int_{r_1-v_2}^{v_1}q^*(x,v_2)dx }{\partial v_2}=\lambda_2
\end{equation}
Then, we take cross partial derivative, with some algebra, we obtain
\begin{equation} \label{eq20}
(v_1+v_2)\frac{\partial q^*(v_1,v_2)}{\partial v_1\partial v_2}=0
\end{equation}
Thus, when $v\in SR^I(1)$,  $q^*(v_1,v_2)$ is separable, which can be written as
\begin{equation}\label{eq21}
    q^*(v_1,v_2)=f(v_1)+g(v_2)
\end{equation}
Plugging \eqref{eq21} into \eqref{eq18} and \eqref{eq19}, we obtain 
\begin{equation}\label{eq22}
r_1f'(v_1)-(f(v_1)+g(r_1-v_1))=\lambda_1
\end{equation}
\begin{equation}\label{eq23}
r_1g'(v_2)-(g(v_2)+f(r_1-v_2))=\lambda_2
\end{equation}
Note both \eqref{eq22} and \eqref{eq23} involve the two functions $f$ and $g$. We guess \textlabel{(C)}{c} that $f(v_1)+g(r-v_1)=0$ and $g(v_2)+f(r_1-v_2)=0$ when $v\in SR^I(1)$, then we can easily solve  \eqref{eq22} and \eqref{eq23}, and we obtain 
\begin{equation}\label{eq24}
    f(v_1)=\frac{\lambda_1}{r_1}v_1+c_1
\end{equation}
\begin{equation}\label{eq25}
    g(v_2)=\frac{\lambda_2}{r_1}v_2+c_2
\end{equation}
In order for \ref{c} to hold, we must have 
\begin{equation}\label{eq26}
    \lambda_1=\lambda_2, c_1+c_2+\lambda_1=0
\end{equation}
Now plugging \eqref{eq24},\eqref{eq25} and \eqref{eq26} into \eqref{eq21}, we obtain for any $v\in SR^I(1)$,
\begin{equation}\label{eq27}
    q^*(v_1,v_2)=\frac{\lambda_1}{r_1}(v_1+v_2-r_1)
\end{equation}
Then, given  $q^*(r_1,r_1)=a$, we obtain  $\lambda_1=a$, and therefore, when $v\in SR^I(1)$, 
\begin{equation}\label{eq28}
    q^*(v_1,v_2)=\frac{a}{r_1}(v_1+v_2-r_1)
\end{equation}
Finally, plugging \eqref{eq28} into \eqref{eq17}, we obtain that $\mu=-ar_1$. Now consider $v\in SR^I(2)$. Given $\lambda_1=\lambda_2=a, \mu=-ar_1$, (iv) in Proposition \ref{p1}, \ref{a} and \eqref{eq16}, we obtain for any $v \in SR^I(2)$,
\begin{equation} \label{eq29}
av_1+av_2-ar_1 = (v_1+v_2)q^*(v)-\int_{r_1-v_2}^{v_1}q^*(x,v_2)dx-\int_{0}^{v_2}q^*(v_1,x)dx
\end{equation}
Note that $q^*(x,v_2)=\frac{a}{r_1}(x+v_2-r_1)$ when $x\le r_1$. Plugging it into \eqref{eq29}, we obtain for any $v \in SR^I(2)$,
\begin{equation} \label{eq30}
av_1+av_2-ar_1  = (v_1+v_2)q^*(v)-\int_{r_1-v_2}^{r_1}\frac{a}{r_1}(x+v_2-r_1)dx-\int_{r_1}^{v_1}q^*(x,v_2)dx-\int_{0}^{v_2}q^*(v_1,x)dx
\end{equation}
We take first order derivatives with respect to $v_1$ and $v_2$ respectively, and we obtain
\begin{equation} \label{eq31}
(v_1+v_2)\frac{\partial q^*(v_1,v_2)}{\partial v_1}-\frac{\partial \int_{0}^{v_2}q^*(v_1,x)dx }{\partial v_1}=a
\end{equation}
\begin{equation} \label{eq32}
(v_1+v_2)\frac{\partial q^*(v_1,v_2)}{\partial v_2}-\frac{av_2}{r_1}-\frac{\partial \int_{r_1}^{v_1}q^*(x,v_2)dx }{\partial v_2}=a
\end{equation}
Then, we take cross partial derivative, with some algebra, we obtain
\begin{equation} \label{eq33}
(v_1+v_2)\frac{\partial q^*(v_1,v_2)}{\partial v_1\partial v_2}=0
\end{equation}
Thus, when $v\in SR^I(2)$,  $q^*(v_1,v_2)$ is separable, which can be written as
\begin{equation}\label{eq34}
    q^*(v_1,v_2)=f(v_1)+g(v_2)
\end{equation}
Plugging \eqref{eq34} into \eqref{eq31} and \eqref{eq32}, we obtain 
\begin{equation}\label{eq35}
v_1f'(v_1)=a
\end{equation}
\begin{equation}\label{eq36}
(r_1+v_2)g'(v_2)-\frac{av_2}{r_1}=a
\end{equation}
The solution to \eqref{eq35} and \eqref{eq36} is 
\begin{equation}\label{eq37}
    f(v_1)=a\ln{v_1}+c_1
\end{equation}
\begin{equation}\label{eq38}
    g(v_2)=\frac{a}{r_1}v_2+c_2
\end{equation}
Then we plug \eqref{eq37} and \eqref{eq38} into \eqref{eq29}, with some algebra, we obtain that  
\begin{equation}\label{eq39}
    c_1+c_2=-a\ln{r_1}
\end{equation}
Therefore, for any $v\in SR^I(2)$,
\begin{equation}\label{eq40}
    q^*(v_1,v_2)=a\ln{v_1}+\frac{a}{r_1}v_2-a\ln{r_1}
\end{equation}
Symmetrically, for $v\in SR^I(3)$,
\begin{equation}\label{eq41}
    q^*(v_1,v_2)=a\ln{v_2}+\frac{a}{r_1}v_1-a\ln{r_1}
\end{equation}
Finally consider $v\in SR^I(4)$.  Given $\lambda_1=\lambda_2=a, \mu=-ar_1$, (iv) in Proposition \ref{p1}, \ref{a} and \eqref{eq16}, we obtain for any $v \in SR^I(4)$,
\begin{equation} \label{eq42}
av_1+av_2-ar_1 = (v_1+v_2)q^*(v)-\int_{0}^{v_1}q^*(x,v_2)dx-\int_{0}^{v_2}q^*(v_1,x)dx
\end{equation}
Note that $q^*(x,v_2)=a\ln{v_2}+\frac{a}{r_1}x-a\ln{r_1}$ when $x\le r_1$ and $q^*(v_1,x)=a\ln{v_1}+\frac{a}{r_1}x-a\ln{r_1}$ when $x\le r_1$ . Plugging them into \eqref{eq42}, we obtain for any $v \in SR^I(4)$,
\begin{equation} \label{eq43}
\begin{split}
av_1+av_2-ar_1  &= (v_1+v_2)q^*(v)-\int_{0}^{r_1}(a\ln{v_2}+\frac{a}{r_1}x-a\ln{r_1})dx-\int_{r_1}^{v_1}q^*(x,v_2)dx\\
&-\int_{0}^{r_1}(a\ln{v_1}+\frac{a}{r_1}x-a\ln{r_1})dx-\int_{r_1}^{v_2}q^*(v_1,x)dx
\end{split}
\end{equation}
We take first order derivatives with respect to $v_1$ and $v_2$ respectively, and we obtain
\begin{equation} \label{eq44}
(v_1+v_2)\frac{\partial q^*(v_1,v_2)}{\partial v_1}-\frac{ar_1}{v_1}-\frac{\partial \int_{r_1}^{v_2}q^*(v_1,x)dx }{\partial v_1}=a
\end{equation}
\begin{equation} \label{eq45}
(v_1+v_2)\frac{\partial q^*(v_1,v_2)}{\partial v_2}-\frac{ar_1}{v_2}-\frac{\partial \int_{r_1}^{v_1}q^*(x,v_2)dx }{\partial v_2}=a
\end{equation}
Then, we take cross partial derivative, with some algebra, we obtain
\begin{equation} \label{eq46}
(v_1+v_2)\frac{\partial q^*(v_1,v_2)}{\partial v_1\partial v_2}=0
\end{equation}
Thus, when $v\in SR^I(2)$,  $q^*(v_1,v_2)$ is separable, which can be written as
\begin{equation}\label{eq47}
    q^*(v_1,v_2)=f(v_1)+g(v_2)
\end{equation}
Plugging \eqref{eq47} into \eqref{eq44} and \eqref{eq45}, we obtain 
\begin{equation}\label{eq48}
(r_1+v_1)f'(v_1)-\frac{ar_1}{v_1}=a
\end{equation}
\begin{equation}\label{eq49}
(r_1+v_2)g'(v_2)-\frac{ar_1}{v_2}=a
\end{equation}
The solution to \eqref{eq48} and \eqref{eq49} is 
\begin{equation}\label{eq50}
    f(v_1)=a\ln{v_1}+c_1
\end{equation}
\begin{equation}\label{eq51}
    g(v_2)=a\ln{v_2}+c_2
\end{equation}
By \ref{b}, we have $c_1+c_2=1$. Therefore, for any $v\in SR^I(4)$,
\begin{equation}\label{eq52}
    q^*(v_1,v_2)=a\ln{v_1}+a\ln{v_2}+1
\end{equation}
Then we plug \eqref{eq52} into  \eqref{eq43} and checked that \eqref{eq43} holds with some algebra. Finally, as $q^*(r_1,r_1)=a=2a\ln{r_1}+1$, we obtain that $a=\frac{1}{1-2\ln{r_1}}$.
\subsection{Characterization of Symmetric Worst-Case Joint Distribution (I)}
We start from constructing the joint distribution for the boundary value profiles, i.e., either $v_1=1$ or $v_2=1$. Assume $Pr^*(1,1)=b$. Consider value profiles $(v_1,1)$ in which $0\le v_1< 1$. Define $S^*(v_1,1)\equiv \int_{[v_1,1)}\pi^*(x,0)dx+Pr^*(1,1)$ for $0\le v_1< 1$; $S^*(1,1)\equiv Pr^*(1,1)=b$.  Then we have $\pi^*(v_1,1)=-\frac{\partial S^*(v_1,1)}{\partial v_1}$ for $0\le v_1< 1$. Since the weighted virtual values for value profiles $(v_1,1)$ in which $0\le v_1 < 1$ are zeroes, we obtain for $0\le v_1< 1$,
\begin{equation}\label{eq56}
    \pi^*(v_1,1)(v_1+1)-S^*(v_1,1)=0
\end{equation}
Note \eqref{eq56} is a simple ordinary differential equation, to which the solution is 
\begin{equation}\label{eq57}
    S^*(v_1,1)=\frac{2b}{v_1+1}, \pi^*(v_1,0)=\frac{2b}{(v_1+1)^2} \quad \forall 0\le v_1< 1
\end{equation}
Then consider value profiles $(1,v_2)$ in which $0\le v_2 \le 1$. Define $S^*(1,v_2)\equiv \int_{[v_2,1)}\pi^*(1,x)dx+Pr^*(1,1)$ for $0\le v_2< 1$.  Symmetrically, we obtain that 
\begin{equation}\label{eq58}
    S^*(1,v_2)=\frac{2b}{1+v_2}, \pi^*(1,v_2)=\frac{2b}{(1+v_2)^2}\quad \forall 0\le v_2< 1
\end{equation}
Now we will construct the joint distribution for the interior value profiles in the support, i.e., $v_1+v_2\ge r_1$ and $v_1\neq 1, v_2\neq 1$. Define $S^*(v_1,v_2)\equiv \int_{[v_1,1)}\pi^*(x,v_2)dx+\pi^*(1,v_2)$ for $v_1+v_2\ge r_1$ and $v_1\neq 1, v_2\neq 1$.  Then we have $\pi^*(v_1,v_2)=-\frac{\partial S^*(v_1,v_2)}{\partial v_1}$ for $v_1+v_2\ge r_1$ and $v_1\neq 1, v_2\neq 1$. Since the weighted virtual values for value profiles $(v_1,v_2)$ in which $v_1+v_2\ge r_1$ and $v_1\neq 1, v_2\neq 1$ are zeroes, we obtain for $v_1+v_2\ge r_1$ and $v_1\neq 1, v_2\neq 1$,
\begin{equation}\label{eq59}
    \pi^*(v_1,v_2)(v_1+v_2)-S^*(v_1,v_2)-\int_{[v_2,1)}\pi^*(v_1,x)dx-\pi^*(v_1,1)=0
\end{equation}
 By taking the cross partial derivative, we find $S^*(v_1,v_2)$ is not separable. We take the guess and verify approach to solve for \eqref{eq59}. We guess that for $v_1+v_2\ge r_1$ and $v_1\neq 1, v_2\neq 1$, 
\begin{equation}\label{eq60}
    S^*(v_1,v_2)=\frac{2b}{(v_1+v_2)^2}
\end{equation}
Then the LHS of \eqref{eq59} is $\frac{4b}{(v_1+v_2)^3}(v_1+v_2)-\frac{2b}{(v_1+v_2)^2}-
\int_{[v_2,1)}\frac{4b}{(v_1+x)^3}ds-\frac{2b}{(v_1+1)^2}$, which can be shown to be 0 with some algebra. Thus, we verified the guess.\\
\indent To solve for $b$, we use the fact that $\pi^*(v)$ is a distribution. We note the marginal distribution for agent 2 (the same for agent 1) is as follows: $\pi^*_2(v_2)=S(r_1-v_2,v_2)=\frac{2b}{(r_1-v_2+v_2)^2}=\frac{2b}{r_1^2}$ for $0\le v_2\le r_1$, $\pi^*_2(v_2)=S(0,v_2)=\frac{2b}{(0+v_2)^2}=\frac{2b}{v_2^2}$ for $r_1< v_2<1$ and $Pr^*_2(v_2=1)=S^*(0,1)=2b$. Since the integration is 1, we obtain \begin{equation}\label{eq61}
    \frac{2b}{r_1^2}\cdot r_1+\int_{(r_1,1)}\frac{2b}{v_2^2}+2b=1
\end{equation}
Thus, we obtain 
\begin{equation}\label{eq62}
   b=\frac{r_1}{4}
\end{equation}
\indent So far we have constructed \textbf{Symmetric Worst-Case Joint Distribution (I)}. The final step is  to make sure that \textbf{Symmetric Worst-Case Joint Distribution (I)} satisfies the mean constraints, which will allow us to solve for the monopoly reserve $r_1$. Given the marginal distribution for agent 2 (the same for agent 1), we have the following mean constraint,
\begin{equation}\label{eq63}
    \int_0^{r_1}\frac{2b}{r_1^2}xdx+\int_{r_1}^1\frac{2b}{x^2}xdx+2b=m
\end{equation}
Plugging \eqref{eq62} into \eqref{eq63}, we obtain 
\begin{equation}\label{eq64}
   \frac{r_1(3-2\ln{r_1})}{4}=m
\end{equation}
Note that the LHS of \eqref{eq64} is strictly increasing\footnote{The first order derivative is $\frac{1-2\ln{r_1}}{4}>0$ for $r_1\in (0,1)$.} with respect to $r_1$ for $r_1\in (0,1)$. In addition, the LHS of \eqref{eq64} is $\frac{3}{4}$ when $r_1=1$. Therefore, for any $m\in (0,\frac{3}{4})$, there is a unique solution $r_1\in(0,1)$ to \eqref{eq64}. Indeed, $r_1=\exp(W_{-1}(-2m\exp(-\frac{3}{2}))+\frac{3}{2})$. 
\subsection{Characterization of Asymmetric Maxmin Public Good Mechanism (I)}
 We first consider $v\in AR^I(1)$.  Together with (iv) in Proposition \ref{p1}, \ref{a3} and \eqref{eq16}, we obtain that for any $v\in AR^I(1)$,
\begin{equation} \label{eq67}
\lambda_1v_1+\lambda_2v_2+\mu = (v_1+v_2)q^*(v)-\int_{s_1-\frac{s_1}{s_2}v_2}^{v_1}q^*(x,v_2)dx-\int_{s_2-\frac{s_2}{s_1}v_1}^{v_2}q^*(v_1,x)dx
\end{equation}
To solve for the provision probability, first we  take first order derivatives with respect to $v_1$ and $v_2$ respectively, and we obtain
\begin{equation} \label{eq68}
(v_1+v_2)\frac{\partial q^*(v_1,v_2)}{\partial v_1}-\frac{\partial \int_{s_2-\frac{s_2}{s_1}v_1}^{v_2}q^*(v_1,x)dx }{\partial v_1}=\lambda_1
\end{equation}
\begin{equation} \label{eq69}
(v_1+v_2)\frac{\partial q^*(v_1,v_2)}{\partial v_2}-\frac{\partial \int_{s_2-\frac{s_2}{s_1}v_2}^{v_1}q^*(x,v_2)dx }{\partial v_2}=\lambda_2
\end{equation}
Then, we take cross partial derivative, with some algebra, we obtain
\begin{equation} \label{eq70}
(v_1+v_2)\frac{\partial q^*(v_1,v_2)}{\partial v_1\partial v_2}=0
\end{equation}
Thus, when $v\in AR^I(1)$,  $q^*(v_1,v_2)$ is separable, which can be written as
\begin{equation}\label{eq71}
    q^*(v_1,v_2)=f(v_1)+g(v_2)
\end{equation}
Plugging \eqref{eq71} into \eqref{eq68} and \eqref{eq69}, we obtain 
\begin{equation}\label{eq72}
(v_1+s_2-\frac{s_2}{s_1}v_1)f'(v_1)-\frac{s_2}{s_1}(f(v_1)+g(s_2-\frac{s_2}{s_1}v_1))=\lambda_1
\end{equation}
\begin{equation}\label{eq73}
(v_2+s_1-\frac{s_1}{s_2}v_2)g'(v_2)-\frac{s_1}{s_2}(g(v_2)+f(s_1-\frac{s_1}{s_2}v_2))=\lambda_2
\end{equation}
Note both \eqref{eq72} and \eqref{eq73} involve the two functions $f$ and $g$. We guess \textlabel{(C3)}{c3} that $f(v_1)+g(s_2-\frac{s_2}{s_1}v_1)=0$ and $g(v_2)+f(s_1-\frac{s_1}{s_2}v_2)=0$ when $v\in AR^I(1)$, then we can easily solve  \eqref{eq72} and \eqref{eq73}, and we obtain 
\begin{equation}\label{eq74}
    f(v_1)=\frac{\lambda_1}{1-\frac{s_2}{s_1}}\ln{(v_1+s_2-\frac{s_2}{s_1}v_1)}+c_1
\end{equation}
\begin{equation}\label{eq75}
    g(v_2)=\frac{\lambda_2}{1-\frac{s_1}{s_2}}\ln{(v_2+s_1-\frac{s_1}{s_2}v_2)}+c_2
\end{equation}
In order for \ref{c3} to hold, we must have 
\begin{equation}\label{eq76}
    \frac{\lambda_1}{1-\frac{s_2}{s_1}}+\frac{\lambda_2}{1-\frac{s_1}{s_2}}=0, c_1+c_2=0
\end{equation}
Now plugging \eqref{eq74},\eqref{eq75} and \eqref{eq76} into \eqref{eq71}, we obtain for any $v\in AR^I(1)$,
\begin{equation}\label{eq77}
    q^*(v_1,v_2)=\frac{\lambda_1}{1-\frac{s_2}{s_1}}(\ln{(v_1+s_2-\frac{s_2}{s_1}v_1)}-\ln{(v_2+s_1-\frac{s_1}{s_2}v_2)})
\end{equation}
Then, given $q^*(s_1,s_2)=c$, we obtain  $\lambda_1=\frac{c(1-\frac{s_2}{s_1})}{\ln{\frac{s_1}{s_2}}}$, and therefore, when $v\in AR^I(1)$, 
\begin{equation}\label{eq78}
    q^*(v_1,v_2)=\frac{c}{\ln{\frac{s_1}{s_2}}}(\ln{(v_1+s_2-\frac{s_2}{s_1}v_1)}-\ln{(v_2+s_1-\frac{s_1}{s_2}v_2)})
\end{equation}
Finally, plugging \eqref{eq78} into \eqref{eq77}, we obtain that $\mu=-\frac{c(s_1-s_2)}{\ln{\frac{s_1}{s_2}}}$. Now consider $v\in AR^I(2)$. Given $\lambda_1= \frac{c(1-\frac{s_2}{s_1})}{\ln{\frac{s_1}{s_2}}}, \lambda_2=-\frac{c(1-\frac{s_1}{s_2})}{\ln{\frac{s_1}{s_2}}}, \mu=-\frac{c(s_1-s_2)}{\ln{\frac{s_1}{s_2}}}$, (iv) in Proposition \ref{p1}, \ref{a3} and \eqref{eq76}, we obtain for any $v \in AR^I(2)$,
\begin{equation} \label{eq79}
\frac{c}{\ln{\frac{s_1}{s_2}}}((1-\frac{s_2}{s_1})v_1-(1-\frac{s_1}{s_2})v_2-(s_1-s_2)) = (v_1+v_2)q^*(v)-\int_{s_1-\frac{s_1}{s_2}v_2}^{v_1}q^*(x,v_2)dx-\int_{0}^{v_2}q^*(v_1,x)dx
\end{equation}
Note that $q^*(x,v_2)=\frac{c}{\ln{\frac{s_1}{s_2}}}(\ln{(x+s_2-\frac{s_2}{s_1}x)}-\ln{(v_2+s_1-\frac{s_1}{s_2}v_2)})$ when $x\le s_1$. Plugging it into \eqref{eq79}, we obtain for any $v \in AR^I(2)$,
\begin{equation} \label{eq80}
\begin{split}
\frac{c}{\ln{\frac{s_1}{s_2}}}((1-\frac{s_2}{s_1})v_1-(1-\frac{s_1}{s_2})v_2-(s_1-s_2))  &= (v_1+v_2)q^*(v) \\ & - \int_{s_1-\frac{s_1}{s_2}v_2}^{s_1}\frac{c}{\ln{\frac{s_1}{s_2}}}(\ln{(x+s_2-\frac{s_2}{s_1}x)}-\ln{(v_2+s_1-\frac{s_1}{s_2}v_2)})dx  \\ &-\int_{s_1}^{v_1}q^*(x,v_2)dx-\int_{0}^{v_2}q^*(v_1,x)dx
\end{split}
\end{equation}
We take first order derivatives with respect to $v_1$ and $v_2$ respectively, and we obtain
\begin{equation} \label{eq81}
(v_1+v_2)\frac{\partial q^*(v_1,v_2)}{\partial v_1}-\frac{\partial \int_{0}^{v_2}q^*(v_1,x)dx }{\partial v_1}=\frac{c}{\ln{\frac{s_1}{s_2}}}(1-\frac{s_2}{s_1})
\end{equation}
\begin{equation} \label{eq82}
(v_1+v_2)\frac{\partial q^*(v_1,v_2)}{\partial v_2}+\frac{c}{\ln{\frac{s_1}{s_2}}}\frac{s_1}{s_2}(1-\frac{s_1}{s_2})\frac{v_2}{v_2+s_1-\frac{s_1}{s_2}v_2}-\frac{\partial \int_{s_1}^{v_1}q^*(x,v_2)dx }{\partial v_2}=-\frac{c}{\ln{\frac{s_1}{s_2}}}(1-\frac{s_1}{s_2})
\end{equation}
Then, we take cross partial derivative, with some algebra, we obtain
\begin{equation} \label{eq83}
(v_1+v_2)\frac{\partial q^*(v_1,v_2)}{\partial v_1\partial v_2}=0
\end{equation}
Thus, when $v\in AR^I(2)$,  $q^*(v_1,v_2)$ is separable, which can be written as
\begin{equation}\label{eq84}
    q^*(v_1,v_2)=f(v_1)+g(v_2)
\end{equation}
Plugging \eqref{eq84} into \eqref{eq81} and \eqref{eq82}, we obtain 
\begin{equation}\label{eq85}
v_1f'(v_1)=\frac{c}{\ln{\frac{s_1}{s_2}}}(1-\frac{s_2}{s_1})
\end{equation}
\begin{equation}\label{eq86}
(s_1+v_2)g'(v_2)-\frac{c}{\ln{\frac{s_1}{s_2}}}\frac{s_1}{s_2}(1-\frac{s_1}{s_2})\frac{v_2}{v_2+s_1-\frac{s_1}{s_2}v_2}=-\frac{c}{\ln{\frac{s_1}{s_2}}}(1-\frac{s_1}{s_2})
\end{equation}
The solution to \eqref{eq85} and \eqref{eq86} is 
\begin{equation}\label{eq87}
    f(v_1)=\frac{c}{\ln{\frac{s_1}{s_2}}}(1-\frac{s_2}{s_1})\ln{v_1}+c_1
\end{equation}
\begin{equation}\label{eq88}
    g(v_2)=-\frac{c}{\ln{\frac{s_1}{s_2}}}\ln{(v_2+s_1-\frac{s_1}{s_2}v_2)}+c_2
\end{equation}
Then we plug \eqref{eq87} and \eqref{eq88} into \eqref{eq80}, with some algebra, we obtain that  
\begin{equation}\label{eq89}
    c_1+c_2=\frac{c}{\ln{\frac{s_1}{s_2}}}\frac{s_2}{s_1}\ln{s_1}
\end{equation}
Therefore, for any $v\in AR^I(2)$,
\begin{equation}\label{eq90}
    q^*(v_1,v_2)=\frac{c}{\ln{\frac{s_1}{s_2}}}((1-\frac{s_2}{s_1})\ln{v_1}-\ln{(v_2+s_1-\frac{s_1}{s_2}v_2)}+\frac{s_2}{s_1}\ln{s_1})  
\end{equation}
Symmetrically, for $v\in AR^I(3)$,
\begin{equation}\label{eq91}
    q^*(v_1,v_2)=\frac{c}{\ln{\frac{s_1}{s_2}}}(\ln{(v_1+s_2-\frac{s_2}{s_1}v_1)}-(1-\frac{s_1}{s_2})\ln{v_2}-\frac{s_1}{s_2}\ln{s_2}) 
\end{equation}
Finally consider $v\in AR^I(4)$.  Given $\lambda_1= \frac{c(1-\frac{s_2}{s_1})}{\ln{\frac{s_1}{s_2}}}, \lambda_2=-\frac{c(1-\frac{s_1}{s_2})}{\ln{\frac{s_1}{s_2}}}, \mu=-\frac{c(s_1-s_2)}{\ln{\frac{s_1}{s_2}}}$, (iv) in Proposition \ref{p1}, \ref{a3} and \eqref{eq76}, we obtain for any $v \in AR^I(4)$,
\begin{equation} \label{eq92}
\frac{c}{\ln{\frac{s_1}{s_2}}}((1-\frac{s_2}{s_1})v_1-(1-\frac{s_1}{s_2})v_2-(s_1-s_2)) = (v_1+v_2)q^*(v)-\int_{0}^{v_1}q^*(x,v_2)dx-\int_{0}^{v_2}q^*(v_1,x)dx
\end{equation}
Note that $q^*(x,v_2)=\frac{c}{\ln{\frac{s_1}{s_2}}}(\ln{(x+s_2-\frac{s_2}{s_1}x)}-(1-\frac{s_1}{s_2})\ln{v_2}-\frac{s_1}{s_2}\ln{s_2}) $ when $x\le s_1$ and $q^*(v_1,x)=\frac{c}{\ln{\frac{s_1}{s_2}}}((1-\frac{s_2}{s_1})\ln{v_1}-\ln{(x+s_1-\frac{s_1}{s_2}x)}+\frac{s_2}{s_1}\ln{s_1})$ when $x\le s_2$ . Plugging them into \eqref{eq92}, we obtain for any $v \in AR^I(4)$,
\begin{equation} \label{eq93}
\begin{split}
\frac{c}{\ln{\frac{s_1}{s_2}}}((1-\frac{s_2}{s_1})v_1-(1-\frac{s_1}{s_2})v_2-(s_1-s_2))  &= (v_1+v_2)q^*(v)-\int_{s_1}^{v_1}q^*(x,v_2)dx-\int_{s_2}^{v_2}q^*(v_1,x)dx\\
&-\int_{0}^{s_1}(\frac{c}{\ln{\frac{s_1}{s_2}}}(\ln{(x+s_2-\frac{s_2}{s_1}x)}-(1-\frac{s_1}{s_2})\ln{v_2}-\frac{s_1}{s_2}\ln{s_2}))dx\\
& -\int_{0}^{s_2}(\frac{c}{\ln{\frac{s_1}{s_2}}}((1-\frac{s_2}{s_1})\ln{v_1}-\ln{(x+s_1-\frac{s_1}{s_2}x)}+\frac{s_2}{s_1}\ln{s_1}))dx
\end{split}
\end{equation}
We take first order derivatives with respect to $v_1$ and $v_2$ respectively, and we obtain
\begin{equation} \label{eq94}
(v_1+v_2)\frac{\partial q^*(v_1,v_2)}{\partial v_1}-\frac{cs_2}{\ln{\frac{s_1}{s_2}}}(1-\frac{s_2}{s_1})\frac{1}{v_1}-\frac{\partial \int_{s_2}^{v_2}q^*(v_1,x)dx }{\partial v_1}=\frac{c}{\ln{\frac{s_1}{s_2}}}(1-\frac{s_2}{s_1})
\end{equation}
\begin{equation} \label{eq95}
(v_1+v_2)\frac{\partial q^*(v_1,v_2)}{\partial v_2}+\frac{cs_1}{\ln{\frac{s_1}{s_2}}}(1-\frac{s_1}{s_2})\frac{1}{v_2}-\frac{\partial \int_{r_1}^{v_1}q^*(x,v_2)dx }{\partial v_2}=-\frac{c}{\ln{\frac{s_1}{s_2}}}(1-\frac{s_1}{s_2})
\end{equation}
Then, we take cross partial derivative, with some algebra, we obtain
\begin{equation} \label{eq96}
(v_1+v_2)\frac{\partial q^*(v_1,v_2)}{\partial v_1\partial v_2}=0
\end{equation}
Thus, when $v\in AR^I(4)$,  $q^*(v_1,v_2)$ is separable, which can be written as
\begin{equation}\label{eq97}
    q^*(v_1,v_2)=f(v_1)+g(v_2)
\end{equation}
Plugging \eqref{eq97} into \eqref{eq94} and \eqref{eq95}, we obtain 
\begin{equation}\label{eq98}
(s_2+v_1)f'(v_1)-\frac{cs_2}{\ln{\frac{s_1}{s_2}}}(1-\frac{s_2}{s_1})\frac{1}{v_1}=\frac{c}{\ln{\frac{s_1}{s_2}}}(1-\frac{s_2}{s_1})
\end{equation}
\begin{equation}\label{eq99}
(s_1+v_2)g'(v_2)+\frac{cs_1}{\ln{\frac{s_1}{s_2}}}(1-\frac{s_1}{s_2})\frac{1}{v_2}=-\frac{c}{\ln{\frac{s_1}{s_2}}}(1-\frac{s_1}{s_2})
\end{equation}
The solution to \eqref{eq98} and \eqref{eq99} is 
\begin{equation}\label{eq100}
    f(v_1)=\frac{c}{\ln{\frac{s_1}{s_2}}}(1-\frac{s_2}{s_1})\ln{v_1}+c_1
\end{equation}
\begin{equation}\label{eq101}
    g(v_2)=-\frac{c}{\ln{\frac{s_1}{s_2}}}(1-\frac{s_1}{s_2})\ln{v_2}+c_2
\end{equation}
By \ref{b}, we have $c_1+c_2=1$. Therefore, for any $v\in AR^I(4)$,
\begin{equation}\label{eq102}
    q^*(v_1,v_2)=\frac{c}{\ln{\frac{s_1}{s_2}}}((1-\frac{s_2}{s_1})\ln{v_1}-(1-\frac{s_1}{s_2})\ln{v_2})+1
\end{equation}
Then we plug \eqref{eq102} into  \eqref{eq93} and checked that \eqref{eq93} holds with some algebra. Finally, as $q^*(s_1,s_2)=c=\frac{c}{\ln{\frac{s_1}{s_2}}}((1-\frac{s_2}{s_1})\ln{s_1}-(1-\frac{s_1}{s_2})\ln{s_2})+1$, we obtain that $c=\frac{1}{1-\frac{(1-\frac{s_2}{s_1})\ln{s_1}-(1-\frac{s_1}{s_2})\ln{s_2}}{\ln{\frac{s_1}{s_2}}}}$.
\section{Appendix B}
\subsection{Proof of Proposition \ref{p1}}
(i) $q(\cdot,v_{-i})$ nondecreasing:\\
Dominant strategy incentive compatibility for a type $v_i$ of agent $i$ requires that for any $v_i'\neq v_i$: $$v_iq(v_i,v_{-i}) - t_i(v_i,v_{-i})\ge v_iq(v_i',v_{-i}) - t_i(v_i',v_{-i})  $$DSIC also requires that:$$ v_i'q(v_i',v_{-i}) - t_i(v_i',v_{-i})\ge v_i'q(v_i,v_{-i}) - t_i(v_i,v_{-i}).$$Adding the two inequalities, we have that: $$(v_i-v_i')(q(v_i,v_{-i})-q(v_i',v_{-i}))\ge 0$$It follows that $q(v_i,v_{-i})\ge q(v_i',v_{-i})$ whenever $v_i>v_i'$ .\\\\
(ii) $t_i(v_i,v_{-i})=v_iq(v_i,v_{-i})-\int_{0}^{v_i}q(s,v_{-i})ds$:\\
Define $$U_i(v_i,v_{-i})=v_iq(v_i,v_{-i})-t_i(v_i,v_{-i})$$ 
By the two inequalities in (i), we get $$(v_i'-v_i)q(v_i,v_{-i})\le U_i(v_i',v_{-i})-U_i(v_i,v_{-i})\le (v_i'-v_i)q(v_i',v_{-i})$$ Dividing throughout by $v_i'-v_i$: $$q(v_i,v_{-i})\le \frac{U_i(v_i',v_{-i})-U_i(v_i,v_{-i})}{(v_i'-v_i)}\le q(v_i',v_{-i})$$As $v_i\uparrow v_i'$, we get: $$
\frac{dU_i(v_i,v_{-i})}{dv_i}=q(v_i,v_{-i})$$Then we get $$t_i(v_i,v_{-i})= v_iq(v_i,v_{-i})-\int_0^{v_i}q(s,v_{-i})-U_i(0,v_{-i})$$
Note $U_i(0,v_{-i})\ge 0$ by the EPIR constraint. If $U_i(0,v_{-i})>0$, then we can reduce it to 0 so that we can increase the revenue from all value profiles and the value of the problem will be strictly greater. Thus, for any maxmin public good mechanism, $U_i(0,v_{-i})=0$ and $t_i(v_i,v_{-i})=v_iq(v_i,v_{-i})-\int_{0}^{v_i}q(s,v_{-i})$.
\subsection{Proof of Lemma \ref{l1}}
Given a DSIC and EPIR mechanism $(q,t)$, the (P) primal minimization problem of Nature is as follows (with dual variables in the bracket): 
 $$(Primal)\min_{F\in \Pi(m_1,\cdots, m_N)}\int \sum_{i=1}^N t_i(v)dF$$ s.t.
    $$\int v_idF=m_i\quad (\lambda_i)\quad \forall i$$
    $$\int dF=1 \quad (\mu)$$
It has the following (D) dual maximization problem:
    $$(Dual)\max_{\lambda_1,\cdots, \lambda_N,\mu \in \mathcal{R}}\sum_{i=1}^N \lambda_i m_i+\mu$$ s.t.
    $$\sum_{i=1}^N \lambda_iv_i+\mu \le \sum_{i=1}^N t_i(v) \quad (dF)\quad \forall v$$
Note that the value of (P) is bounded by $N$ as $\sum_{i=1}^Nt_i(v)\le N$. In addition, the trivial joint distribution that put all probability mass on the point $(m_1,\cdots, m_N)$  is in the interior of the primal cone. Then by Theorem 3.12 in \cite{anderson1987linear}, strong duality holds. Then, by the Complementary Slackness, \eqref{eq16} holds. And \eqref{eq15} is  the feasibility constraint of (D).
\subsection{Proof of Theorem \ref{t1}}
(i): \textbf{Symmetric Maxmin Public Good Mechanism (I)} is a best response to \textbf{Symmetric Worst-Case Joint Distribution (I)}. Note Symmetric Worst-Case Joint Distribution (I) satisfies \eqref{eq56} and \eqref{eq59}. Also note there is a probability mass on the value profile (1,1). Thus \eqref{eq53}, \eqref{eq54} and \eqref{eq55} hold. Then any feasible and monotone mechanism in which the public good is provided  with some positive probability if and only if $v_1 + v_2> r_1$ and the public good is provided with probability 1 when $(v_1,v_2)  =  (1,1)$ is a best response for the principal. It is easy to see that  Symmetric Maxmin Public Good Mechanism (I) is such a mechanism.\\
(ii): \textbf{Symmetric Worst-Case Joint Distribution (I)} is a best response to \textbf{Maxmin Public Good Mechanism (I)}. We use the duality theory to show (ii). First note that by \eqref{eq61} and \eqref{eq63}, all the three constraints in (P) holds. By the weighted virtual value representation, the value of (P) given Symmetric Worst-Case Joint Distribution (I) and Symmetric Maxmin Public Good Mechanism (I) is simply $Pr(1,1)\times (1+1)=\frac{r_1}{2}$. Second, the constraints in (D) hold for all value profiles. To see this, note for any value profile $v=(v_1,v_2)$ outside the support of Symmetric Worst-Case Joint Distribution (I), since $\lambda_1=\lambda_2=a>0$, we have $$\lambda_1 v_1+\lambda_2 v_2 + \mu < \lambda_1 r_1+\lambda_2 0 + \mu=0$$ For any value profile $v=(v_1,v_2)$ inside the support of Symmetric Worst-Case Joint Distribution (I), the constraints (the complementary slackness)  hold given \eqref{eq17}, \eqref{eq29} and \eqref{eq41}.  Finally, the value of (D) given the constructed $\lambda_1,\lambda_2, \mu$ is $\lambda_1 m+\lambda_2 m +\mu$, which, by some algebra, is equal to $\frac{r_1}{2}$. By the linear programming duality theory, (ii) holds and the revenue guarantee is $\frac{r_1}{2}$.
\subsection{Proof of Theorem \ref{t2}}
(i): \textbf{Symmetric Maxmin Public Good Mechanism (II)} is a best response to \textbf{Symmetric Worst-Case Joint Distribution (II)}. The proof is similar to the proof of (i) of Theorem \ref{t1}.\\
(ii): \textbf{Symmetric Worst-Case Joint Distribution (II)} is a best response to \textbf{Symmetric Maxmin Public Good Mechanism (II)}. We use the duality theory to show (ii). First note that by construction, all the  constraints in (P) holds. By the weighted virtual value representation, the value of (P) given Symmetric Worst-Case Joint Distribution (II) and Symmetric Maxmin Public Good Mechanism (II) is simply $Pr(1,1)\times (1+1)=\frac{(r_2+1)^2}{2}$. Second, the constraints in (D) hold for all value profiles. To see this, note for any value profile $v=(v_1,v_2)$ outside the support of Symmetric Worst-Case Joint Distribution (II), since $\lambda_1=\lambda_2=\frac{1+r_2}{1-r_2}>0$, we have $$\lambda_1 v_1+\lambda_2 v_2 + \mu < \lambda_1 r_2+\lambda_2 \cdot 1 + \mu=0$$ For any value profile $v=(v_1,v_2)$ inside the support of Symmetric Worst-Case Joint Distribution (II), the constraint (the complementary slackness)  holds given \eqref{eq65}.  Finally, the value of (D) given the constructed $\lambda_1,\lambda_2, \mu$ is $\lambda_1 m+\lambda_2 m +\mu$, which, by some algebra, is equal to $\frac{(1+r_2)(2m-(1+r_2))}{1-r_2}=\frac{(1+r_2)^2}{2}$. By the linear programming duality theory, (ii) holds and the revenue guarantee is $\frac{(1+r_2)^2}{2}$.
\subsection{Proof of Lemma \ref{l2}}
To facilitate the exposition, we  define new functions $H^{I^*}_1(s_1, s_2)$ and $H^{I^*}_2(s_1, s_2)$
as follows. 
$$H_1^{I^*}(s_1,s_2):=  \left\{
\begin{array}{lll}
H_1^{I}(s_1,s_2)     &      & { s_1\neq s_2, s_1\neq 0, s_1\neq 0}\\
\frac{-2s_1\ln{s_1}+3s_1}{4}    &      & {s_1 =s_2\neq 0}\\
0    &      & {s_1=0\quad\text{or}\quad s_2=0}
\end{array} \right. $$
$$H_2^{I^*}(s_1,s_2):=  \left\{
\begin{array}{lll}
H_2^{I}(s_1,s_2)     &      & { s_1\neq s_2, s_1\neq 0, s_1\neq 0}\\
\frac{-2s_1\ln{s_1}+3s_1}{4}    &      & {s_1 =s_2\neq 0}\\
0    &      & {s_1=0\quad\text{or}\quad s_2=0}
\end{array} \right. $$
We start from establishing the following  claims regarding some properties of the functions $H_1^{I^*}(s_1,s_2)$ and $H_2^{I^*}(s_1,s_2)$, which will play a crucial role in establishing Lemma \ref{l2}.
\begin{claim}\label{cl1}
$H_1^{I^*}(s_1,s_2)$ and $H_2^{I^*}(s_1,s_2)$ are both continuous for $s_1\in[0,1]$ and $s_2\in[0,1]$.
\end{claim}
\begin{proof}[Proof of Claim \ref{cl1}]
We will first establish the continuity of $H_1^{I^*}(s_1,s_2)$. Note when $s_1\neq s_2, s_1\neq 0, s_2\neq 0$, the continuity holds as $H_1^I(s_1,s_2)$ is some analytic function. Therefore it suffices to show that $\lim_{s_2\to s_1\neq 0} H_1^I(s_1,s_2)=\frac{-2s_1\ln{s_1}+3s_1}{4}$, $\lim_{s_1\to 0}H_1^I(s_1,s_2)=0$ and  $\lim_{s_2\to 0}H_1^I(s_1,s_2)=0$. To see these, note 
\begin{equation*}
    \begin{split}
 \lim_{s_2\to s_1\neq 0} H_1^I(s_1,s_2)&=
 \lim_{s_2- s_1:=\epsilon \to 0} H_1^I(s_1,s_2)\\ &=\lim_{\epsilon \to 0} \frac{s_1(s_1+\epsilon)}{s_1+s_1+\epsilon}(\frac{s_1^2\ln{\frac{s_1}{s_1+\epsilon}}+\epsilon(\epsilon+s_1)}{\epsilon^2}-\ln{s_1})\\
 &= \lim_{\epsilon \to 0}\frac{s_1}{2}(\frac{-\frac{s_1^2}{s_1+\epsilon}+2\epsilon+s_1}{2\epsilon}-\ln{s_1})\\
 &= \lim_{\epsilon \to 0}\frac{s_1}{2}(\frac{\frac{s_1^2}{(s_1+\epsilon)^2}+2}{2}-\ln{s_1})\\
 &=\frac{-2s_1\ln{s_1}+3s_1}{4}
    \end{split}
\end{equation*}
The third equality and the fourth equality hold by the L'H\^{o}pital's Rule. Also
note that $\lim_{x\to 0}x\ln{x}=0$ and $\lim_{x\to 0}x^2\ln{x}=\lim_{x\to 0}x\cdot \lim_{x\to 0}x\ln{x}=0$, which imply that $\lim_{s_1\to 0}H_1^I(s_1,s_2)=0$ and  $\lim_{s_2\to 0}H_1^I(s_1,s_2)=0$. By symmetry, the continuity of $H_2^{I^*}(s_1,s_2)$ can be similarly established.
\end{proof}

\begin{claim}\label{cl2}
Fix any $s_2\in(0,1]$, $H_1^{I^*}(s_1,s_2)$ is strictly increasing w.r.t. $s_1$ for $s_1\le s_2$.
\end{claim}
\begin{proof}[Proof of Claim \ref{cl2}]
Taking first order derivative w.r.t. $s_1$ to $H_1^{I^*}$ when $s_1\neq s_2, s_1\neq 0, s_2\neq 0$, with some algebra, we obtain that 
\begin{equation*}
\frac{\partial H_1^{I^*}(s_1,s_2)}{\partial s_1 }=\frac{s_2^2}{(s_1-s_2)^2(s_1+s_2)^2}[-\frac{s_1^2(s_1+3s_2)}{s_1-s_2}\ln{\frac{s_1}{s_2}}-
(s_1-s_2)^2\ln{s_1}+s_1(3s_1+s_2)]
\end{equation*}
Then it suffices to show that for any $s_1\in(0,s_2)$,
\begin{equation}\label{eq111}
  -\frac{s_1^2(s_1+3s_2)}{s_1-s_2}\ln{\frac{s_1}{s_2}}-
(s_1-s_2)^2\ln{s_1}+s_1(3s_1+s_2)>0  
\end{equation}
Let $z_1:=\frac{s_1}{s_2}\in (0,1)$. Plugging $s_2=\frac{s_1}{z_1}$ into \eqref{eq111}, it suffices to show that for $z_1\in(0,1)$,
\begin{equation}\label{eq112}
\frac{z_1^2(3+z_1)}{1-z_1}\ln{z_1}-(z_1-1)^2\ln{s_1}+z_1(3z_1+1)>0
\end{equation}
Note that $-(z_1-1)^2\ln{s_1}>0$ and that $1-z_1>0$, then it suffices to show that for $z_1\in(0,1)$,
\begin{equation}\label{eq113}
h(z_1):=\ln{z_1}+\frac{z_1(1-z_1)(3z_1+1)}{z_1^2(3+z_1)}>0
\end{equation}
Now taking first order derivative to $h(z_1)$, with some algebra, we obtain that for $z_1\in (0,1)$,
\begin{equation}\label{eq114}
h'(z_1)=\frac{(z_1-3)(1-z_1)^2}{z_1^2(3+z_1)^2}<0
\end{equation}
Note that $h(1)=0$, then together with \eqref{eq114}, \eqref{eq113} holds.

\end{proof}
\begin{claim}\label{cl9}
$\frac{\partial H_1^{I^*}(s_1,s_2)}{\partial s_1 }\to \frac{1-6\ln{s_1}}{24}$ as $s_2\to s_1\neq 0$.
\end{claim}
\begin{proof}[Proof of Claim \ref{cl9}]
\begin{equation*}
    \begin{split}
        \lim_{s_2 \to s_1\neq 0}\frac{\partial H_1^{I^*}(s_1,s_2)}{\partial s_1 }&=\lim_{\epsilon:=s_2-s_1 \to 0}\frac{\partial H_1^{I^*}(s_1,s_2)}{\partial s_1 }\\
        &=\lim_{\epsilon \to 0}\frac{(s_1+\epsilon)^2}{(s_1+s_1+\epsilon)^2}\frac{-s_1^2(4s_1+3\epsilon)\ln{\frac{s_1}{s_1+\epsilon}}+\epsilon^3\ln{s_1}-\epsilon s_1(\epsilon+4s_1)}{-\epsilon^3}\\
        &=\lim_{\epsilon \to 0}\frac{-s_1^2(3\ln{\frac{s_1}{s_1+\epsilon}-\frac{4s_1+3\epsilon}{s_1+\epsilon}})+3\epsilon^2\ln{s_1}-s_1(4s_1+2\epsilon)}{-12\epsilon^2}
        \\
        &=\lim_{\epsilon \to 0}\frac{-s_1^2(-\frac{3}{s_1+\epsilon}+\frac{s_1}{(s_1+\epsilon)^2})+6\epsilon\ln{s_1}-2s_1}{-24\epsilon}\\
        &=\lim_{\epsilon \to 0}\frac{-s_1^2(\frac{3}{(s_1+\epsilon)^2}-\frac{2s_1}{(s_1+\epsilon)^3})+6\ln{s_1}}{-24}\\
        &= \frac{1-6\ln{s_1}}{24}
    \end{split}
\end{equation*}
where the third equality, the fourth equality and the fifth equality hold by the L'H\^{o}pital's Rule. 
\end{proof}
\begin{claim}\label{cl3}
Fix any $s_1\in(0,1)$, $H_1^{I^*}(s_1,s_2)$ is strictly increasing w.r.t. $s_2$ for $s_2\in (0,1)$.
\end{claim}
\begin{proof}[Proof of Claim \ref{cl3}]
Taking first order derivative w.r.t. $s_2$ to $H_1^{I^*}$ when $s_1\neq s_2, s_1\neq 0, s_2\neq 0$, with some algebra, we obtain that 
\begin{equation*}
\frac{\partial H_1^{I^*}(s_1,s_2)}{\partial s_2 }=\frac{s_2^2}{(s_1-s_2)^2(s_1+s_2)^2}[(\frac{2s_1s_2(s_1+s_2)}{s_1-s_2}+s_1^2)\ln{\frac{s_1}{s_2}}-
(s_1-s_2)^2\ln{s_1}-s_1(s_1+3s_2)]
\end{equation*}
Then it suffices to show that for any $s_1\in(0,1)$,
\begin{equation}\label{eq115}
  -\frac{s_1^2(s_1+3s_2)}{s_1-s_2}\ln{\frac{s_1}{s_2}}-
(s_1-s_2)^2\ln{s_1}+s_1(3s_1+s_2)>0  
\end{equation}
Let $z_2:=\frac{s_1}{s_2}\in (0,1)\cup (1,\infty)$. Plugging $s_2=\frac{s_1}{z_2}$ into \eqref{eq115}, it suffices to show that for $z_2\in(0,1)\cup (1,\infty)$,
\begin{equation}\label{eq116}
\frac{z_2(z_2^2+z_2+2)}{z_2-1}\ln{z_1}-(z_1-1)^2\ln{s_1}-z_2(z_2+3)>0
\end{equation}
Note that $-(z_2-1)^2\ln{s_1}>0$, then it suffices to show that for  $z_2\in(0,1)\cup (1,\infty)$,
\begin{equation}\label{eq117}
    \frac{z_2(z_2^2+z_2+2)}{z_2-1}\ln{z_1}-z_2(z_2+3)>0
\end{equation}
Note that $z_2>0$, then it suffices to show that for $z_2\in(1,\infty)$,
\begin{equation}\label{eq118}
i(z_2):=\ln{z_2}-\frac{(z_2+3)(z_2-1)}{z_2^2+z_2+2}>0
\end{equation}
and that for $z_2\in(0,1)$,
\begin{equation}\label{eq119}
i(z_2)<0
\end{equation}
Now taking first order derivative to $i(z_2)$, with some algebra, we obtain that  for $z_2\in (0,\infty)$,
\begin{equation}\label{eq120}
i'(z_2)=\frac{(z_2+1)(z_2+4)(1-z_2)^2}{z_2(z_2^2+z_2+2)^2}>0
\end{equation}
Note that $i(1)=0$, then together with \eqref{eq120}, \eqref{eq118} and \eqref{eq119} hold.
\end{proof}
\begin{claim}\label{cl10}
$\frac{\partial H_1^{I^*}(s_1,s_2)}{\partial s_2 }\to \frac{5-6\ln{s_1}}{24}$ as $s_2\to s_1\neq 0$.
\end{claim}
\begin{proof}[Proof of Claim \ref{cl10}]
\begin{equation*}
\begin{split}
\lim_{s_2 \to s_1\neq 0}\frac{\partial H_1^{I^*}(s_1,s_2)}{\partial s_2 }&=\lim_{\epsilon:=s_2-s_1 \to 0}\frac{\partial H_1^{I^*}(s_1,s_2)}{\partial s_2 }\\
&=\lim_{\epsilon \to 0}\frac{(s_1+\epsilon)^2}{(s_1+s_1+\epsilon)^2}\frac{(2s_1(s_1+\epsilon)(2s_1+\epsilon)-\epsilon s_1^2)\ln{\frac{s_1}{s_1+\epsilon}}+\epsilon^3\ln{s_1}+\epsilon s_1(4s_1+3\epsilon)}{-\epsilon^3}\\
&=\lim_{\epsilon \to 0}\frac{(2s_1^2(2\epsilon+3s_1)-s_1^2)\ln{\frac{s_1}{s_1+\epsilon}}-(2s_1(2s_1+\epsilon)-\frac{\epsilon s_1^2}{s_1+\epsilon})+3\epsilon^2\ln{s_1}+s_1(4s_1+6\epsilon)}{-12\epsilon^2}\\
&=\lim_{\epsilon \to 0}\frac{4s_1\ln{\frac{s_1}{s_1+\epsilon}}-\frac{s_1(4s_1+5\epsilon)}{s_1+\epsilon}-(2s_1-\frac{s_1^3}{(s_1+\epsilon)^2})+6\epsilon \ln{s_1}+6s_1}{-24\epsilon}\\
&=\lim_{\epsilon \to 0}\frac{-\frac{4s_1}{s_1+\epsilon}+\frac{s_1^2}{(s_1+\epsilon)^2}-\frac{2s_1^3}{(s_1+\epsilon)^3}+6\ln{s_1}}{-24} \\
&= \frac{5-6\ln{s_1}}{24}
\end{split}
\end{equation*}
where the third equality, the fourth equality and the fifth equality hold by the L'H\^{o}pital's Rule. 
\end{proof}
Now we are ready to prove Lemma \ref{l2}. Fix any $m_1\in (0,\frac{3}{4})$. Let $s_2^*(m_1)\in (0,1)$ denote the solution to $\frac{-2s_2\ln{s_2}+3s_2}{4}=m_1$. Let $s_1^*(m_1)\in (0,1)$ denote the solution to $\frac{s_1}{1+s_1}(\frac{2s_1-1}{(1-s_1)^2}\ln{s_1}+\frac{1}{1-s_1})=m_1$. Then by Claim \ref{cl1}, Claim \ref{cl2}, Claim \ref{cl3}, when $s_1\in [s_1^*(m_1), s_2^*(m_1)]$,  there exists a  strictly decreasing function $F^I$ such that $s_1=F^I(s_2)\le s_2$ is the unique solution to $H_1^{I^*}(s_1,s_2)=m_1$ for any $s_2\in[s_2^*(m_1),1]$. By Claim \ref{cl2}, \ref{cl9}, Claim \ref{cl3} and Claim \ref{cl10}, the continuous functions\footnote{When $s_1=s_2$, let $\frac{\partial H_1^{I^*}(s_1,s_2)}{\partial s_1 }=\frac{1-6\ln{s_1}}{24}$ and $\frac{\partial H_1^{I^*}(s_1,s_2)}{\partial s_2 }=\frac{5-6\ln{s_1}}{24}$, then by Claim \ref{cl9} and Claim \ref{cl10}, $\frac{\partial H_1^{I^*}(s_1,s_2)}{\partial s_1 }$ and $\frac{\partial H_1^{I^*}(s_1,s_2)}{\partial s_2}$ are continuous. } $\frac{\partial H_1^{I^*}(s_1,s_2)}{\partial s_1 }$ and $\frac{\partial H_1^{I^*}(s_1,s_2)}{\partial s_2 }$ are strictly positive and bounded on the compact set $[s_1^*(m_1), s_2^*(m_1)]\times [s_2^*(m_1),1]$. Then by the (Global) Implicit Function Theorem, $F^I(s_2)$ is continuous on $[s_2^*(m_1),1]$.  Plugging $s_1=F^I(s_2)$ to $H_2^{I^*}(s_1,s_2)$, we obtain $G^I(s_2):=H_2^{I^*}(F^I(s_2),s_2)$. Given the continuity of $H_2^{I^*}$ and $F^I$, we see that $G^I$ is also continuous at any $s_2\in[s_2^*(m_1),1]$. Note that when $s_2=s_2^*(m_1)$, $F^I(s_2)=s_2$ and therefore $G^I(s_2)=m_1$; when $s_2=1$, $G^I(s_2)=B_I(m_1)$. Then by the Intermediate Value Theorem, there exists $s_2\in[s_2^*(m_1),1]$ such that $G^I(s_2)=m_2$ for any $m_2\in(B_I(m_1), m_1)$.

\subsection{Proof of Theorem \ref{t3}}
(i): \textbf{Asymmetric Maxmin Public Good Mechanism (I)} is a best response to \textbf{Asymmetric Worst-Case Joint Distribution (I)}. The proof is similar to the proof of (i) of Theorem \ref{t1}.\\
(ii): \textbf{Asymmetric Worst-Case Joint Distribution (I)} is a best response to \textbf{Asymmetric Maxmin Public Good Mechanism (I)}: we use the duality theory to show (ii). First note that by construction, all the  constraints in (P) holds. By the weighted virtual value representation, the value of (P) given Asymmetric Worst-Case Joint Distribution (I) and Asymmetric Maxmin Public Good Mechanism (I) is simply $Pr(1,1)\times (1+1)=\frac{s_1s_2}{s_1+s_2}$. Second, the constraints in (D) hold for all value profiles. To see this, note for any value profile $v=(v_1,v_2)$ outside the support of Asymmetric Worst-Case Joint Distribution (I), since $\lambda_1>0$ and $\lambda_2=\frac{s_1}{s_2}\lambda_1$>0, we have $$\lambda_1 v_1+\lambda_2 v_2 + \mu < \lambda_1 s_1+\lambda_2 \cdot 0 + \mu=0$$ For any value profile $v=(v_1,v_2)$ inside the support of Asymmetric Worst-Case Joint Distribution (I), the constraints (the complementary slackness)  hold given \eqref{eq67},\eqref{eq79} and \eqref{eq92}.  In addition, the value of (D) given the constructed $\lambda_1,\lambda_2, \mu$ is $\lambda_1 m_1+\lambda_2 m_2 +\mu$, which, by some algebra, is equal to $\frac{c(1-\frac{s_2}{s_1})}{\ln{\frac{s_1}{s_2}}}(m_1+\frac{s_1}{s_2}m_2-s_1)=\frac{s_1s_2}{s_1+s_2}$. Finally, by Lemma \ref{l2}, the solution to \eqref{eq2} and \eqref{eq3} exists. By the linear programming duality theory, (ii) holds and the revenue guarantee is $\frac{s_1s_2}{s_1+s_2}$.
\subsection{Proof of Lemma \ref{l3}}
To facilitate the exposition, we  define new functions $H^{II^*}_1(t_1, t_2)$ and $H^{II^*}_2(t_1, t_2)$
as follows. 
$$H_1^{II^*}(t_1,t_2):=  \left\{
\begin{array}{lll}
H_1^{II}(t_1,t_2)     &      & { t_1\neq t_2}\\
\frac{-t_1^2+2t_1+3}{4}    &      & {t_1 =t_2}
\end{array} \right. $$
$$H_2^{II^*}(t_1,t_2):=  \left\{
\begin{array}{lll}
H_2^{II}(t_1,t_2)     &      & { t_1\neq t_2}\\
\frac{-t_1^2+2t_1+3}{4}    &      & {t_1 =t_2}
\end{array} \right. $$
We start from establishing the following  claims regarding some properties of the functions $H_1^{I^*}(s_1,s_2)$ and $H_2^{I^*}(s_1,s_2)$, which will play a crucial role in establishing Lemma \ref{l3}.
\begin{claim}\label{cl4}
$H_1^{II^*}(t_1,t_2)$ and $H_2^{II^*}(t_1,t_2)$ are both continuous for $t_1\in[0,1]$ and $t_2\in[0,1]$.
\end{claim}
\begin{proof}[Proof of Claim \ref{cl4}]
We will first establish the continuity of $H_1^{II^*}(t_1,t_2)$. Note when $t_1\neq t_2$, the continuity holds as $H_1^{II}(t_1,t_2)$ is some analytic function. Therefore it suffices to show that $\lim_{t_2\to t_1} H_1^{II}(t_1,t_2)=\frac{-t_1^2+2t_1+3}{4}$. To see this, note we have
\begin{equation*}
    \begin{split}
 \lim_{t_2\to t_1} H_1^{II}(t_1,t_2)&=
 \lim_{t_2- t_1:=\epsilon \to 0} H_1^{II}(t_1,t_2)\\ &=\lim_{\epsilon \to 0} \frac{(1+t_1)(1-t_1)^2(1+t_1+\epsilon)\ln{\frac{1+t_1+\epsilon}{1+t_1}}-\epsilon(1-t_1)(1-t_1(t_1+\epsilon))}{2\epsilon^2}+\frac{1+t_1}{2}\\
 &= \lim_{\epsilon \to 0}\frac{(1+t_1)(1-t_1)^2(1+\ln{\frac{1+t_1+\epsilon}{1+t_1}})-(1-t_1)(1-t_1^2-2t_1\epsilon)}{4\epsilon}+\frac{1+t_1}{2}\\
 &= \lim_{\epsilon \to 0}\frac{(1+t_1)(1-t_1)^2\frac{1}{1+t_1+\epsilon}+2t_1(1-t_1)}{4}+\frac{1+t_1}{2}\\
 &=\frac{-t_1^2+2t_1+3}{4}
    \end{split}
\end{equation*}
where the third equality and the fourth equality hold by the L'H\^{o}pital's Rule.
 By symmetry, the continuity of $H_2^{II^*}(t_1,t_2)$ can be similarly established.
\end{proof}
\begin{claim}\label{cl5}
Fix any $t_2\in[0,1]$, $H_1^{II^*}(t_1,t_2)$ is strictly increasing w.r.t. $t_1$ for $t_1\in (0,1)$.
\end{claim}
\begin{proof}[Proof of Claim \ref{cl5}]
Taking first order derivative w.r.t. $t_1$ to $H_1^{II^*}$ when $s_1\neq s_2$, with some algebra, we obtain that 
\begin{equation*}
\frac{\partial H_1^{II^*}(t_1,t_2)}{\partial t_1 }=\frac{(1-t_1)(1+t_2)}{2(t_1-t_2)^2}[(-1-3t_1-\frac{2(1+t_1)(1-t_1)}{t_1-t_2})\ln{\frac{1+t_2}{1+t_1}}+2(t_2-1)]
\end{equation*}
Then it suffices to show that for any $t_1\in(0,1)$,
\begin{equation}\label{eq121}
  (-1-3t_1-\frac{2(1+t_1)(1-t_1)}{t_1-t_2})\ln{\frac{1+t_2}{1+t_1}}+2(t_2-1)>0  
\end{equation}
Let $z_3:=\frac{1+t_2}{1+t_1}\in [\frac{1}{2},1)\cup(1,2]$. Plugging $t_2=z_3(1+t_1)-1$ into \eqref{eq121}, it suffices to show that for $z_3\in[\frac{1}{2},1)\cup(1,2]$,
\begin{equation}\label{eq122}
t_1(2z_3-(3-\frac{2}{1-z_3})\ln{z_3})+(-1-\frac{2}{1-z_3})\ln{z_3}+2z_3-4>0
\end{equation}
Then it suffices to show that for $z_3\in[\frac{1}{2},1)\cup(1,2]$,
\begin{equation}\label{eq123}
2z_3-(3-\frac{2}{1-z_3})\ln{z_3}>0
\end{equation}
\begin{equation}\label{eq124}
(-1-\frac{2}{1-z_3})\ln{z_3}+2z_3-4>0
\end{equation}
To show \eqref{eq123}, let $j(z_3):=-\ln{z_3}+\frac{2z_3(1-z_3)}{1-3z_3}$. Taking first order derivative, with some algebra, we obtain that
\begin{equation}\label{eq125}
    j'(z_3)=\frac{(6z_3-1)(1-z_3)^2}{z_3(1-3z_3)^2}
\end{equation}
Observe that \label{eq125} implies that $j(z_3)$ is increasing for $z_3\ge \frac{1}{2}$. Also note $j(1)=0$. Then $j(z_3)>0$ for $z_3> 1$ and $j(z_3)<0$ for $z_3\in [\frac{1}{2},1)$. Therefore \eqref{eq123} holds. To show \eqref{eq124}, let $k(z_3):=\ln{z_3}+\frac{(1-z_3)(2z_3-4)}{z_3-3}$.
Taking first order derivative to $k(z_3)$, with some algebra, we obtain that
\begin{equation}\label{eq126}
k'(z_3)=\frac{(9-2z_3)(1-z_3)^2}{z_3(z_3-3)^2}
\end{equation}
Observe that \label{eq126} implies that $k(z_3)$ is increasing for $z_3\le 2$. Also note $k(1)=0$. Then $k(z_3)>0$ for $2\ge z_3> 1$ and $j(z_3)<0$ for $z_3\in [\frac{1}{2},1)$. Therefore \eqref{eq124} holds.

\end{proof}
\begin{claim}\label{cl11}
$\frac{\partial H_1^{II^*}(t_1,t_2)}{\partial t_1 }\to \frac{(1-t_1)(5t_1+7)}{12(1+t_1)}$ as $t_2\to t_1$.
\end{claim}
\begin{proof}[Proof of Claim \ref{cl11}]
\begin{equation*}
    \begin{split}
        \lim_{t_2\to t_1}\frac{\partial H_1^{II^*}(t_1,t_2)}{\partial t_1 }&=
         \lim_{\epsilon:=t_2- t_1\to 0}\frac{\partial H_1^{II^*}(t_1,t_2)}{\partial t_1 }\\
         &= \lim_{\epsilon:=t_2- t_1\to 0}\frac{(1-t_1)(1+t_1+\epsilon)}{2}\frac{((1+3t_1)\epsilon-2(1-t_1^2))\ln{\frac{1+t_1+\epsilon}{1+t_1}}-2\epsilon(t_1-1+\epsilon)}{-\epsilon^3}\\
         &=\lim_{\epsilon\to 0}\frac{(1-t_1)(1+t_1)((1+3t_1)\ln{\frac{1+t_1+\epsilon}{1+t_1}}+\frac{(1+3t_1)\epsilon-2(1-t_1^2)}{1+t_1+\epsilon}-2(t_1-1)-4\epsilon)}{-6\epsilon^2}\\&=\lim_{\epsilon\to 0}\frac{(1-t_1)(1+t_1)(\frac{1+3t_1}{1+t_1+\epsilon}+\frac{(1+t_1)(3+t_1)}{(1+t_1+\epsilon)^2}-4)}{-12\epsilon}\\&= \lim_{\epsilon\to 0}\frac{(1-t_1)(1+t_1)(-\frac{1+3t_1}{(1+t_1+\epsilon)^2}-\frac{2(1+t_1)(3+t_1)}{(1+t_1+\epsilon)^3})}{-12}\\&=
          \frac{(1-t_1)(5t_1+7)}{12(1+t_1)}
    \end{split}
\end{equation*}
where the third equality, the fourth equality and the fifth equality hold by the L'H\^{o}pital's Rule.
\end{proof}
\begin{claim}\label{cl6}
Fix any $t_1\in[0,1]$, $H_1^{II^*}(t_1,t_2)$ is strictly decreasing w.r.t. $t_2$ for $t_2\in (0,1)$.
\end{claim}
\begin{proof}[Proof of Claim \ref{cl6}]
Taking first order derivative w.r.t. $t_2$ to $H_1^{II^*}$ when $t_1\neq t_2$, with some algebra, we obtain that 
\begin{equation*}
\frac{\partial H_1^{II^*}(t_1,t_2)}{\partial t_2 }=\frac{(1+t_1)(1-t_1)^2}{2(t_1-t_2)^2}[(1+\frac{2(1+t_2)}{t_1-t_2})\ln{\frac{1+t_2}{1+t_1}}+2]
\end{equation*}
Then it suffices to show that for any $t_1\in[0,1]$ and $t_2\neq t_1$,
\begin{equation}\label{eq127}
  (1+\frac{2(1+t_2)}{t_1-t_2})\ln{\frac{1+t_2}{1+t_1}}+2<0  
\end{equation}
Plugging $t_2=z_3(1+t_1)-1$ into \eqref{eq127}, it suffices to show that for $z_3\in[\frac{1}{2},1)\cup(1,2]$,
\begin{equation}\label{eq128}
(1+\frac{2z_3}{1-z_3})\ln{z_3}+2<0
\end{equation}
To show \eqref{eq128}, let $l(z_3):=\ln{z_3}+\frac{2(1-z_3)}{1+z_3}$. Taking first order derivative, with some algebra, we obtain that
\begin{equation}\label{eq129}
    l'(z_3)=\frac{1}{z_3}+\frac{4}{(1+z_3)^2}>0
\end{equation}
 Also note $l(1)=0$. Then $l(z_3)>0$ for $z_3> 1$ and $l(z_3)<0$ for $z_3\in [\frac{1}{2},1)$. Therefore \eqref{eq128} holds.
\end{proof}
\begin{claim}\label{cl12}
$\frac{\partial H_1^{II^*}(t_1,t_2)}{\partial t_2 }\to -\frac{(1-t_1)^2}{12(1+t_1)}$ as $t_2\to t_1$.
\end{claim}
\begin{proof}[Proof of Claim \ref{cl12}]
\begin{equation*}
    \begin{split}
        \lim_{t_2\to t_1}\frac{\partial H_1^{II^*}(t_1,t_2)}{\partial t_2 }&=
         \lim_{\epsilon:=t_2- t_1\to 0}\frac{\partial H_1^{II^*}(t_1,t_2)}{\partial t_2 }\\
         &= \lim_{\epsilon:=t_2- t_1\to 0}\frac{(1+t_1)(1-t_1)^2}{2}\frac{(2(1+t_1)+\epsilon)\ln{\frac{1+t_1+\epsilon}{1+t_1}}-2\epsilon}{-\epsilon^3}\\
         &=\lim_{\epsilon\to 0}\frac{(1+t_1)(1-t_1)^2(\frac{2(1+t_1)+\epsilon}{1+t_1+\epsilon}+\ln{\frac{1+t_1+\epsilon}{1+t_1}}-2)}{-6\epsilon^2}\\&=\lim_{\epsilon\to 0}\frac{(1+t_1)(1-t_1)^2(-\frac{1+t_1}{(1+t_1+\epsilon)^2}+\frac{1}{1+t_1+\epsilon})}{-12\epsilon}\\&= \lim_{\epsilon\to 0}\frac{(1+t_1)(1-t_1)^2(\frac{2(1+t_1)}{(1+t_1+\epsilon)^3}-\frac{1}{(1+t_1+\epsilon)^2})}{-12}\\&=
         -\frac{(1-t_1)^2}{12(1+t_1)}
    \end{split}
\end{equation*}
where the third equality, the fourth equality and the fifth equality hold by the L'H\^{o}pital's Rule.
\end{proof}
Now we are ready to prove Lemma \ref{l3}. Fix any $m_1\in (\frac{3}{4},1)$. Let $t_2^*(m_1)\in (0,1)$ denote the solution to $\frac{-t_2^2+2t_2+3}{4}=m_1$. Let $t_1^*(m_1)\in (0,1)$ denote the solution to $ \frac{(1+t_1)(1-t_1^2)}{2t_1^2}\ln{\frac{1}{1+t_1}}+\frac{1+t_1^2}{2t_1}=m_1$.  Then by Claim \ref{cl4}, Claim \ref{cl5} and  Claim \ref{cl6}, there exists a  strictly increasing function $F^{II}$ such that $t_1=F^{II}(t_2)\ge t_2$ is the solution to $H_1^{II^*}(t_1,t_2)=m_1$ for any $t_2\in[0,t_2^*(m_1)]$. In addition,  $F^{II}(t_2) \in [t_1^*(m_1), t_2^*(m_1)]$.  By Claim \ref{cl5}, Claim \ref{cl11}, Claim \ref{cl6} and Claim \ref{cl12}, the continuous function\footnote{When $t_1=t_2$, let $\frac{\partial H_1^{II^*}(t_1,t_2)}{\partial t_1 }=\frac{(1-t_1)(5t_1+7)}{12(1+t_1)}$ and $\frac{\partial H_1^{II^*}(t_1,t_2)}{\partial t_2 }=-\frac{(1-t_1)^2}{12(1+t_1)}$, the by Claim \ref{cl11} and Claim \ref{cl12}, $\frac{\partial H_1^{II^*}(t_1,t_2)}{\partial t_1 }$ and $\frac{\partial H_1^{II^*}(t_1,t_2)}{\partial t_2 }$ are continuous. } $\frac{\partial H_1^{II^*}(t_1,t_2)}{\partial t_1 }$ (and $\frac{\partial H_1^{II^*}(t_1,t_2)}{\partial t_2 }$) is strictly positive (and strictly negative) and bounded on the compact set $[t_1^*(m_1), t_2^*(m_1)]\times [0,t_2^*(m_1)]$. Then by the (Global) Implicit Function Theorem, $F^{II}(t_2)$ is continuous on $[0,t_2^*(m_1)]$.  Plugging $t_1=F^{II}(t_2)$ to $H_2^{II^*}(t_1,t_2)$, we obtain $G^{II}(t_2):=H_2^{II^*}(F^{II}(t_2),t_2)$. Given the continuity of $H_2^{II^*}$ and $F^{II}$, we see that $G^{II}$ is also continuous at any $t_2\in[0,t_2^*(m_1)]$. Note that when $t_2=t_2^*(m_1)$, $F^{II}(t_2)=t_2$ and therefore $G^{II}(t_2)=m_1$; when $s_2=0$, $G^{II}(t_2)=B_{II}(m_1)$. Then by the Intermediate Value Theorem, there exists $t_2\in[0,t_2^*(m_1)]]$ such that $G(s_2)=m_2$ for any $m_2\in[B_{II}(m_1), m_1)$. Finally, when $m_1=1$, then $t_1=1$ and \eqref{eq5} becomes $(1+t_2)\ln{\frac{2}{1+t_2}}+t_2=m_2$. Note when $t_2=0$, then L.H.S.$=\ln{2}$; when $t_2=1$, then L.H.S. $=1$. Therefore by the Intermediate Value Theorem, there exists $t_2\in [0,1)$ that solves \eqref{eq5} for any $m_2\in [B_{II}(1), 1)$. 

\subsection{Proof of Theorem \ref{t4}}
(i): \textbf{Asymmetric Maxmin Public Good Mechanism (II)} is a best response to \textbf{Asymmetric Worst-Case Joint Distribution (II)}.The proof is similar to the proof of (i) of Theorem \ref{t1}.\\
(ii): \textbf{Asymmetric Worst-Case Joint Distribution (II)} is a best response to \textbf{Asymmetric Maxmin Public Good Mechanism (II)}: we use the duality theory to show (ii). First note that by construction, all the  constraints in (P) holds. By the weighted virtual value representation, the value of (P) given Asymmetric Worst-Case Joint Distribution (II) and Asymmetric Maxmin Public Good Mechanism (II) is simply $Pr(1,1)\times (1+1)=\frac{(1+t_1)(1+t_2)}{2}$. Second, the constraints in (D) hold for all value profiles. To see this, note for any value profile $v=(v_1,v_2)$ outside the support of Asymmetric Worst-Case Joint Distribution (II), since $\lambda_1>0$ and $\lambda_2>0$, we have $$\lambda_1 v_1+\lambda_2 v_2 + \mu < \lambda_1 t_1+\lambda_2 \cdot 1 + \mu=0$$ For any value profile $v=(v_1,v_2)$ inside the support of Asymmetric Worst-Case Joint Distribution (II), the constraints (the complementary slackness)  hold given \eqref{eq103}.  In addition, the value of (D) given the constructed $\lambda_1,\lambda_2, \mu$ is $\lambda_1 m_1+\lambda_2 m_2 +\mu$, which, by some algebra, is equal to $\frac{(1+t_1)(1+t_2)}{2}$. Finally, by Lemma \ref{l3}, the solution to \eqref{eq4} and \eqref{eq5} exists. By the linear programming duality theory, (ii) holds and the revenue guarantee is $\frac{(1+t_1)(1+t_2)}{2}$.
\subsection{Proof of Lemma \ref{l4}}
We start from establishing the following  claims regarding some properties of the function $H_2^{III}(u_1,u_2)$, which will play a crucial role in establishing Lemma \ref{l4}.

\begin{claim}\label{cl7}
Fix any $u_1\in(0,1]$, $H_2^{III}(u_1,u_2)$ is strictly decreasing w.r.t. $u_2$ for $u_1\in (0,1)$.
\end{claim}
\begin{proof}[Proof of Claim \ref{cl7}]
Taking first order derivative w.r.t. $u_2$ to $H_2^{III}$, with some algebra, we obtain that 
\begin{equation*}
\frac{\partial H_2^{III}(u_1,u_2)}{\partial u_2 }=\frac{u_1}{(1+u_1)(u_2-u_1+1)^2}[-
\frac{1+u_2+u_1}{1+u_2-u_1}\ln{\frac{1+u_2}{u_1}}+2]
\end{equation*}
Then it suffices to show that for any $u_2\in(0,1)$,
\begin{equation}\label{eq130}
-\frac{1+u_2+u_1}{1+u_2-u_1}\ln{\frac{1+u_2}{u_1}}+2<0  
\end{equation}
Let $z_4:=\frac{1+u_2}{u_1}\in (1,\infty)$. Plugging $u_2=z_4u_1-1$ into \eqref{eq130}, it suffices to show that for $z_4\in(1,\infty)$,
\begin{equation}\label{eq131}
-\frac{z_4+1}{z_4-1}\ln{z_4}+2 <0
\end{equation}
To show \eqref{eq131}, let $n(z_4):=-\ln{z_4}+\frac{2(z_4-1)}{z_4+1}$. Taking first order derivative, with some algebra, we obtain that
\begin{equation}\label{eq132}
    n'(z_4)=-\frac{(z_4-1)^2}{z_4(z_4+1)^2}
\end{equation}
Observe that \label{eq132} implies that $n(z_4)$ is strictly decreasing for $z_4\ge 1$. Also note $n(1)=0$. Then $n(z_3)<0$ for $z_4> 1$. Therefore \eqref{eq131} holds. 

\end{proof}
\begin{claim}\label{cl13}
$\frac{\partial H_2^{III}(1,u_2)}{\partial u_1}\to -\frac{1}{12}$ as $u_2\to 0$.

\end{claim}
\begin{proof}[Proof of Claim \ref{cl13}]
\begin{equation*}
    \begin{split}
        \lim_{u_2\to 0}\frac{\partial H_2^{III}(1,u_2)}{\partial u_2}&=\lim_{u_2\to 0}\frac{-(2+u_2)\ln{(1+u_2)}+2u_2}{2u_2^3}\\&=\lim_{u_2\to 0}\frac{-\frac{1}{1+u_2}-\ln{(1+u_2)}+1}{6u_2^2}\\&=\lim_{u_2\to 0}\frac{\frac{1}{(1+u_2)^2}-\frac{1}{1+u_2}}{12u_2}\\&=\lim_{u_2\to 0}\frac{-\frac{2}{(1+u_2)^3}+\frac{1}{(1+u_2)^2}}{12}\\&=-\frac{1}{12}
    \end{split}
\end{equation*}
where the second equality, the third equality and the fourth equality hold by the L'H\^{o}pital's Rule.
\end{proof}
\begin{claim}\label{cl8}
Fix any $u_2\in[0,1]$, $H_2^{III}(u_1,u_2)$ is strictly increasing w.r.t. $u_1$ for $u_1\in (0,1)$.
\end{claim}
\begin{proof}[Proof of Claim \ref{cl8}]
Taking first order derivative w.r.t. $u_1$ to $H_2^{III}$, with some algebra, we obtain that 
\begin{equation*}
\begin{split}
\frac{\partial H_2^{III}(u_1,u_2)}{\partial u_1 }=\frac{1}{(1+u_1)^2(1+u_2-u_1)^2}[
(u_2+1)(1+\frac{2u_1(1+u_1)}{u_2-u_1+1})\ln{\frac{1+u_2}{u_1}}\\+(1+u_2-u_1)^2-(1+u_2-u_1)-(1+u_1)(u_2+1+u_1)]
\end{split}
\end{equation*}
Then it suffices to show that for any $u_1\in (0,1)$,
\begin{equation}\label{eq133}
 (u_2+1)(1+\frac{2u_1(1+u_1)}{u_2-u_1+1})\ln{\frac{1+u_2}{u_1}}+(1+u_2-u_1)^2-(1+u_2-u_1)-(1+u_1)(u_2+1+u_1)>0  
\end{equation}
Plugging $u_2=z_4u_1-1$ into \eqref{eq133}, with some algebra,  it suffices to show that for $z_4\in(1,\infty)$,
\begin{equation}\label{eq134}
u_1^2(\frac{2z_4\ln{z_4}}{z_4-1}+(z_4-1)^2+z_4+1)+u_1[(z_4+\frac{2z_4}{z_4-1})\ln{z_4}+2]>0
\end{equation}
Then it suffices to show that for $z_4\in(1,\infty)$,
\begin{equation}\label{eq135}
    \frac{2z_4\ln{z_4}}{z_4-1}+(z_4-1)^2+z_4+1>0
\end{equation}
\begin{equation}\label{eq136}
    (z_4+\frac{2z_4}{z_4-1})\ln{z_4}+2>0
\end{equation}
Note both \eqref{eq135} and \eqref{eq136} hold trivially when $z_4>1$.
\end{proof}
\begin{claim}\label{cl14}
$\frac{\partial H_2^{III}(1,u_2)}{\partial u_2}\to \frac{5}{24}$ as $u_2\to 0$.

\end{claim}
\begin{proof}[Proof of Claim \ref{cl14}]
\begin{equation*}
    \begin{split}
        \lim_{u_2\to 0}\frac{\partial H_2^{III}(1,u_2)}{\partial u_1}&=\lim_{u_2\to 0}\frac{(u_2+1)(u_2+4)\ln{(1+u_2)}+u_2(u_2^2-3u_2-4)}{4u_2^3}\\&=\lim_{u_2\to 0}\frac{(2u_2+5)\ln{(1+u_2)}+3u_2^2-5u_2}{12u_2^2}\\&=\lim_{u_2\to 0}\frac{\frac{3}{1+u_2}+2\ln{(1+u_2)}+6u_2-3}{24u_2}\\&=\lim_{u_2\to 0}\frac{-\frac{3}{(1+u_2)^3}+\frac{2}{1+u_2}+6}{24}\\&=\frac{5}{24}
    \end{split}
\end{equation*}
where the second equality, the third equality and the fourth equality hold by the L'H\^{o}pital's Rule.
\end{proof}
Now we are ready to prove Lemma \ref{l4}. To facilitate the analysis, we rewrite $Area (III)$ as $Area (III)=\{(m_1,m_2)|m_1\ge B_I^{-1}(m_2)), m_1< B_{II}^{-1}(m_2), m_1<B^{-1}_{III}(m_2),  0<m_2<\frac{3}{4}\}$. This can be done since $B_{I}, B_{II}$ and $B_{III}$ are all strictly monotone functions, and therefore the inverse functions exist.  Next, fix any $m_2\in (0,\ln{2}]$. Let $u_1^*(m_2)\in (0,1]$ denote the solution to $u_1\ln{\frac{u_1+1}{u_1}}=m_2$. Let $u_1^{**}(m_2)$  denote the solution to $\frac{u_1}{1+u_1}(-\frac{1}{(1-u_1)^2}\ln{u_1}-\frac{u_1}{1-u_1})=m_2$. Then by Claim \ref{cl7} and Claim \ref{cl8}, there exists a strictly increasing function $F_1^{III}$ such that $u_2=F_1^{III}(u_1)\le u_1$ is the unique solution to $H_2^{III}(u_1,u_2)=m_2$ for any $u_1\in[u_1^{**}(m_2),u_1^*(m_2)]$. In addition, $F_1^{III}\in [0, u_1^*(m_2)]$. Note the continuous functions $\frac{\partial H_2^{III}(u_1,u_2)}{\partial u_1 }$ (and $\frac{\partial H_2^{III}(u_1,u_2)}{\partial u_2 }$) is strictly negative (and strictly positive) and bounded on the compact set $[u_1^{**}(m_2),u_1^*(m_2)]\times [0, u_1^*(m_2)]$. Then by the (Global) Implicit Function Theorem, $F_1^{III}(u_1)$ is continuous on $[u_1^{**}(m_2),u_1^*(m_2)]$.   Plugging $u_2=F_1^{III}(u_1)$ to $H_1^{III}(u_1,u_2)$, we obtain $G_1^{III}(u_1):=H_1^{III}(u_1,F_1^{III}(u_1))$. Given the continuity of $H_1^{III}$ and $F_1^{III}$, we see that $G_1^{III}$ is also continuous at any $u_1\in[u_1^{**}(m_2),u_1^*(m_2)]$. Note that when $u_1=u_1^*(m_2)$, $F_1^{III}(u_1)=u_1$ and therefore $G^{III}(u_1)=B^{-1}_{III}(m_2)$; when $u_1=u_1^{**}(m_2)$, $F_1^{III}(u_1)=0$ and therefore $G_1^{III}(u_1)=B^{-1}_{I}(m_2)$. Then by the Intermediate Value Theorem, there exists $u_1\in[u_1^{**}(m_2),u_1^*(m_2))$ such that $G_1^{III}(u_1)=m_1$ for any $m_1\in[B^{-1}_{I}(m_2), B^{-1}_{III}(m_2))$. Then fix any $m_2 \in (\ln{2},\frac{3}{4})$. Let $u_2^*(m_2)\in (0,1)$ denote the solution to $\frac{1+u_2}{2}(\frac{\ln{1+u_2}}{u_2^2}+\frac{u_2-1}{u_2(1+u_2)})$.  Then by Claim \ref{cl7} and Claim \ref{cl8}, there exists a strictly increasing function $F_2^{III}$ such that $u_2=F_2^{III}(u_1)\le u_1$ is the unique solution to $H_2^{III}(u_1,u_2)=m_2$ for any $u_1\in[u_1^{**}(m_2),1]$. In addition, $F_2^{III}\in [0, u_2^*(m_2)]$. By Claim \ref{cl7},  Claim \ref{cl13},  Claim \ref{cl8} and  Claim \ref{cl14},  continuous function\footnote{When $(u_1,u_2)=(1,0)$, let
$\frac{\partial H_2^{III}(u_1,u_2)}{\partial u_2 }=-\frac{1}{12}$ and $\frac{\partial H_2^{III}(u_1,u_2)}{\partial u_1 }=\frac{5}{24}$, then by Claim \ref{cl13} and Claim \ref{cl14}, $\frac{\partial H_2^{III}(u_1,u_2)}{\partial u_2 }$ and $\frac{\partial H_2^{III}(u_1,u_2)}{\partial u_1 }$ are continuous. } $\frac{\partial H_2^{III}(u_1,u_2)}{\partial u_2 }$ (and $\frac{\partial H_2^{III}(u_1,u_2)}{\partial u_1 }$) is strictly negative (and strictly positive) and bounded on the compact set $[u_1^{**}(m_2),1]\times [0,u_2^*(m_2)]$. Then by the (Global) Implicit Function Theorem, $F_1^{III}(u_1)$ is continuous on $[u_1^{**}(m_2),1]$.  Plugging $u_2=F_2^{III}(u_1)$ to $H_1^{III}(u_1,u_2)$, we obtain $G_2^{III}(u_1):=H_1^{III}(u_1,F_2^{III}(u_1))$. Given the continuity of $H_1^{III}$ and $F_2^{III}$, we see that $G_2^{III}$ is also continuous at any $u_1\in[u_1^{**}(m_2),1]$. Note that when $u_1=1$,  $G^{III}(u_1)=B^{-1}_{II}(m_2)$; when $u_1=u_1^{**}(m_2)$, $F_1^{III}(u_1)=0$ and therefore $G_1^{III}(u_1)=B^{-1}_{I}(m_2)$. Then by the Intermediate Value Theorem, there exists $u_1\in[u_1^{**}(m_2),1)$ such that $G_2^{III}(u_1)=m_1$ for any $m_1\in[B^{-1}_{I}(m_2), B^{-1}_{II}(m_2))$.
\subsection{Proof of Theorem \ref{t5}}
(i): \textbf{Asymmetric Maxmin Public Good Mechanism (III)} is a best response to \textbf{Asymmetric Worst-Case Joint Distribution (III)}. The proof is similar to the proof of (i) of Theorem \ref{t1}.\\
(ii): \textbf{Asymmetric Worst-Case Joint Distribution (III)} is a best response to \textbf{Asymmetric Maxmin Public Good Mechanism (III)}: we use the duality theory to show (ii). First note that by construction, all the  constraints in (P) holds. By the weighted virtual value representation, the value of (P) given Asymmetric Worst-Case Joint Distribution (III) and Asymmetric Maxmin Public Good Mechanism (III) is simply $Pr(1,1)\times (1+1)=\frac{u_1(1+u_2)}{u_1+1}$. Second, the constraints in (D) hold for all value profiles. To see this, note for any value profile $v=(v_1,v_2)$ outside the support of Asymmetric Worst-Case Joint Distribution (III), since $\lambda_1>0$ and $\lambda_2>0$, we have $$\lambda_1 v_1+\lambda_2 v_2 + \mu < \lambda_1 u_1+\lambda_2 \cdot 0 + \mu=0$$ For any value profile $v=(v_1,v_2)$ inside the support of Asymmetric Worst-Case Joint Distribution (III), the constraints (the complementary slackness)  hold given \eqref{eq104}.  In addition, the value of (D) given the constructed $\lambda_1,\lambda_2, \mu$ is $\lambda_1 m_1+\lambda_2 m_2 +\mu$, which, by some algebra, is equal to $\frac{u_1(1+u_2)}{u_1+1}$. Finally, by Lemma \ref{l4}, the solution to \eqref{eq6} and \eqref{eq7} exists. By the linear programming duality theory, (ii) holds and the revenue guarantee is $\frac{u_1(1+u_2)}{u_1+1}$.
\subsection{Proof of Lemma \ref{l5}}
Fix any $m_1\in [0,1]$. Let $w_1^*(m_1)\in [0,1]$ denote the solution to \eqref{eq8}. Note that $H^{IV}(w_1,w_2)$ is continuous and strict increasing w.r.t. $w_2$. Also note that $H^{IV}(w_1^*(m_1),0)=0$ and $H^{IV}(w_1^*(m_1),1)=B_{III}(m_1)$.  Then by the Intermediate Value Theorem, there exists (unique) $w_2\in[0,1]$ such that $H^{IV}(w_1^*(m_1),w_2)=m_2$ for any $m_2\in[0, B_{III}(m_1)]$.
\subsection{Proof of Theorem \ref{t6}}
(i): \textbf{Asymmetric Maxmin Public Good Mechanism (IV)} is a best response to \textbf{Asymmetric Worst-Case Joint Distribution (IV)}. The proof is similar to the proof of (i) of Theorem \ref{t1}.\\
(ii): \textbf{Asymmetric Worst-Case Joint Distribution (IV)} is a best response to \textbf{Asymmetric Maxmin Public Good Mechanism (IV)}: we use the duality theory to show (ii). First note that by construction, all the  constraints in (P) holds. By the weighted virtual value representation, the value of (P) given Asymmetric Worst-Case Joint Distribution (IV) and Asymmetric Maxmin Public Good Mechanism (IV) is simply $Pr(1,w_2)\times (1+w_2)=w_1$. Second, the constraints in (D) hold for all value profiles. To see this, note for any value profile $v=(v_1,v_2)$ in which $v_1<w_1$, since $\lambda_1=-\frac{1}{\ln{w_1}}>0$ and $\lambda_2=0$, we have $$\lambda_1 v_1+\lambda_2 v_2 + \mu < \lambda_1 w_1+\lambda_2 \cdot 0 + \mu=0$$ For any value profile $v=(v_1,v_2)$ in which $v_1< w_1$, the constraints (the complementary slackness)  hold given \eqref{eq105}.  Finally, the value of (D) given the constructed $\lambda_1,\lambda_2, \mu$ is $\lambda_1 m_1+\lambda_2 m_2 +\mu$, which, by some algebra, is equal to $w_1$. By the linear programming duality theory, (ii) holds and the revenue guarantee is $w_1$.

\section{Appendix C}
\subsection{Proof of Theorem \ref{t8}}
(i): \textbf{$N$-agent Symmetric  Maxmin Public Good Mechanism} is a best response to \textbf{$N$-agent Symmetric Worst-Case Joint Distribution}.  First we verify that $N$-agent Symmetric Worst-Case Joint Distribution is a legal joint distribution, i.e., its integral is 1 over $[0,1]^N$. Note it is symmetric, then it suffices to derive the marginal distribution for agent 1 and show its integral is 1 over [0,1]. When $v_1=1$, with some algebra,  $Pr(v_1=1)=\frac{r+N-1}{N}$; when $r\le v_1<1$, with some algebra,  $f(v_1)=\frac{1}{N^{N-1}}\sum_{j=0}^{N-2}{N-2 \choose j}(N-1-j)(r+N-1)^j(v_1-r)^{N-2-j}$. Therefore the integral is
\begin{equation*}
    \begin{split}
       \int_r^1f(v_1)dv_1+Pr(v_1=1)&= \frac{1}{N^{N-1}}\sum_{j=0}^{N-2}{N-2 \choose j}(r+N-1)^j(1-r)^{N-1-j}+\frac{r+N-1}{N}\\
       &=\frac{1-r}{N^{N-1}}\sum_{j=0}^{N-2}{N-2 \choose j}(r+N-1)^j(1-r)^{N-2-j}+\frac{r+N-1}{N}\\
       &= \frac{1-r}{N^{N-1}}(r+N-1+1-r)^{N-2}+\frac{r+N-1}{N}\\
        &=1
    \end{split}
\end{equation*}
Next we show the weighted virtual value is 0 for any  $v\in SR_N$ and $v\neq (\underbrace{1,\cdots, 1}_{N})$. We discuss two cases: a). $|A(v)|=1$. Without loss, we can assume $v_1\neq 1$ and $v_j=1$ for any $j\neq 1$. Then
\begin{equation*}
\begin{split}
    \Phi(v_1,\underbrace{1,\cdots, 1}_{N-1})&=(v_1+N-1)\pi^*(v_1,\underbrace{1,\cdots, 1}_{N-1})-\int_{v_1}^1\pi^*(x,\underbrace{1,\cdots, 1}_{N-1})dx-Pr^*(\underbrace{1,\cdots, 1}_{N})\\
  &= (v_1+N-1)\cdot \frac{(r+N-1)^N}{N^{N-1}(v_1+N-1)^2}-\int_{v_1}^1\frac{(r+N-1)^N}{N^{N-1}(x+N-1)^2}dx-\frac{(r+N-1)^N}{N^N}\\
  &=0
\end{split}
\end{equation*}
b). $|A(v)|>1$. 
\begin{equation*}
\begin{split}
    \Phi(v)&=(\sum_{i=1}^Nv_i)\pi^*(v)-\sum_{j\in A(v)}\int_{v_j}^1\pi^*(x,v_{-j})dx-\sum_{j\in A(v)}\pi^*(1,v_{-j})\\
  &= (\sum_{j\in A(v)}v_j+N-|A(v)|)\cdot \frac{|A(v)|!(r+N-1)^N}{N^{N-1}(\sum_{j\in A(v)}v_j+N-|A(v)|)^{|A(v)|+1}}\\&-\sum_{j\in A(v)}\int_{v_j}^1\frac{|A(v)|!(r+N-1)^N}{N^{N-1}(x+\sum_{i\in A(v),i\neq j }v_i+N-|A(v)|)^{|A(v)|+1}}dx\\&-\sum_{j\in A(v)}\frac{(|A(v)|-1)!(r+N-1)^N}{N^{N-1}(\sum_{i\in A(v), i\neq j}v_j+N+1-|A(v)|)^{|A(v)|}}\\
  &=0
\end{split}
\end{equation*}
Therefore,   $N$-agent Symmetric Worst-Case Joint Distribution exhibits the property that the weighted virtual value is positive only for the highest type $(\underbrace{1,\cdots, 1}_{N})$, zero for the other value profiles in the support and weakly negative for value profiles outside the support.  Then any feasible and monotone mechanism in which the public good is provided with some positive probability if and only if $\sum_{i=1}^Nv_i> r+N-1$ and the public good is provided with probability 1 when $v  = (\underbrace{1,\cdots, 1}_{N})$ is a best response for the 
principal. It is easy to see that $N$-agent  Maxmin Public Good Mechanism is such a mechanism.\\
(ii): \textbf{$N$-agent Symmetric Worst-Case Joint Distribution} is a best response to \textbf{$N$-agent Symmetric Maxmin Public Good Mechanism}. We use the duality theory to show (ii). First we will show that all the constraints in (P) holds. Note that $N$-agent Symmetric Worst-Case Joint Distribution is a legal joint distribution by the argument in (i). Also given the marginal distribution for agent 1, by some algebra, we have $\int_{r}^1v_1f(v_1)dv_1+Pr(v_1=1)\cdot 1=\frac{(r+N-1)(N^N-(r+N-1)^{N-1})}{(N-1)N^N}$. By the weighted virtual value representation, the value of P given $N$-agent Symmetric  Maxmin Public Good Mechanism and $N$-agent Symmetric Worst-Case Joint Distribution is simply $Pr^*(\underbrace{1,\cdots, 1}_{N})\times N=\frac{(r+N-1)^N}{N^{N-1}}$. Second, we will show that all the constraints in (D) hold for all value profiles. To see this, note for any value profile $v\in SR_N$, by some algebra, we have $(\sum_{i=1}^N{v_i})q(v)-\sum_{i=1}^N\int_{r+N-1-\sum_{j\neq i}v_j}^{v_i}q(s,v_{-i})ds=\frac{(N-1)(r+N-1)^{N-1}}{N^{N-1}-(r+N-1)^{N-1}}(\sum_{i=1}^N{v_i}-(r_2+N-1))$. Consider the following dual variables: $\lambda_i=\frac{(N-1)(r+N-1)^{N-1}}{N^{N-1}-(r+N-1)^{N-1}}$ for any $i$ and $\mu=-\frac{(N-1)(r+N-1)^{N}}{N^{N-1}-(r+N-1)^{N-1}}$. Then the complementary slackness holds by the above equation.  Since $\lambda_i>0$ for any $i$, we have $$\sum_{i=1}\lambda_i v_i+ \mu < \lambda_i (r+N-1) + \mu=0$$ for any value profile $v\notin SR_N$.  Finally, the value of (D) given the constructed $\lambda_i,\mu$ is, by some algebra,  $\sum_{i=1}^N\lambda_i m +\mu=\frac{(N-1)(r+N-1)^{N-1}}{N^{N-1}-(r+N-1)^{N-1}}(Nm-(r_2+N-1))=\frac{(r+N-1)^N}{N^{N-1}}$. By the linear programming duality theory, (ii) holds and the revenue guarantee is $\frac{(r+N-1)^N}{N^{N-1}}$.\\
\indent Finally, Let $H(r):=\frac{(r+N-1)(N^N-(r+N-1)^{N-1})}{(N-1)N^N}$. Note $H'(r)=\frac{N^{N-1}-(r+N-1)^{N-1}}{(N-1)N^{N-1}}\ge 0$ for $0\le r \le 1$, we have $m=H(r)\ge H(0)=1-\frac{(N-1)^{N-1}}{N^N}$.
\subsection{Proof of Theorem \ref{t7}}
To facilitate the analysis, we first consider a benchmark case in which there is only one agent whose expectation is known to be $m$. \cite{carrasco2018optimal} has shown that the optimal revenue guarantee $RG(m):=\max_{\tau\in[0,1]}\tau\cdot\frac{m-\tau}{1-\tau}$ and the maximizer $\tau^*:=1-\sqrt{1-m}$. In addition, $RG(m)$ is strictly increasing in $m$. \\
\indent Now we divide all deterministic, DSIC and EPIR public good mechanisms into the following four classes:\\
\textit{Class 1}: the provision boundary touches on the value profiles $(d_1,0)$ and $(0,d_2)$ for some $0\le d_1\le 1,  0\le d_2\le 1$.\\
\textit{Class 2}: the provision boundary touches on the value profiles $(d_1,0)$ and $(d_2,1)$ for some $0\le d_2\le  d_1\le 1$.\footnote{$d_2\le d_1$ is implied by the monotone property of the provision boundary (Observation \ref{ob1}).}\\
\textit{Class 3}: the provision boundary touches on the value profiles $(0,d_1)$ and $(1,d_2)$ for some $0\le d_2\le d_1 \le 1$.\\
\textit{Class 4}: the provision boundary touches on the value profiles $(d_1,1)$ and $(1,d_2)$ for some $0\le d_1\le 1,  0\le d_2\le 1$.\\
\indent The following lemmas establish a upper bound of revenue guarantee for each class of mechanisms respectively. 
\begin{lemma}\label{l6}
$RG(m_1)$ is an upper bound of the revenue guarantee for any mechanism in \textit{Class 1}.
\end{lemma}
\begin{proof}
\indent We propose a relaxation of (D) by omitting many constraints. Specifically, the only remaining constraints are for the four vertices (0,0), (1,0), (0,1) and (1,1) and the value profiles $(d_1,0)$ and $(0,d_2)$. Formally, we have the following relaxed problem  (D'): $$
\max_{\lambda_1,\lambda_2,\mu \in \mathcal{R}}\lambda_1m_1 +\lambda_2 m_2+\mu$$s.t. 
\begin{equation}\label{127}
    \mu \le 0
\end{equation}
\begin{equation}\label{128}
    \lambda_1 d_1 +\mu \le 0
\end{equation}
\begin{equation}\label{129}
 \lambda_2 d_2 +\mu \le 0
\end{equation}
\begin{equation}\label{130}
    \lambda_1 +\mu \le d_1
\end{equation}
\begin{equation} \label{131}
    \lambda_2 +\mu \le d_2
\end{equation}
\begin{equation} \label{132}
    \lambda_1 +\lambda_2 + \mu \le 0
\end{equation}
Note the value of (D') (denoted by $val(D')$) is weakly greater than the value of (D). Now we are trying to find a  upper bound of the value of (D') across $0\le d_1,d_2\le 1$.  We discuss four cases:\\
\textit{Case 1}: $\lambda_1\le 0, \lambda_2\le 0$. Note then by \eqref{127}, $val(D')\le 0$ for any $0\le d_1,d_2\le 1$.\\
\textit{Case 2}: $\lambda_1\ge 0, \lambda_2\ge 0$. Note  by \eqref{132}, we have $\lambda_1m_1+\lambda_2 m_2+\mu\le \lambda_1 +\lambda_2 + \mu \le 0$.
Thus, $val(D')\le 0$ for any $d_1,d_2$.\\
\textit{Case 3}: $\lambda_1\ge 0, \lambda_2\le 0$.   Now we are left with \eqref{128}, \eqref{130} and \eqref{132} as they imply the other three constraints. We  further ignore \eqref{132}, which will make the value of (D')  (weakly) greater.  Note at least one of \eqref{128} and  \eqref{130}  is binding, otherwise we can increase the value of (D') by increasing $\lambda_1$ by a small amount. We thus discuss two situations:\\
$(a): \lambda_1 d_1 +\mu = 0$.\\
\indent We plug $\mu=-\lambda_1 d_1$ into \eqref{130}, and we obtain $\lambda_1\le \frac{d_1}{1-d_1}$. 
Then we have $\lambda_1m_1+\lambda_2 m_2 +\mu\le \lambda_1m_1+\mu=\lambda_1(m_1-d_1)\le \max\{0,\frac{d_1(m_1-d_1)}{1-d_1}\}\le RG(m_1)$. The first inequality is implied by $\lambda_2\le 0$ and the second inequality is implied by $0\le \lambda_1 \le \frac{d_1}{1-d_1}$.\\
$(b): \lambda_1 + \mu = d_1$.\\
\indent We plug $\mu= d_1 -\lambda_1$ into \eqref{128}, and we obtain $\lambda_1\ge \frac{d_1}{1-d_1}$. 
Then we have $\lambda_1m_1+\lambda_2 m_2 +\mu\le \lambda_1m_1+\mu=\lambda_1(m_1-1)+d_1\le \frac{d_1(m_1-d_1)}{1-d_1} \le RG(m_1)$. The first inequality is implied by $\lambda_2\le 0$ and the second inequality is implied by $ \lambda_1 \ge \frac{d_1}{1-d_1}$.\\
\textit{Case 4}: $\lambda_1\le 0, \lambda_2\ge 0$. Similar to \textit{Case 3}, we obtain $\lambda_1m_1+\lambda_2 m_2 +\mu\le RG(m_2)\le RG(m_1)$ where the second inequality is implied by  the monotonicity of $RG(\cdot)$.

\end{proof}
\begin{lemma}\label{l7}
$RG(m_1)$ is an upper bound of the revenue guarantee for any mechanism in \textit{Class 3}.
\end{lemma}
\begin{proof}
\indent We propose a relaxation of (D) by omitting many constraints. Specifically, the only remaining constraints are for the four vertices (0,0), (1,0), (0,1) and (1,1) and the value profiles $(d_1,0)$ and $(d_2,1)$. Formally, we have the following relaxed problem  (D''): $$
\max_{\lambda_1,\lambda_2,\mu \in \mathcal{R}}\lambda_1m_1 +\lambda_2 m_2+\mu$$s.t. 
\begin{equation}\label{133}
    \mu \le 0
\end{equation}
\begin{equation}\label{134}
    \lambda_1 d_1 +\mu \le 0
\end{equation}
\begin{equation}\label{135}
 \lambda_1 d_2 +\lambda_2 +\mu \le 0
\end{equation}
\begin{equation}\label{136}
    \lambda_1 +\mu \le d_1
\end{equation}
\begin{equation} \label{137}
    \lambda_2 +\mu \le 0
\end{equation}
\begin{equation} \label{138}
    \lambda_1 +\lambda_2 + \mu \le d_2
\end{equation}
Note the value of (D'') (denoted by $val(D'')$) is weakly greater than the value of (D). Now we are trying to find a  upper bound of the value of (D'') across $0\le d_1,d_2\le 1$.  We discuss four cases:\\
\textit{Case 1'}: $\lambda_1\le 0, \lambda_2\le 0$. Note then by \eqref{133}, $val(D'')\le 0$ for any $0\le d_1,d_2\le 1$.\\
\textit{Case 2'}: $\lambda_1\ge 0, \lambda_2\le 0$. We further ignore \eqref{133}, \eqref{135}, \eqref{137} and \eqref{138}, which will make the value of D'' (weakly) greater. Then by similar argument with \textit{Case 3} in the proof of Lemma \ref{l6}, we obtain $\lambda_1m_1+\lambda_2m_2+\mu\le RG(m_1)$. \\
\textit{Case 3'}: $\lambda_1\le 0, \lambda_2\ge 0$. Then $\lambda_1m_1+\lambda_2m_2+\mu\le \lambda_2m_2+\mu \le 0$. The first inequality is implied by $\lambda_1\le 0$ and the second inequality is implied by \eqref{131}.\\
\textit{Case 4'}: $\lambda_1\ge 0, \lambda_2\ge 0$.   Now we are left with \eqref{134}, \eqref{135} and \eqref{138} as they imply the other three constraints.  Note at least one of \eqref{134}, \eqref{135} and \eqref{138} is binding, otherwise we can increase the value of (D'') by increasing $\lambda_1$ by a small amount. We thus discuss three situations:\\
$(a'): \lambda_1 d_1 +\mu = 0$.\\
\indent We plug $\mu=-\lambda_1 d_1$ into \eqref{135} and \eqref{138}, and we obtain
\begin{equation}\label{139}
  \lambda_2\le \lambda_1(d_1-d_2)
\end{equation} 
\begin{equation}\label{140}
    \lambda_2 \le d_2-\lambda_1(1-d_1) 
\end{equation} 
Compare the RHS of \eqref{139} and \eqref{140}. (i) When $\lambda_1\le \frac{d_2}{1-d_2}$, \eqref{139} is binding. Then $\lambda_1m_1+\lambda_2m_2+\mu=\lambda_1(m_1-d_1)+\lambda_2m_2\le \lambda_1(m_1-d_1(1-m_2)-d_2m_2)\le \max\{0, \frac{d_2}{1-d_2}(m_1-d_1(1-m_2)-d_2m_2)\}\le \max\{0, \frac{d_2}{1-d_2}(m_1-d_2) \}\le RG(m_1)$. The first inequality is implied by \eqref{139} and the third inequality is implied by $d_1\ge d_2$. (ii) When $\lambda_1\ge \frac{d_2}{1-d_2}$, \eqref{140} is binding. Also \eqref{140} and $\lambda_2\ge 0$ imply that $\lambda_1\le \frac{d_2}{1-d_1}$. Then   $\lambda_1m_1+\lambda_2m_2+\mu=\lambda_1(m_1-d_1)+\lambda_2m_2\le\lambda_1(m_1-d_1-(1-d_1)m_2)+d_2m_2\le \max\{\frac{d_2}{1-d_1}(m_1-d_1-(1-d_1)m_2)+d_2m_2, \frac{d_2}{1-d_2}(m_1-d_1-(1-d_1)m_2)+d_2m_2\}\le \max\{\frac{d_1}{1-d_1}(m_1-d_1), \frac{d_2}{1-d_2}(m_1-d_2) \}\le RG(m_1)$. The first inequality is implied by \eqref{140}, the second inequality is implied by $\frac{d_2}{1-d_2}\le\lambda_1\le \frac{d_2}{1-d_1}$, and the third inequality is implied by $d_1\ge d_2$.\\
$(b'): \lambda_1d_2+\lambda_2 + \mu = 0$.\\
\indent We plug $\mu= -\lambda_1d_2-\lambda_2$ into \eqref{134} and \eqref{138}, and we obtain 
\begin{equation}\label{141}
    \lambda_2\ge \lambda_1(d_1-d_2)
\end{equation}
\begin{equation}\label{142}
    \lambda_1\le \frac{d_2}{1-d_2}
\end{equation}
Then we have $\lambda_1m_1+\lambda_2 m_2 +\mu =\lambda_1m_1+\lambda_2m_2-\lambda_1d_2-\lambda_2\le \lambda_1(m_1-d_2-(d_1-d_2)(1-m_2)) \le RG(m_1)$. The first inequality is implied by \eqref{141} and the second inequality holds by \eqref{142} and the same argument as in (i) of $(a')$.\\
$(c'): \lambda_1+\lambda_2 + \mu = d_2$.\\
\indent We plug $\mu= d_2-\lambda_1-\lambda_2$ into \eqref{134} and \eqref{135}, and we obtain 
\begin{equation}\label{143}
    \lambda_2\ge \lambda_1(d_1-1)+d_2
\end{equation}
\begin{equation}\label{144}
    \lambda_1\ge \frac{d_2}{1-d_2}
\end{equation}
(i') When $\lambda_1 >\frac{d_2}{1-d_1}$, \eqref{143} is not binding.  Then we have $\lambda_1m_1+\lambda_2 m_2 +\mu =\lambda_1m_1+\lambda_2m_2 + d_2-\lambda_1-\lambda_2\le \lambda_1m_1 -\lambda_1 +d_2\le \max\{0, \frac{d_2}{1-d_1}(m_1-d_1)\}\le \max\{0, \frac{d_1}{1-d_1}(m_1-d_1)\} \le RG(m_1)$. The first inequality is implied by $\lambda_2 \ge 0$ and the second inequality is implied by \eqref{144} and the third inequality is implied by $d_1\ge d_2$. (ii') When $\lambda_1 \le \frac{d_2}{1-d_1}$, \eqref{143} is  binding. Then we have $\lambda_1m_1+\lambda_2 m_2 +\mu =\lambda_1m_1+\lambda_2m_2 + d_2-\lambda_1-\lambda_2\le \lambda_1(m_1-d_1-(1-d_1)m_2) +d_2m_2  \le RG(m_1)$. The first inequality is implied by \eqref{143} and the second inequality is implied by $\frac{d_2}{1-d_2}\le \lambda_1\le \frac{d_2}{1-d_1}$ and the same argument as in (ii) of $(a')$.
\end{proof}
\begin{lemma}\label{l8}
$RG(m_1)$ is an upper bound of the revenue guarantee for any mechanism in \textit{Class 3}.
\end{lemma}
\begin{proof}
By similar argument with the proof of Lemma \ref{l7},  $RG(m_2)$ is an upper bound of the revenue guarantee for any mechanism in \textit{Class 3}. Then Lemma \ref{l8} holds by $m_1\ge m_2$ and the monotonicity of $RG(\cdot)$.
\end{proof}
\begin{lemma}\label{l9}
When $m_2\ge 2(\sqrt{2}-1)$, $2(1-\sqrt{\frac{1-m_1}{2}}-\sqrt{\frac{1-m_2}{2}})^2$ is an upper bound of the revenue guarantee of any mechanism in \textit{Class 4}; otherwise $RG(m_1)$  is an upper bound of the revenue guarantee of any mechanism in \textit{Class 4}.
\end{lemma}
\begin{proof}
\indent We propose a relaxation of (D) by omitting many constraints. Specifically, the only remaining constraints are for the four vertices (0,0), (1,0), (0,1) and (1,1) and the value profiles $(d_1,1)$ and $(1,d_2)$. Formally, we have the following relaxed problem  (D3): $$
\max_{\lambda_1,\lambda_2,\mu \in \mathcal{R}}\lambda_1m_1 +\lambda_2 m_2+\mu$$s.t. 
\begin{equation}\label{145}
    \mu \le 0
\end{equation}
\begin{equation}\label{146}
    \lambda_1 d_1 +\lambda_2+\mu \le 0
\end{equation}
\begin{equation}\label{147}
 \lambda_1  +\lambda_2d_2 +\mu \le 0
\end{equation}
\begin{equation}\label{148}
    \lambda_1 +\mu \le 0
\end{equation}
\begin{equation} \label{149}
    \lambda_2 +\mu \le 0
\end{equation}
\begin{equation} \label{150}
    \lambda_1 +\lambda_2 + \mu \le d_1+d_2
\end{equation}
Note the value of (D3) (denoted by $val(D3)$) is weakly greater than the value of (D). Now we are trying to find a  upper bound of the value of (D3) across $0\le d_1,d_2\le 1$.  We discuss four cases:\\
\textit{Case 1''}: $\lambda_1\le 0, \lambda_2\le 0$. Note then by \eqref{145}, $val(D3)\le 0$ for any $0\le d_1,d_2\le 1$.\\
\textit{Case 2''}: $\lambda_1\le 0, \lambda_2\ge 0$. Then $\lambda_1m_1+\lambda_2m_2+\mu\le \lambda_2m_2+\mu\le 0$ where the second inequality is implied by \eqref{149}.\\
\textit{Case 3''}: $\lambda_1\ge 0, \lambda_2\le 0$. Then $\lambda_1m_1+\lambda_2m_2+\mu\le \lambda_1m_1+\mu\le 0$ where the second inequality is implied by \eqref{148}.\\
\textit{Case 4''}: $\lambda_1\ge 0, \lambda_2\ge 0$.  Now we are left with \eqref{146}, \eqref{147} and \eqref{150} as they imply the other three constraints.  Note at least one of \eqref{146}, \eqref{147} and \eqref{150}  is binding, otherwise we can increase the value of (D3) by increasing $\lambda_1$ by a small amount. Also note that it is without loss to restrict attention to $d_1\neq 1$ and $d_2\neq 1$, because \eqref{146} or \eqref{147} would imply $\lambda_1m_1+\lambda_2m_2+\mu\le 0$ otherwise.   We thus discuss three situations:\\
$(a''): \lambda_1  +\lambda_2d_2+\mu = 0$.\\
\indent We plug $\mu=-\lambda_1-\lambda_2d_2$ into \eqref{146} and \eqref{150}, and we obtain 
\begin{equation}\label{151}
    \lambda_2 \le \frac{d_1+d_2}{1-d_2}
\end{equation}
\begin{equation}\label{152}
    \lambda_1 \ge \frac{1-d_2}{1-d_1}\lambda_2
\end{equation}
Then $\lambda_1m_1+\lambda_2m_2+\mu=\lambda_1(m_1-1)+\lambda_2(m_2-d_2)\le \lambda_2(\frac{1-d_2}{1-d_1}(m_1-1)+m_2-d_2)\le \max\{0,(1-\frac{1-m_1}{1-d_1}-\frac{1-m_2}{1-d_2})(d_1+d_2)\}$. The first inequality is implied by \eqref{152} and the second inequality is implied by \eqref{151}.\\
$(b''): \lambda_1d_1  +\lambda_2+\mu = 0$.\\
\indent By similar argument with $(a'')$, we can show $\lambda_1m_1+\lambda_2m_2+\mu \le \max\{0,(1-\frac{1-m_1}{1-d_1}-\frac{1-m_2}{1-d_2})(d_1+d_2)\}$.\\
$(c''): \lambda_1  +\lambda_2+\mu = d_1+d_2$.\\
\indent We plug $\mu=d_1+d_2-\lambda_1-\lambda_2$ into \eqref{146} and \eqref{147}, and we obtain 
\begin{equation}\label{153}
    \lambda_1 \ge \frac{d_1+d_2}{1-d_1}
\end{equation}
\begin{equation}\label{154}
     \lambda_2 \ge \frac{d_1+d_2}{1-d_2}
\end{equation}
Then $\lambda_1m_1+\lambda_2m_2+\mu=\lambda_1(m_1-1)+\lambda_2(m_2-1)+d_1+d_2\le (1-\frac{1-m_1}{1-d_1}-\frac{1-m_2}{1-d_2})(d_1+d_2)$. The  inequality is implied by \eqref{153} and \eqref{154}.\\
\indent Let $K(d_1,d_2):=(1-\frac{1-m_1}{1-d_1}-\frac{1-m_2}{1-d_2})(d_1+d_2)$. We are now solving for $\max_{0\le d_1<1, 0\le d_2< 1}K(d_1,d_2)$. First it is without loss to assume $m_1\ge d_1$, because otherwise $1-\frac{1-m_1}{1-d_1}-\frac{1-m_2}{1-d_2}\le 0$. Now fix any $0\le d_1\le m_1$. Taking first order derivative with respect to $d_2$, we obtain that 
\begin{equation}
    \frac{\partial K(d_1,d_2)}{\partial d_2}=
    \frac{m_1-d_1}{(1-d_1)^3}[(1-d_2)^2-\frac{(1-m_2)(1+d_1)(1-d_1)}{m_1-d_1}]
\end{equation}
Let $G(d_1):= \frac{(1-m_2)(1+d_1)(1-d_1)}{m_1-d_1}$. Note $G'(d_1)=\frac{(d_1-m_1)^2+1-m_1^2}{(m_1-d_1)^2}\ge 0$. (i) If $G(d_1)\ge 1$, then $\frac{\partial K(d_1,d_2)}{\partial d_2}\le 0$ for any $d_2$. Then $K(d_1,d_2)\le K(d_1,0)=d_1(m_2-\frac{1-m_1}{1-d_1})\le d_1(1-\frac{1-m_1}{1-d_1})\le RG(m_1)$. (ii) If $G(d_1)\le 1$, let $d_2^*(d_1):=1-\sqrt{\frac{(1-m_2)(1+d_1)(1-d_1)}{m_1-d_1}}$. Note $\frac{\partial K(d_1,d_2)}{\partial d_2}\le 0$ when $d_2\ge d_2^*(d_1)$ and $\frac{\partial K(d_1,d_2)}{\partial d_2}\ge 0$ when $d_2\le d_2^*(d_1)$. Therefore $d_2^*(d_1)$ is the maximizer. With some algebra, we have 
\begin{equation}
    K(d_1,d_2^*)=(\sqrt{\frac{(m_1-d_1)(1+d_1)}{1-d_1}}-\sqrt{1-m_2})^2
\end{equation}
Let $L(d_1):= \frac{(m_1-d_1)(1+d_1)}{1-d_1}$. First note $L(d_1)=\frac{1+d_1}{G(d_1)}\sqrt{1-m_2}>\sqrt{1-m_2}$. Therefore it suffices to maximize $L(d_1)$ subject to $G(d_1)\le 1$.  Note $L'(d_1)=\frac{d_1^2-2d_1+2m_1-1}{(1-d_1)^2}$. Let $d_1^*:= 1-\sqrt{2(1-m_1)}$. Then $L'(d_1)\ge 0$ when $d_1\le d_1^*$ and $L'(d_1)\le 0$ when $ d_1^*\le d_1 \le m_1$. Note $G(d_1^*)=\sqrt{2(1-m_2)}$. Then if $m_2\le \frac{1}{2}$, then by the monotonicty of $G(\cdot)$ and $L(\cdot)$ (when $d_1\le d_1^*$), $d_1$ such that $G(d_1)=1$ is the maximizer. Then by (i), $K(d_1,d_2)\le RG(m_1)$. If $m_2> \frac{1}{2}$, then $d_1^*$ is the maximizer, and $d_2^*(d_1^*)=1-\sqrt{2(1-m_2)}$, then, with some algebra,  $K(d_1,d_2)\le 2(1-\sqrt{\frac{1-m_1}{2}}-\sqrt{\frac{1-m_2}{2}})^2$. Finally, compare $2(1-\sqrt{\frac{1-m_1}{2}}-\sqrt{\frac{1-m_2}{2}})^2$ with $RG(m_1)$ when $m_2>\frac{1}{2}$. When $m_2\ge 2(\sqrt{2}-1)$, $2(1-\sqrt{\frac{1-m_1}{2}}-\sqrt{\frac{1-m_2}{2}})^2\ge RG(m_1)$; otherwise $2(1-\sqrt{\frac{1-m_1}{2}}-\sqrt{\frac{1-m_2}{2}})^2 < RG(m_1)$. 
\end{proof}
To prove Theorem \ref{t7}, it suffices to show that the upper bounds identified in Lemma \ref{l6}, Lemma \ref{l7} and Lemma \ref{l8} are attainable. Note this is obvious when $m_2\le 2(\sqrt{2}-1)$, as we can just ignore agent 2. When  $m_2\ge 2(\sqrt{2}-1)$, by the proof of Lemma \ref{l8}, consider the dual variables $\lambda_1=\frac{d_1^*+d_2^*(d_1^*)}{1-d_1^*},\lambda_2=\frac{d_1^*+d_2^*(d_1^*)}{1-d_2^*(d_1^*)}$ and $\mu=-\frac{(d_1^*+d_2^*(d_1^*))(1-d_1^*d_2^*(d_1^*))}{(1-d_1^*)(1-d_2^*(d_1^*))}$. Note $\lambda_1m_1+\lambda_2m_2+\mu=2(1-\sqrt{\frac{1-m_1}{2}}-\sqrt{\frac{1-m_2}{2}})^2$. Then it suffices to show that the constructed dual variables satisfy  the constraints of (D) for any mechanism in the Theorem \ref{t7}. Note by the above argument, $\lambda_1 v_1+\lambda_2 v_2+\mu=0$ for  the value profiles $(d_1^*,1)$ and $(1,d_2^*(d_1^*))$. Then by linearity, any value profile in the line boundary satisfies $\lambda_1 v_1+\lambda_2 v_2+\mu=0$. Because $\lambda_1>0$ and $\lambda_2>0$,  we have $\lambda_1 v_1+\lambda_2 v_2+\mu< 0$ for any value profile below the provision boundary $\mathcal{B}$. Therefore the the constraints of (D) hold for any value profile below the provision boundary $\mathcal{B}$. Now consider any value profile $(v_1,v_2)$ above the provision boundary $\mathcal{B}$. Then $\lambda_1v_1+\lambda_2v_2+\mu-t(v)=\lambda_1v_1+\lambda_2v_2+\mu-b_1(v_1)-b_2(v_2)\le \lambda_1+\lambda_2+\mu-d_1^*-d_2^*(d_1^*)=0$ where $(v_1,b_2(v_1))$ and $(b_1(v_2),v_2)$ lie in the provision boundary $\mathcal{B}$. The first equality holds by $t(v)=b_2(v_1)+b_1(v_2)$, the inequality holds because $\lambda_1>0, \lambda_2>0$ and $b_1, b_2$ are non-increasing, and the last equality holds by our construction. 

\subsection{Proof of Theorem \ref{t9}}
(i): \textbf{N-agent Maxmin Excludable Public Good Mechanism} is a best response to \textbf{N-agent Independent Equal Revenue Distribution}. Note under N-agent Independent Equal Revenue Distribution,  $\Phi_i^E(v)=0$ if $\gamma_i\le v_i<1$ for any $i$ and  $\Phi_i^E(v)>0$ if $v_i=1$ and $v$ is in the support for any $i$. Then any feasible and monotone mechanism in which the public good is provided to agent $i$  with some positive probability if and only if $v_i > \gamma_i$ and the public good is provided with probability 1 to agent $i$ when $v_i=1$ is a best response for the principal. It is easy to see that N-agent  Maxmin Excludable Public Good Mechanism is such a mechanism.\\
(ii): \textbf{N-agent Independent Equal Revenue Distribution} is a best response to \textbf{N-agent Maxmin Excludable Public Good Mechanism}. We use the duality theory to show (ii). The primal and the dual are simple adaptions of those in the proof of Lemma \ref{l1}.   First note that N-agent Independent Equal Revenue Distribution is a legal joint distribution. And given the definition of $\gamma_i$, it satisfies all of the mean constraints. By the weighted virtual value representation, the value of the primal given N-agent Independent Equal Revenue Distribution and N-agent Maxmin Excludable Public Good Mechanism is simply $\sum_{i=1}^N\gamma_i$. Second, the constraints in the dual hold for all value profiles. To see this, we first construct the dual variables as follows: $\lambda_i=-\frac{1}{\ln{\gamma_i}}$ for any $i$ and $\mu=\sum_{i=1}^N\frac{\gamma_i}{\ln{\gamma_i}}$. Note that $t_i^{E*}(v)=-\frac{v_i-\gamma_i}{\ln{\gamma_i}}$ if $v_i\ge \gamma_i$ and $t_i^{E*}(v)=0$ if $v_i< \gamma_i$  by Proposition \ref{p1'}. Then for any value profile inside the support of N-agent Independent Equal Revenue Distribution, the constraints of the dual hold with equality (complementary slackness). For any value profile outside the support of N-agent Independent Equal Revenue Distribution, the constraints also hold because $\lambda_i v_i+\frac{\gamma_i}{\ln{\gamma_i}}< 0$ if $v_i<\gamma_i$.   Finally, the value of the dual given the constructed dual variables is $\sum_{i=1}^N\lambda_i m_i +\mu$, which, by some algebra, is equal to $\sum_{i=1}^N\gamma_i$. By the linear programming duality theory, (ii) holds and the revenue guarantee is $\sum_{i=1}^N\gamma_i$.

\newpage

\bibliographystyle{apalike}
\bibliography{abc}

\begin{thebibliography}{}

\bibitem[Anderson and Nash, 1987]{anderson1987linear}
Anderson, E.~J. and Nash, P. (1987).
\newblock {\em Linear programming in infinite-dimensional spaces: theory and
  applications}.
\newblock John Wiley \& Sons.

\bibitem[Bergemann et~al., 2017]{bergemann2017first}
Bergemann, D., Brooks, B., and Morris, S. (2017).
\newblock First-price auctions with general information structures:
  Implications for bidding and revenue.
\newblock {\em Econometrica}, 85(1):107--143.

\bibitem[Bergemann et~al., 2019]{bergemann2019revenue}
Bergemann, D., Brooks, B., and Morris, S. (2019).
\newblock Revenue guarantee equivalence.
\newblock {\em American Economic Review}, 109(5):1911--29.

\bibitem[Bergemann et~al., 2016]{bergemann2016informationally}
Bergemann, D., Brooks, B.~A., and Morris, S. (2016).
\newblock Informationally robust optimal auction design.

\bibitem[Bergemann and Morris, 2005]{bergemann2005robust}
Bergemann, D. and Morris, S. (2005).
\newblock Robust mechanism design.
\newblock {\em Econometrica}, 73(6):1771--1813.

\bibitem[Brooks and Du, 2021]{brooks2021optimal}
Brooks, B. and Du, S. (2021).
\newblock Optimal auction design with common values: An informationally robust
  approach.
\newblock {\em Econometrica}, 89(3):1313--1360.

\bibitem[Carrasco et~al., 2018]{carrasco2018optimal}
Carrasco, V., Luz, V.~F., Kos, N., Messner, M., Monteiro, P., and Moreira, H.
  (2018).
\newblock Optimal selling mechanisms under moment conditions.
\newblock {\em Journal of Economic Theory}, 177:245--279.

\bibitem[Carroll, 2017]{carroll2017robustness}
Carroll, G. (2017).
\newblock Robustness and separation in multidimensional screening.
\newblock {\em Econometrica}, 85(2):453--488.

\bibitem[Che, 2020]{che2020distributionally}
Che, E. (2020).
\newblock Distributionally robust optimal auction design under mean
  constraints.

\bibitem[Chen and Li, 2018]{chen2018revisiting}
Chen, Y.-C. and Li, J. (2018).
\newblock Revisiting the foundations of dominant-strategy mechanisms.
\newblock {\em Journal of Economic Theory}, 178:294--317.

\bibitem[Chung and Ely, 2007]{chung2007foundations}
Chung, K.-S. and Ely, J.~C. (2007).
\newblock Foundations of dominant-strategy mechanisms.
\newblock {\em The Review of Economic Studies}, 74(2):447--476.

\bibitem[Du, 2018]{du2018robust}
Du, S. (2018).
\newblock Robust mechanisms under common valuation.
\newblock {\em Econometrica}, 86(5):1569--1588.

\bibitem[G{\"u}th and Hellwig, 1986]{guth1986private}
G{\"u}th, W. and Hellwig, M. (1986).
\newblock The private supply of a public good.
\newblock {\em Journal of Economics}, 46(1):121--159.

\bibitem[He and Li, 2020]{he2020correlation}
He, W. and Li, J. (2020).
\newblock Correlation-robust auction design.

\bibitem[Ko{\c{c}}yi{\u{g}}it et~al., 2020]{koccyiugit2020distributionally}
Ko{\c{c}}yi{\u{g}}it, {\c{C}}., Iyengar, G., Kuhn, D., and Wiesemann, W.
  (2020).
\newblock Distributionally robust mechanism design.
\newblock {\em Management Science}, 66(1):159--189.

\bibitem[Libgober and Mu, 2021]{libgober2021informational}
Libgober, J. and Mu, X. (2021).
\newblock Informational robustness in intertemporal pricing.
\newblock {\em The Review of Economic Studies}, 88(3):1224--1252.

\bibitem[Myerson, 1981]{myerson1981optimal}
Myerson, R.~B. (1981).
\newblock Optimal auction design.
\newblock {\em Mathematics of operations research}, 6(1):58--73.

\bibitem[Suzdaltsev, 2020]{suzdaltsev2020optimal}
Suzdaltsev, A. (2020).
\newblock An optimal distributionally robust auction.
\newblock {\em arXiv preprint arXiv:2006.05192}.

\bibitem[Yamashita and Zhu, 2018]{yamashita2018foundations}
Yamashita, T. and Zhu, S. (2018).
\newblock On the foundations of ex post incentive compatible mechanisms.

\bibitem[Zhang, 2021a]{zhang2021correlation}
Zhang, W. (2021a).
\newblock Correlation robustly optimal auctions.
\newblock {\em arXiv preprint arXiv:2105.04697}.

\bibitem[Zhang, 2021b]{zhang2021robust}
Zhang, W. (2021b).
\newblock Robust bilateral trade mechanisms with known expectations.
\newblock {\em arXiv preprint arXiv:2105.05427}.

\end{thebibliography}

\end{document}